%% file: Uniform_DET_fullpage.tex
\documentclass[12pt,a4paper,reqno]{article}
\usepackage{fullpage}
\usepackage{xifthen}

\usepackage{amssymb,latexsym}
\usepackage{graphicx,amsmath,amssymb}
\usepackage{amsthm}
\usepackage{microtype}
\usepackage{enumitem}
\usepackage{pstricks}

\usepackage{crimson}
\usepackage{setspace}
\usepackage{color}
\let\G\undefined
\usepackage{phuong}
\usepackage{comment}
\newcommand{\lrg}{off}

\ifthenelse{\equal{\lrg}{on}}
{

\newcommand{\makelrg}{\Large}

\usepackage[vmargin=.4in,hmargin=0.35in]{geometry}
} 
{
\newcommand{\makelrg}{}

}

\newcommand{\cmnt}{off}

\ifthenelse{\equal{\cmnt}{on}}{
\excludecomment{strongcomment}
\includecomment{comment}
\newcommand{\iddo}[1]{\footnote{\textcolor[rgb]{0.5,0,0.9}{\uline{Note}: #1}}}
\newcommand {\mar}[1]
{
              \marginpar{
                       \begin{minipage}{60pt}
                             \tiny \color{lightgray}#1
                       \end{minipage}
                 }
}
}
{\newcommand{\iddo}[1]{}
\excludecomment{comment}
\excludecomment{strongcomment}
\newcommand{\mar}[1]{}
}

\usepackage[normalem]{ulem}
\usepackage{pifont}

\newcommand{\DET}{\textsf{DET}}

\date{}

\newcommand{\HRule}{\rule{\linewidth}{0.2mm}}

\renewcommand{\GFTwo}{\ensuremath{GF(2)}}
\newcommand{\dd}{\ensuremath{\delta}}

\definecolor{lightgray}{rgb}{.7,.7,.7}
\definecolor{gray}{rgb}{.2,.2,.2}
\definecolor{verydarkblue}{rgb}{.0,.0,.2}
\definecolor{darkerblue}{rgb}{0,0,.4}
\definecolor{mycyan}{rgb}{0.0117647,0.533333,0.647059}

\newtheorem{theorem}{Theorem}[section]
\newtheorem*{theorem*}{Theorem}
\newtheorem{lemma}[theorem]{Lemma}

\newtheorem*{claim*}{Claim}
\newtheorem{definition}[theorem]{Definition}
\newtheorem{proposition}[theorem]{Proposition}
\newtheorem{corollary}[theorem]{Corollary}
\newtheorem{remark}[theorem]{Remark}
\newtheorem{fact}[theorem]{Fact}

\newcommand\F{{\ensuremath{\mathbb F}}}
\renewcommand\G{{\ensuremath{\mathbb G}}}
\newcommand\PP{{\mathbb P}}

\newcommand\PC{\ensuremath{\PP_c}}
\newcommand\PI{\ensuremath{\PP_{c}^{{-1}}}}

\newcommand\cplus{{\oplus}}
\newcommand\ctimes{{\otimes}}
\newcommand{\tr}[1]{\ensuremath{{[{#1}]}}}

\newcommand\Num{{\hbox{Num} }}
\newcommand\Den {{\hbox{Den} }}
\newcommand\poly{\hbox{\rm poly}}
\newcommand\Det{\hbox{\rm DET}}

\newcommand\Coef{{\rm coeff}}
\newcommand\nbyn{\ensuremath{n\times n}}
\newcommand\dbyd{\ensuremath{d\times d}}
\newcommand{\set}[1]{\{#1\}}
\newcommand {\cd}{\cdot}
\newcommand{\Base}{\mbox{}\\ \ind{\textit{Base case: }}}
\newcommand{\Induction}{\mbox{}\\ \ind{\textit{Induction step: }}}
\newcommand{\case}[1]{\ind\textbf{Case #1}:\,}
\newcommand{\induction}{\Induction}

\newcommand{\TODO}[1]{\mar{\texttt{\textcolor[rgb]{0.396078,0.639216,0.490196}{:#1}}}}
\renewcommand{\TODO}[1]{}

\newcommand {\ind} {\noindent}

\newcommand {\para}[1] {\paragraph{#1.}}

\DeclareMathAlphabet{\mathitbf}{OML}{cmm}{b}{it}
\newcommand{\rst}{\upharpoonright}

\ifx\theorem\undefined

\newtheorem{theorem}{Theorem}[section]
\newtheorem{lemma}[theorem]{Lemma}
\newtheorem{proposition}[theorem]{Proposition}
\newtheorem{corollary}[theorem]{Corollary}
\newtheorem{definition}{Definition}[section]

\newtheorem{mycomment}{Comment}
\newtheorem{fact}[theorem]{Fact}

\fi

\newtheorem{claim}[theorem]{Claim}
\newenvironment{notation}{\QuadSpace\par\noindent{\bf Notation}:}{\HalfSpace}
\newenvironment{note}{\QuadSpace\par\noindent{\bf Note}:}{\HalfSpace}

\newenvironment{mycomment}{\QuadSpace\par\noindent{\bf Comment}:}{\HalfSpace}

\newenvironment{example}{\QuadSpace\par\noindent{\bf Example}:}{\HalfSpace}

\ifx\proof\undefined
\newenvironment{proof}{\QuadSpace\par\noindent{\bf Proof}:}{\EndProof\HalfSpace}
\fi

\newenvironment{proofclaim}{\QuadSpace\par\noindent{\it Proof of claim}:}
{\vrule width 1ex height 1ex depth 0pt $_{\textrm{\,Claim}}$ \HalfSpace}

\newcommand{\QuadSpace}{\vspace{0.25\baselineskip}}
\newcommand{\HalfSpace}{\vspace{0.5\baselineskip}}
\newcommand{\FullSpace}{\vspace{1.0\baselineskip}}
\newcommand{\EndProof}{ \hfill \vrule width 1ex height 1ex depth 0pt }

\newcommand{\NCTwo}{\ensuremath{\mathbf{NC^2}}}
\renewcommand{\NCOne}{\ensuremath{\mathbf{NC^1}}}

\newcommand{\Z}{\ensuremath{\mathbb{Z}}}
\newcommand{\N}{\ensuremath{\mathbb{N}}}
\newcommand{\abs}[1]{|#1|}

\newcommand{\eac}{\ensuremath{{\textsf{Eval}_{al}}}} 
\newcommand{\mbeac}{\ensuremath{{\textsf{Eval}_{lmb}}}} \newcommand{\encd}[1]{\ensuremath{{\sf write}_{#1} }}
\newcommand{\cdeti}{\ensuremath{\mathsf {Det}_{circ^{-1}}}} 
\newcommand{\cdet}{\ensuremath{\mathsf {Det}_{circ}}} \newcommand{\taydet}{\ensuremath{\mathsf{Det}_{Taylor}}}
\newcommand{\taydetsharp}{\ensuremath{\mathsf{Det}_{Taylor}^{\#}}}
\newcommand{\taydetsharpprime}{\ensuremath{\mathsf{Det}_{Taylor}^{\star}}}
\newcommand{\cdetbl}{\ensuremath{\mathsf {Det}_{balanced}}} 
\newcommand{\cdetb}{\ensuremath{\mathsf{Det}_{bool}}}
\newcommand{\PCZ}{\ensuremath{\PC(\Z)}}

\newcommand{\writepci}{\ensuremath{{\encd{\PI(\Z)\vdash \text{\eqref{eq:1},\eqref{eq:2}}}(n,W)}}} \newcommand{\lgsq}{\ensuremath{\log^2 n}}
\newcommand{\olgsq}{\ensuremath{O(\lgsq)}}
\newcommand{\inv}{{\rm Inv}}
\newcommand{\hci}[1]{\ensuremath{{#1}^{(i)}}}
\newcommand{\hc}[2]{\ensuremath{{#1}^{(#2)}}}
\newcommand{\hck}[1]{\ensuremath{{#1}^{(k)}}}

\newcommand{\PIZ}{\ensuremath{\PI(\Z)}}
\renewcommand{\PCZ}{\ensuremath{\PC(\Z)}}
\newcommand{\hcz}[1]{\ensuremath{#1^{(0)}}}
\newcommand{\eval}[1]{\ensuremath{\mathrm{eval}\left(#1\right)}}
\newcommand{\matz}{\ensuremath{\mathrm {Mat}_\Z}}
\newcommand{\trianmatz}{\ensuremath{\mathrm {triangMat}_\Z}}
\renewcommand{\deg}{\ensuremath{d}}

\newcommand{\ix}{\ensuremath{X^{-1}}} 
 
\newcommand{\outix}{\ensuremath{\encd{\ix}}}

\newcommand{\outixd}
    {\ensuremath{\encd{\text{level}(\ix)}}}
\newcommand{\degub}{\ensuremath{d_{\rm{ub}}}}
\newcommand{\degubp}{\ensuremath{d^+_{\rm ub}}}   
\renewcommand{\l}{\ell}

\input{bibmacros}

\author{Iddo Tzameret\thanks{Department of Computer Science, Royal Holloway, University of London. Iddo.Tzameret@rhul.ac.uk} \and Stephen A.~Cook\thanks{Department of Computer Science, University of Toronto. sacook@cs.toronto.edu}}

\newlength{\defbaselineskip}
\renewcommand{\singlespacing}{\setlength{\baselineskip}{1.97\defbaselineskip}}

\ifthenelse{\equal{\lrg}{off}}
{
}
{}

\usepackage{wasysym}

\begin{document}

\title{Uniform, Integral and Feasible Proofs for the Determinant Identities}
\maketitle

\begin{abstract}
Aiming to provide  weak as possible axiomatic assumptions in which one can develop basic linear algebra, we give a uniform and integral version of the short propositional proofs for the determinant identities demonstrated over $GF(2)$ in Hrube\v s-Tzameret \cite{HT12}. Specifically, we show that the multiplicativity of the  determinant function and the Cayley-Hamilton theorem over the integers
are provable in the bounded arithmetic theory \VNCTwo; the latter is a first-order theory corresponding to the complexity class \NCTwo\ consisting of problems solvable by uniform families of polynomial-size circuits and $O(\log ^2 n)$-depth. This also establishes the existence of uniform polynomial-size \NCTwo-Frege proofs  of the basic determinant identities over the integers (previous propositional proofs hold only over the two element field).

%
%
%
\end{abstract}

\singlespacing
{\small
\tableofcontents
}
\normalsize
\makelrg

\section{Introduction}
The complexity of linear algebraic operations such as matrix inverse and the determinant is well studied (cf.~Cook \cite{Coo85}). It is well known that many linear algebraic operations like the determinant  can be computed quickly in parallel, and specifically are in \NCTwo, which is the complexity class consisting of all languages that can be decided by uniform families of $O(\log ^2 n)$-depth and polynomial-size circuits (ignoring for now the distinction between function and language classes). This complexity class captures fast parallel computation in the sense  that a language in it can be decided in time $O(\log ^2 n)$ while using polynomially many processors working in parallel. In fact, within the $\NC:=\cup_{i=0}^\infty\NC^i$ hierarchy, which consists of all polynomial-size circuit families of poly-logarithmic depth,  \NCTwo\ is the weakest level known to compute the determinant (formally, the weakest circuit class computing integer determinants is the class \DET\ that lies between \NCOne\ and \NCTwo; see below).

Furthermore, the importance of linear algebra in bounded arithmetic and proof complexity has been identified in many works,  and it has been  conjectured  that the determinant identities, and specifically the multiplicativity of the determinant function $\Det(A)\cd\Det(B)=\Det(AB)$, for two matrices $A,B$, can be proved in a formal theory that, loosely speaking, reasons with \NCTwo\ concepts (Cook and Nguyen present this specific question in their monograph \cite{CN10}; see also \cite{CF10,BBP95,BP98,Sol_PhD,SC04}). 
This conjecture  
is aligned with the intuition that basic properties of many constructions and functions of a given complexity class are provable in logical theories not using concepts beyond that class.

The weakest theory known to date to prove the determinant identities is  \PV\ which corresponds to polynomial-time reasoning; this was shown  by Soltys and Cook \cite{SC04} (cf.~\cite{CF10,Jer05-PhD}).
Quite recently, Hrube\v s and Tzameret \cite{HT12} showed that at least in the \emph{propositional} case, the determinant identities expressing the multiplicativity of the determinant over $GF(2)$ can be proved with polynomial-size propositional proofs operating with \NCTwo-circuits (as well as with  quasipolynomial size Frege proofs). However, this does not lend itself immediately to the uniform framework of bounded arithmetic. That is, the fact that a statement admits polynomial-size propositional proofs in a certain proof-system does not imply that the same statement (suitably translated to first-order logic) is provable in the bounded arithmetic theory corresponding to the proof-system. For example, a short propositional proof may be shown to exist, but without knowing whether it could be constructed uniformly, and  let alone in a restricted computational model such as uniform-\NCTwo---making it thus impossible to carry out directly in bounded arithmetic. 

Furthermore, \cite{HT12}  crucially used in their construction elimination of division gates from algebraic circuits, which we do not know how to do  using uniform weak computational models like  uniform-\NCTwo\ (since for general  division elimination one needs to use the existential statement about  field assignments that do not nullify a given polynomial  \cite{Str73}). \smallskip

The main goal of this work is to prove the determinant identities in the theory \VNCTwo\ (corresponding to ``\NCTwo-reasoning''). 
We will show that similar reasoning as in \cite{HT12} can be carried over to \VNCTwo, with further complications imposed by uniformity and parallelism. As a result of working in bounded arithmetic it will also become possible to conclude  short  propositional-proofs  over the integers (while the previous propositional proofs  worked only over \GFTwo).

\para{Organization} The preliminaries for this work are somewhat  long. \textit{For this reason we begin with a \emph{high-level} overview of the results and their proofs in Section \ref{sec:overview}} (readers who are unfamiliar with some of the concepts in the overview can consult the preliminaries section for those). The preliminaries themselves are given in Section \ref{sec:prelim}, consisting of basic definitions from bounded arithmetic, the uniform complexity class \NCTwo, the corresponding theory \VNCTwo~\cite{CN10}, basic definitions of algebraic circuits, as well as proof systems operating with algebraic circuits  establishing polynomial identities (PI-proofs \cite{HT08,HT12}).
In Section \ref{sec:carry-proof-in-theory}
we give a much more detailed guide to the proof of the determinant identities in the theory, while still leaving out many of  the technical details and proofs. 
%
%
 Section \ref{sec:Encoding Circuits and PI-Proofs} explains in some detail how we encode certain algebraic circuits in the theory. Sections \ref{sec:From-Rational-Function-to-the-Determinant-as-a-Polynomial} to \ref{sec:balancing-algebraic-circ-in-theory} are dedicated to the construction in uniform \NCTwo\ the PI-proof from \cite{HT12}. Section \ref{sec:applying-the-reflection} wraps-up the proof by establishing the reflection principle for Polynomial Identity (PI) proofs, and Section \ref{sec:corollaries} provides \VNCTwo\ proofs of further basic statement in linear algebra. We finish with conclusions and open problems in Section \ref{sec:Conclusions and Open Problems}. The appendix provides more background details about bounded arithmetic as well as some technical lemmas that do not appear in the main text.

\section{Overview }\label{sec:overview}
Our goal is to prove the determinant identities
inside \VNCTwo. For the logical setting and \VNCTwo\ see Section~\ref{sec:prelim}. Specifically, we want to have a \SigOneB-definable in \VNCTwo\  function $\Det(\cd)$ with input an integer  matrix and output an integer represented as a binary string, such that
\VNCTwo\ proves:\iddo{}
\begin{equation}\label{eq:one}
\forall A, B ~\text{\nbyn\ matrices over \Z},~\Det(A)\cd \Det(B)=\Det(AB)
\end{equation}
and
\begin{equation}\label{eq:two}
\forall C~\text{\nbyn\ triangular matrix over \Z},~ ~\Det(C)=c_{11}\cdots c_{nn}.
\end{equation}

Note that these two identities can be considered as the \emph{defining} identities of the determinant polynomial, in the sense that every polynomial for which these two identities hold is the determinant polynomial. One way of seeing this is to observe that every square matrix is equal to a product of upper and lower triangular matrices.
\iddo{} 


Integer numbers are represented as binary strings in the theory, where the least significant bit (lsb) is 0 (resp.~1) when the  integer is positive (resp.~negative), and where the rest of the string is the binary representation of the absolute value of the integer.
An \nbyn\ matrix over \Z\ is  encoded as a two-dimensional array (cf.~\cite{CF10}). 

It is not hard to show that we can prove \emph{simple }facts about matrices, such as the definability of matrix product $AB, $ the statement expressing associativity and commutativity of matrix products $ A(BC)=(AB)C $ and \(A+B=B+A\), resp., and so forth (see for example \cite{SC04,CN10} and \cite[Lemma 28]{HT12} about these basic identities that can be proved already in the theory \VNCOne, that corresponds to \NCOne).

All circuit classes discussed in this work (except when otherwise stated)  are assumed to be uniform circuit classes. Formally, we require uniformity in the sense that  the extended connection language of the circuit family is in \FO\ (see \cite[Chapter A.5]{CN10} for the definition).

Let us now sketch briefly how we define the determinant function in the theory and then how we \emph{prove} its identities \eqref{eq:1} and \eqref{eq:2} in the theory.  

\para{Defining the determinant function in the theory} Given an \nbyn\ integer  matrix, the \SigOneB-definable  string function (recall that we encode integers as strings)  in \VNCTwo\ for the determinant is defined roughly as follows: first, construct an $O(\log ^2 n)$-depth algebraic circuit computing the determinant of \nbyn\ integer matrices, and then evaluate the circuit under the input assignment.

More specifically, the determinant function in the theory  first constructs a  recursive algebraic circuit (or equivalently, a straight-line program) computing the symbolic \nbyn\ determinant \textit{with division gates} (``symbolic'' here means that the algebraic circuit computes the determinant as the formal polynomial over  $n^2$ distinct variables). This is  done using the standard recursive formula of the determinant (using ``Schur complement''), simulating in a sense Gaussian elimination (cf.~\cite{HT12}). Then, \emph{eliminate the division gates }in the determinant circuit using, among other conversions,  substitutions of power series  in the circuit. Then,  \emph{homogenize} the circuit getting rid of high degrees, \emph{balance the circuit} to achieve the squared logarithmic depth, and  finally \emph{evaluate }the result under the input integer matrix.

The function that evaluates a balanced algebraic circuit in itself consists of several steps, as follows: given as an input a balanced algebraic circuit, the function: (i) converts it into a layered circuit (namely, a circuit in which each node connects only to the subsequent layer); (ii)  transforms it into a \emph{Boolean circuit }computing the same polynomial over the integers (coded as bit-strings) while taking care that the negations appear only in the bottom layer;  and finally (iii) evaluates the Boolean circuit using the fact that the \textsc{Monotone \NCTwo\ Circuit Evaluation Problem}  is \NCTwo-complete (under \ACZ-reductions \cite{CN10}).
\smallskip

Note that since we show that the determinant function as defined above is $\SigOneB$-definable in \VNCTwo, by \cite{CN10} it means that this function is in uniform-\NCTwo. 

\para{Proving the determinant equalities in the theory} Informally, the basic argument formalized in the theory
is that there exists a balanced  PI-proof (for \textit{Polynomial Identity  proof}), in symbols, a   $\PCZ$-proof (as in \cite{HT12}; see Section \ref{sec:PI-proofs}), of these identities. 
Thus, by soundness of balanced \PCZ-proofs, which  we show is provable in \VNCTwo, these identities must  be true. Informally, a \PCZ-proof is a sequence of equations between algebraic circuits over \Z, each of which is either an instance of the polynomial-ring axioms or was derived by addition or multiplication of previous equations. 

More precisely, 
we  demonstrate a \SigOneB-definable function in \VNCTwo\ that given an input $n$ in unary, outputs  \PCZ-proofs of the determinant  identities  (see equations (\ref{eq:1}) and (\ref{eq:2})). In this \PCZ-proof every proof-line is an equation between depth $O(\log^2(n))$ algebraic circuits (without division gates) of a polynomial syntactic-degree. To conclude the argument, we use the soundness of $O(\log^2(n))$-depth $\PCZ$-proofs: using induction on proof-length we argue that for \textit{every assignment of integers}, the determinant identity (equations (\ref{eq:1}) and (\ref{eq:2})) must hold.

One important observation in this work, that is central in constructing the PI-proofs in the theory, is that  for some parts in the construction the only properties that the theory is required to express and prove about these PI-proofs are ``local and syntactic" properties, namely the fact that the each proof-line follows syntactically from previous ones. 

For example, let $C$ be an algebraic circuit of polynomial-size and exponential syntactic-degree; e.g., $(x^2)^{2\cdots 2}-(x^2)^{2\cdots 2}+1$, where $(x^2)^{2\cdots 2}$ is written as a chain of $n$ product gates. The theory cannot express the fact that $C$ has exponential syntactic-degree (because the theory defines only polynomially bounded number functions). Nevertheless, the theory can prove, for example, that $(x^2)^{2\cdots 2}-(x^2)^{2\cdots 2}+1=1$ has a legal PI-proof, using possibly an axiom of the form $F-F=0$.  


Overall, in our argument, the main ``non-syntactic" property we need the theory to express about algebraic circuits is the evaluation of \olgsq-depth circuits over \Z. The axiom of \VNCTwo\ is specifically tailored for this purpose (see Section~\ref{sec:algebraic_circuit_value_problem}). We also use the ability to power matrices in \NCTwo\ when balancing the PI-proofs in the theory. 
\iddo{}  
\subsection{Technical Challenges}

 Showing that the  long  and nontrivial constructions from \cite{HT12}  can be carried out  in \VNCTwo\ requires quite  a lot of work. The main technical obstacles that we face are \emph{parallelism} and \emph{uniformity} as we explain in what follows. \smallskip

\emph{\uline{Parallelism}} here means that the construction of the original PI-proofs from \cite{HT12}  \textit{must be done by itself in \NCTwo}. The construction in \cite{HT12} is quite involved, and to make it parallel we need to devise several \ACZ- and  \NCTwo-algorithms (all \SigOneB-definable in \VZ\ and \VNCTwo, respectively). In fact we show that  most parts of the construction can be carried out already in  \ACZ\ (or its functional version \FACZ), namely we carry out the construction in \VZ.
Among the  algorithms we devise are the following ones:

\TODO{Check!}

(\textbf{i}) Division normalization: converting algebraic circuits with division gates into circuits with a single division gate at the output gate (in \FACZ); This follow Strassen's algorithm \cite{Str73}. 
(\textbf{ii}) Converting algebraic circuits $C$ into the sum of their syntactic-homogeneous components, given as input an upper bound on the syntactic-degree of $C$; i.e., each summand $C^{(i)}$ is a syntactic-homogeneous circuit computing the degree $i$ homogeneous component of $C$ (in \FACZ); This also follows Strassen's algorithm \cite{Str73}, only that we show that for most purposes there is no need to compute syntactic-degrees of nodes, rather \emph{upper bounds} on syntactic-degrees suffice. Such upper bounds are easy to compute in \ACZ.
(\textbf{iii}) An \FNCTwo\ algorithm for balancing an algebraic circuit of size $s$ and syntactic-degree  $d$ into a $\poly(s,d)$-size algebraic circuit of depth $O(\log s\cd \log d+\log ^2d)$, given as input an upper bound on the syntactic-degree of $C$. This part combines the original balancing algorithm by Valiant \textit{et al.}~\cite{VSB+83} with ideas from Miller \textit{at al}.~\cite{MRK88}, and further new ideas entailed by the need to work in \FNCTwo. Specifically, we use matrix powering to power adjacency matrices of graphs to find out, for example, whether a node has a directed path to another node, as well as to compute coefficients of linear polynomials computed by circuits with syntactic-degree 1.  

By first balancing an input circuit and then evaluating it (both in \FNCTwo) our results give rise to: (\textbf{iv}) an F\NCTwo\ \emph{evaluation procedure for algebraic circuits of any depth} (given as input an upper bound on their syntactic-degree and assuming the syntactic degree of the circuit is polynomial\footnote{Formally, we need to assume that the syntactic-degree of every node in the circuit when constant nodes are replaces by corresponding variables is polynomially bounded.}) that is different from the previously known  algorithm by Miller \textit{et al.}~\cite{MRK88} (their algorithm does not require the syntactic-degree as input) and that of Allender \emph{et al.} \cite{AJMV98} (which is implicit in that work, and can be extracted from the text \cite{All18}; see also Vinay \cite{Vin91}). 



\smallskip 

Proving parallel algorithms for structural results on algebraic circuits is however not enough. We further need to show that the correctness of these algorithms can be formalized efficiently with PI-proofs and that these proofs are constructible in \VZ\ and \VNCTwo, in order to conclude that \VNCTwo\ proves the existence of a (uniform \NCTwo) function that constructs the low depth PI-proofs of the determinant identities.      

\smallskip  

\emph{\uline{Uniformity}} here means that we need the whole proof to be constructible in uniform-\NCTwo. For instance, we need to eliminate division gates from certain algebraic circuits and proofs. To eliminate division gates like $u/v$ (for two nodes \(u,v\)), one needs to find an assignment to the variables in which the polynomial computed at node $v$ is nonzero. In general we do \emph{not }know how to do this in the theory. Nevertheless, we show that for our purposes it is enough to eliminate only those division gates that occur in some specific circuits.
In order to eliminate division gates we will also need to find `inverse elements' in the ring of integers, and hence we will have to show that for our purposes it is enough to consider only the inverse of 1 in \Z. 
\medskip 

Apart from uniformity and parallelism, working in bounded arithmetic allows us to work more easily over the integers, where previously short \NCTwo-Frege proofs of the determinant identities were known only over $GF(2)$ \cite{HT12}.

\subsection{Note on the Choice of Theory}
\label{sec:Note-on-the-choice-of-Theory}
 It is interesting to consider whether the theory in which the determinant identities is  proved can be pushed even further down  to a theory that corresponds to a complexity class that lies somewhere between \NCOne\ and \NCTwo. 

Cook and Fontes \cite{CF10} developed a bounded arithmetic theory $V\# L$, corresponding to $\DET$, where \DET\ is the  class of functions that can be computed by uniform families of polynomial-size  constant-depth Boolean circuits with oracle access to the determinant over \Z\ (where integer entries of matrices are presented in binary). In other words, \DET\ is the \ACZ-closure of integer determinants. Complete problems for the class \DET\ include computing matrix powers and the determinant itself. We have the following class inclusions (we ignore here the distinction between function and decision classes): $\NCOne\subseteq\DET\subseteq\NCTwo$, to which the theories $\VNCOne\subseteq V\#L\subseteq\VNCTwo$ correspond.

Our argument cannot be carried out in $V\#L$ since the evaluation of algebraic circuits, even those with squared logarithmic depth (or those in algebraic-\ACOne)  over the integers, which is crucial to our argument,  is apparently not definable in $V\#L$. Note that excluding the evaluation of low-depth algebraic circuits all our arguments seem to  carry over to $V\#L$.  This also includes for example our algorithm for balancing algebraic circuits.\footnote{It is possible also to balance algebraic circuits to squared logarithmic depth in \DET\ using some variants of the algorithm  in \cite{AJMV98}, as we were informed by Eric Allender \cite{All18}.} 


Note also that the two classes $\SharpSACOne\subseteq\TCOne$ that are above \DET\ but below \NCTwo, can compute the required depth reduction and the evaluation of algebraic circuits. \emph{We believe that our construction can be carried out more or less  the same in theories corresponding to these classes}. However, for these two classes we are not aware of established bounded arithmetic theories, hence we shall work in \VNCTwo.  

\section{Preliminaries}\label{sec:prelim}

In this section we present  some of the necessary  background from bounded arithmetic as well as algebraic circuit complexity. Specifically, we describe the two-sorted bounded arithmetic theory \VNCTwo\ as developed by Cook and Nguyen \cite{CN10} and show how to define the evaluation of algebraic  circuits over the integers in the theory, and then define algebraic circuits computing formal polynomials and proof systems for polynomial identities \cite{HT08,HT12} (cf.~\cite{PT16} for a survey). We start with an exposition of bounded arithmetic.

\textit{Bounded arithmetic} is a general name for  weak formal  systems of arithmetic, namely, fragments of Peano Arithmetic. The bounded arithmetic theories we use are first-order two-sorted theories, having a first-sort for natural numbers and a second-sort for finite sets of numbers, representing bit-strings via their characteristic functions (for the original  \textit{single-sort} treatment of theories of bounded arithmetic see  \cite{Bus86,HP93,Kra95}). The theory \VZ\ corresponds to the complexity class uniform-\ACZ, and \VNCTwo\ corresponds to uniform-\NCTwo. The complexity classes \ACZ, \NCTwo, and their corresponding function classes \FACZ\ and \FNCTwo\ are  defined using a two-sorted universe (specifically, the first-ordered sort [numbers] are given to the machines in unary representation and the second-sort as binary strings).
See Section \ref{sec:the-complexity-class-NCTwo} below for the definitions of \NCTwo\ and \FNCTwo, and Definition \ref{def:FACZ} in the appendix for \ACZ\ and \FACZ. 
\begin{definition}[Language of two-sorted arithmetic \LTwoA] The language of two-sorted arithmetic, denoted \LTwoA, consists of the following relation, function and constant symbols:            
\[ \set{+,\cd,\le, 0,1,|\ |,=_1,=_2,\in}.\]
\end{definition}

We describe the intended meaning of the symbols by considering the standard model $\N_2$ of two-sorted Peano Arithmetic. It consists of a first-sort universe $U_1=\N$ and a second-sort universe $U_2$ of all finite subsets of $\N$, which are thought of as strings. The constants $0$ and $1$ are interpreted in $\N_2$ as the appropriate natural numbers zero and one, respectively. The functions $+$ and $\cd$ are the usual addition and multiplication on the universe of natural numbers, respectively. The relation $\le$ is the appropriate ``less or equal than'' relation on the first-sort universe. The function $\abs{\cd}$ maps a finite set of numbers to its largest element plus one. The relation $=_1$ is interpreted as equality between numbers, $=_2$ is interpreted as equality between finite sets of numbers. The relation $n\in N$ holds for a number $n$ and a finite set of numbers $N$ if and only if $n$ is an element of $N$.

We denote the first-sort (number) variables by lower-case letters $ x,y,z,\dots$, and the second-sort (string) variables by capital letters $ X,Y,Z,\dots$. \iddo{***\MakeUppercase{Building terms?}***}

We build formulas in the usual way, using two sorts of quantifiers: number quantifiers and string quantifiers. A number quantifier is said to be \emph{bounded} if it is of the form $\exists x (x\leq t\wedge\dots)$ or $\forall x (x\leq t\rightarrow\dots)$, respectively, for some \textit{number term} $t$ that does not contain $ x $. We abbreviate  $\exists x (x\leq t\wedge\dots)$ and $\forall x (x\leq t\rightarrow\dots)$ by $ \exists x\le t $ and $ \forall x\le t $, respectively.  A string quantifier is said to be \emph{bounded} if it is of the form $\exists X (\abs{X}\leq t \wedge\dots)$ or $\forall X (\abs{X}\leq t\rightarrow\dots)$ for some \textit{number term} $t$ that does not contain $ X $. We abbreviate $\exists X (\abs{X}\leq t \wedge\dots)$ and $\forall X (\abs{X}\leq t\rightarrow\dots)$ by $ \exists X\le t $ and $ \forall X\le t $, respectively. 

A formula is in the class of formulas $\SigZB$ or $\mathbf\Pi^B_{0}$ if it uses \emph{no string quantifiers} and all number quantifiers are bounded. A formula is\mar{Change Sigma to bold..} in $\mathbf\Sigma^B_{i+1}$ or $\mathbf\Pi^B_{i+1}$ if it is of the form $\exists X_1\leq t_1 \dots\exists X_m\leq t_m \psi$ or $\forall X_1\leq t_1\dots\forall X_m\leq t_m \psi$, where $\psi\in\mathbf\Pi^B_i$ and $\psi\in\mathbf\Sigma^B_i$, respectively, and $ t_i $ does not contain $ X_i $, for all $ i =1,\dots,m $. We write $\forall \SigZB$ to denote the universal closure of $ \SigZB$ (i.e., the class of $\SigZB$-formulas that possibly have [not necessarily bounded] universal quantifiers on   their front [left]). We  write $T(t)$ to abbreviate $ t\in T $, for a number term $ t $ and a string term $ T $.


As mentioned before, a finite set of natural numbers $N$ represents a finite string $S_N=S^0_N\dots S^{\abs{N}-1}_N$ such that $S^i_N=1$ if and only if $i\in N$. We will abuse notation and identify $N$ and $S_N$.



\subsection{The Theory \VZ} \label{sec:VZ}

The base theory $\VZ$, which corresponds to the computational class \ACZ, consists of the following axioms:\vspace{-2pt} 

\noindent\HRule\vspace{-6pt} 
\begin{gather*}
\textbf{Basic 1}.\  x+1\neq 0  ~~~~~~~~~ \textbf{Basic 2}.\ x+1=y+1\rightarrow x=y \\
\textbf{Basic 3}.\   x+0=x ~~~~~
\textbf{Basic 4}.\ x+(y+1)=(x+y)+1
\\
\textbf{Basic 5}.\   x\cdot 0=0 ~~~~~~~~~~~
\textbf{Basic 6}.\ x\cdot(y+1)=(x\cdot
y)+x 
\\
\textbf{Basic 7}.\  (x\leq y \wedge y\leq x)\rightarrow x=y~~~~~~                \textbf{Basic 8}.\   x\leq x+y 
\\
\textbf{Basic 9}.\   0\leq x ~~~~~~~~~~~~~~~~~~~~~~~~~~~~~~ \textbf{Basic 10}.\   x\leq y\vee y\leq x     
\\
\textbf{Basic 11}.\  x\leq y\leftrightarrow x<y+1~~~~~~~~~~~~~ ~~~~~~~~~~~~~~~~~~~~~~~~~~~~~~~~ \\         
\textbf{Basic 12}.\   x\neq    0\rightarrow\exists y\leq x(y+1=x)~~~~~~~~~~~~~~~~~~~~~~~~~~~~~~~
\\
\textbf{L1}.\  X(y)\rightarrow y<\abs{X}~~~~~~~~~~~~~~~~~~~~\textbf{L2}.\  y+1=\abs{X}\rightarrow X(y)
\end{gather*}
\begin{equation*}
\begin{split}
\mbox{\bf{SE}. }(\abs{X}=\abs{Y}\wedge \forall i\leq\abs{X}\,
        (X(i)\leftrightarrow Y(i)))\rightarrow X=Y~~~~~~~~
\\
\mbox{{\bf $\SigZB$-COMP.\ }}   \exists X\leq y\forall z<y\, (z\in X\leftrightarrow \varphi (z))\,,~~\mbox{for all}~
\varphi \in\SigZB 
\\
\qquad\qquad  \ \ \ \mbox{where $ X $ does not occur freely in $ \varphi $}\,.
\end{split}
\end{equation*}
\noindent \HRule%

Here, the axioms {\bf Basic 1} through {\bf Basic 12} are the usual axioms used to define Peano Arithmetic without induction ($\mathsf{PA^-}$), which settle the basic properties of addition, multiplication, ordering, and of the constants 0 and 1. The Axiom {\bf L1} says that the length of a string coding a finite set is an upper bound to the size of its elements. {\bf L2} says that $\abs{X}$ gives the largest element of $X$ plus $ 1 $. {\bf SE} is the  axiom for strings which states that two strings
are equal if they code the same sets. Finally, {\bf
$\SigZB$-COMP} is the comprehension axiom \emph{scheme} for
$\SigZB$-formulas (i.e., it is an axiom for each such formula) and
implies the existence of all sets which contain exactly the
elements that fulfill any given $\SigZB$ property.


\begin{proposition}[Corollary V.1.8. \cite{CN10}]\label{prop:number-induction}
The theory \VZ\ proves the \emph{number} induction axiom scheme for \SigZB-formulas $ \Phi $:\vspace{-5pt} 
\[
        \left(
                        \Phi(0)\land \forall x \left(
                                                                                                                        \Phi(x) \rightarrow \Phi(x+1)
                                                                                                 \right)
        \right)
                        \rightarrow \forall z\,\Phi(z).
\]
\end{proposition}
In the above induction axiom, $ x $ is a number variable and $ \Phi $ can have additional free variables of both sorts.

We  seek to define the determinant function in a theory via a \SigOneB-formula, where a function is said to be \emph{defined in a theory} if the theory can prove that given an input to the function there always exists a unique output. For the exact definition of definability of functions in \VZ\ (and \VNCTwo) consult the appendix (Section \ref{sec:Definability-in-Bounded-Arithmetic}). Note that the \SigOneB-definable functions of \VZ\ (equivalently, the \SigZB-definable functions of \VZ) are precisely the \FACZ\ functions, and that the \SigOneB-definable functions of \VNCTwo\  are precisely the \FNCTwo\ functions (see Theorem \ref{thm:definable_func_of_vnctwo} below).


\subsection{The Complexity Class \NCTwo}
\label{sec:the-complexity-class-NCTwo}

\label{sec:the-uniform-complexity-class-nctwo}
The uniform complexity class \NCTwo\ is defined using an alternating time-space (nondeterministic) Turing machine.

\para{Alternating Turing machines} An alternating Turing machine is a \emph{nondeterministic} Turing machine in which every state, except the halting states, is either an \emph{existential state} or a \emph{universal state}. A \textit{computation} in such a machine can be viewed as an (unbounded fan-in) tree of configurations as follows. A configuration is said to be \textit{existential} (resp.~\textit{universal}) if its state is existential (resp.~universal). In a computation tree of an alternating Turing machine every \emph{existential} configuration has one or more children, such that each child is a configuration reachable in one step from the configuration in the parent node; and every universal configuration has as its set of children \emph{all} configurations reachable in one step from the configuration on the parent in node. We say that a computation of an alternating Turing machine is \emph{accepting} when all the leaves of the computation tree  are accepting configurations. We say that an alternating Turing machine \emph{accepts an input} $x$ if there \textit{exists} an accepting computation tree whose root is the initial configuration with the input $x$.

A computation tree is said to have \textit{$k$ alternations }if the number of alternations between existential and universal states in every branch of the tree is at most $k$. An alternating Turing machine  is said to \emph{work in $f(n)$ alternations} if for every input $x$ of length $n$ the number of alternations in \emph{every }computation tree of $x$ is at most $f(n)$. A computation tree is said to have \textit{space $s$ }if the working space used in every configuration of the tree is at most $s$.  An alternating Turing machine is said to \emph{work in space $g(n)$} if for every input $x$ of length $n$ the space of every computation tree of $x$ is at most $g(n)$.\bigskip

\begin{definition}[Uniform \NCTwo]
\label{def:nctwo}
The uniform  complexity class \NCTwo\ is defined to be the class of languages that can be decided by alternating Turing machines with  $O(\log n)$ space and $O(\log^2 n)$ time. 
\end{definition}

%

\TODO{check \FO. Also didn't define it. Isn't logspace uniform enough? Same?}

We define the \emph{function} class \FNCTwo\ as the function class containing all number functions $f(\vec x, \vec X)$ and
string functions $F(\vec x,\vec X)$, where $\vec x$ and $\vec X$ are number and string variables,
respectively, such that the relation of the function is defined (resp.~bit-defined; see Definition \ref{def:sigB1-definability-in-a-langauge})\TODO{defined?} in \NCTwo\ (a binary relation $R$ is \textit{defined in \NCTwo} if the language containing the set of pairs in $R$ is decidable in \NCTwo).

\para{\NCTwo\ Boolean circuit families} 
Let $\{C_n\}_{n=1}^\infty$ be a family of Boolean circuits (with fan-in at most two $\lor,\land,\neg$ gates). We say that this family is an  \emph{\NCTwo\ circuit family} if every circuit $C_n$ in the family has depth $O(\log^2 n)$ and size $n^{O(n)}$. A circuit taken from a given Boolean  \NCTwo\ circuit family is said to be an \emph{\NCTwo-circuit}. It is known that the \NCTwo\ circuit value problem is complete under \ACZ-reductions for the class \NCTwo\ (Definition \ref{def:nctwo}). We say that $\{C_n\}_{n=1}^\infty$ is a \textbf{\textit{uniform \NCTwo-circuit family}} if its extended connection language is in \FO\ (we refer the reader to \cite[page 455]{CN10} for the definitions). This definition coincides with Definition \ref{def:nctwo}. 

For the definition of uniform \NCOne\ (and \ACOne) we also refer the reader to \cite{CN10}.


\subsection{The Theory \VNCTwo}\label{sec:vnctwo}
Here we define the theory \VNCTwo\ as developed in \cite{CN10}.
It is an extension of \VZ\ over the language \LTwoA\ where
we add the axiom stating the existence of a sequence of
values that represent the evaluation of monotone Boolean
circuits of $O(\log^2(n))$-depth. It is known (cf.~\cite{CN10}) that the Monotone Boolean Circuit Value problem for circuits of $O(\log^2(n))$-depth is complete under \ACZ-reductions for \NCTwo.

 The \NCTwo\ \textsc{circuit value problem} is the problem that  determines the value computed by a Boolean \NCTwo-circuit, given a 0-1 assignment to its input variables. An input circuit  to the problem is encoded  as a \textit{layered circuit} with $d+1$ layers, namely, a circuit in which every node in layer $j$ is connected only to zero or more nodes in layer $j+1$. The actual evaluation of such an (\NCTwo) circuit within the class \NCTwo\ is done in stages, where we start from layer $0$
and ``compute'' (using alternations and nondeterminism)
the values of every node in every layer. Formally, we define this evaluation process as follows (see also \cite[Chap.~IX.5.6]{CN10}).

The layered monotone Boolean circuit with $d+1$ layers is encoded with a string variable $I$, with $|I|\le n$, which defines the (Boolean) input gates to the circuit. Then we have a string variable $G$ such that $G(x,y)$, for $x\in[d]$, holds iff the $y$th gate in layer $x$ is $\land$, and is $\lor$ otherwise. Also the wires of $C$ are encoded by a three-dimensional array, namely a string variable $E$ such that $E(z,x,y)$ holds iff the output of gate $x$ on layer $z$ is connected to the input of gate $y$ on layer $z+1$. To compute the     value of each of the gates in the circuit $C$ on input $I$, simply compute the values of the gates in each layer, starting from the input layer, in $d+1$ stages, using the values of the previous layer. The formula $\delta_{LMCV}(n,d,E,G,I,Y)$ below formalizes this evaluation procedure (where \textit{LMCV} stands for ``layered monotone circuit value"). The two-dimensional array $Y$ stores the result of computation, namely the evaluation string: for $1\le z\le d$, row $Y^{[z]}$ contains the gates on layer $z$ that output $1$.
\begin{multline}\label{eq:LMCV}
\delta_{LMCV}(n,d,E,G,I,Y)\equiv \\
 \forall x<n\forall z<d~ 
    \Big(
        \left(
            Y(0,x)\leftrightarrow I(x)
        \right)
        \land \\
        \big(
            Y(z+1,x)\leftrightarrow
                \big(\big(
                    G(z+1,x)\land\forall u<n,                         E(z,u,x)\to \\ Y(z,u)
                 \big)\lor 
                \left(
                    \neg G(z+1,x)\land \exists                         u<n, E(z,u,x)\land                             Y(z,u)
                \right)
         \big)
   \big)\Big).
\end{multline}     
The following formula states that the circuit with underlying graph $(n,d,E)$ has fan-in two:
\begin{multline}
Fanin2(n,d,E)\equiv \\ \forall z<d \,\forall x<n\exists u_1<n\exists u_2<n\exists v<n (E(z,v,x)\to\\ (v=u_1\lor v=u_2)\big).
\end{multline}
%

Finally, we arrive at the definition of \VNCTwo:
 \begin{definition}[\VNCTwo]\label{def:vnctwo}
The theory \VNCTwo\ has vocabulary \LTwoA\ and is axiomatized by the axioms of \VZ\ and the axiom:
\begin{equation*}
Fanin2(n,|n|^2,E)\to  \exists Y \le \langle |n|^2+1,n \rangle \delta_{LMCV}(n,|n|^2,E,G,I,Y).
\end{equation*}
\end{definition} 
In this definition $\langle\cd \rangle $ is the pairing function, and $\langle |n|^2+1,n \rangle$ is an upper bound on the length needed for the two-dimensional array $Y$. Also, note that given a natural number $n$ the binary representation length of $n$, denoted $|n|$, that is, $\lceil \log_2(n+1)\rceil$, is an \ACZ\ function of $n$ (see \cite[Exercise III.3.30]{CN10}). 
\smallskip 

Recall the concept of a \SigOneB-definable function in a theory (see the  appendix Section \ref{sec:Definability-in-Bounded-Arithmetic}).
The following is the main theorem for \VZ\ and \VNCTwo:

\begin{theorem}\rm{(\cite[Corollaries V.5.2 and IX.5.31]{CN10}})\label{thm:definable_func_of_vnctwo}
A function is $\SigOneB$-definable in \VZ\ iff it is $\SigZB$-definable in \VZ\ iff it is in \FACZ. A function is $\SigOneB$-definable in \VNCTwo\ iff it is in \FNCTwo. 
\iddo{}
\end{theorem}

Note that the fact that a function is defined in the theory does not mean that we can prove all of its properties, or even anything interesting about it. To actually prove statements about a  $\SigOneB$-definable function in \VNCTwo, for example, we need to carefully consider the  $\SigOneB$-formula defining it, formulate the property that we want to prove in the theory as a formula in the language \LTwoA, and verify that indeed the formula is provable in the theory.   



\subsection{Polynomials and Algebraic Circuits}\label{sec:algebraic_circuits} \label{sec:notation}
For a good monograph  on algebraic circuits and their complexity see Shpilka and Yehudayoff \cite{SY10}. Let \G\ be a ring. Denote by $\G[X]$ the ring of (commutative) polynomials with coefficients from $ \G $ and variables $X:=\{x_1,x_2,\,\dots\,\}$. A \emph{polynomial }is a formal linear combination of monomials, where a \emph{monomial} is a product of variables. Two polynomials are \emph{identical }if all their monomials have the same coefficients. The \emph{degree} of a polynomial is the maximal total degree of a monomial in it.

Algebraic circuits and formulas over the ring \G\   compute polynomials in $\G[X]$ via addition and multiplication gates, starting from the input variables and constants from the field. More precisely, an \emph{algebraic circuit} $C$ is a finite directed acyclic graph (DAG) with \textit{input nodes} (i.e., nodes  of in-degree zero) and a single \textit{output node} (i.e.,  a node of out-degree zero).  Input nodes are labeled with  either a variable or a field element in $\F$.
All the other nodes have in-degree two (unless otherwise
stated) and  are labeled by
either an addition gate  $+$ or a product gate   $\times$. An input node is said to \emph{compute}  %
%
the variable or scalar that   labels  itself. A $+$ (or $\times$) gate
is said to compute the addition (product, resp.) of the (commutative) polynomials
computed by its incoming nodes. An algebraic circuit is called a \emph{formula}, if the underlying directed acyclic graph  is a tree (that is, every node has at most one outgoing edge). The \emph{size} of a circuit $C$ is the number of nodes in it, denoted $|C|$,
and the \emph{depth} of a circuit is the length of the longest directed path in it. \mar{Check cdot was change to  times gate!}

We say that a polynomial is \emph{homogeneous} whenever every monomial in it has the same (total) degree. 

\begin{definition}[Syntactic-degree $\deg(\cd)$]\label{def:syntactic-degree}
Let $C$ be a circuit and $v$ a node in $C$. The \emph{syntactic-degree} \(\deg(v) \) of $v$ is defined as follows:\vspace{-0pt} 
\begin{enumerate}
\item If $v$ is a field element or a variable, then $\deg(v):=0$ and  $ \deg(v):=1$, respectively;

\item  If $v=u+w$ then $\deg(v):= \max\{\deg(u), \deg(w)\}$;
\item If $v=u\cdot w$ then $\deg(v):= \deg(u)+\deg(w)$.
\end{enumerate}
\end{definition}

An algebraic circuit is said to be \emph{syntactic-homogeneous} if for every plus gate $u+v$, $\deg(u)=\deg(v)$.

Given a  circuit $ F $ and a node $u$ in $ F $, $ F_u$ denotes the subcircuit of $F$ with output node $u$. 
\label{defu}
If $F,G$ are two circuits then
\[
\mbox{  $F\cplus G$ and $F\ctimes G$ }
\]
denotes  \emph{any} circuit $H$ whose output node is $u+v$ or $ u\times v$, respectively, where $H_u$ is the circuit $F$ and $H_v$ the circuit $G$. In other words, $F\cplus G$ denotes a circuit with output node $+$ with the two incoming subcircuits $F$ and $G$, where $F$ and $G$ \emph{may not be disjoint} (so $F\cplus G$ is a set of possible different circuits, from which we assume one is picked; the two subcircuits $F,G$ of the output node of $F\cplus G$ are \emph{identical} to $F,G$, respectively). Furthermore, \[
\mbox{  $F+ G$ and $F\times G$ }
\]
 denote the \emph{unique} circuit of the form $F'\cplus G'$ and $F'\ctimes G'$, respectively, where $F'$, $G'$ are disjoint copies of $F$ and $G$. In particular, if $F$ and $G$ are formulas then so are $F+ G$ and $F\times G$. For example, $(1+x_5)\ctimes x_5$ can be any of the following two circuits:


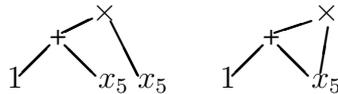
\newlength{\mylength}
\setlength{\mylength}{\unitlength}

\setlength{\unitlength}{0.4cm}
~~~~~~~~~~~~~~~~~~~~~~~~~~~~~~~~~~~~~~~~~~~~~
\begin{picture}(2,4)
\thicklines
\put(1,0.5){$1$}
\put(1.4,0.9){\line(1,1){1}}
\put(2.4,1.9){+}
\put (2.9, 1.9){\line(1,-1){1}}
\put(4, 0.5){$x_5$}
\put(2.8, 2.3){\line(2,1){1}}
\put(3.8, 2.7){$\times$}
\put (4.4, 2.7){\line(1,-2){0.9}}
\put (5.3, 0.5){$x_5$}
\put(8,0.5){$1$}
\put(8.4,0.9){\line(1,1){1}}
\put(9.4, 1.9){+}
\put(9.9, 1.9){\line(1,-1){1}}
\put(11, 0.5){$x_5$}
\put(9.8, 2.3){\line(3,1){1.2}}
\put(11.1, 2.7){$\times$}
\put (11.3, 1){\line(1,6){0.25}}
\end{picture}

\setlength{\unitlength}{\mylength}


\subsection{Polynomial Identity (PI-) Proofs}\label{sec:PI-proofs}

In this section we give the necessary background on the
PI-proof system  \PC. This proof-system was first introduced in \cite{HT08} (under the name ``arithmetic proofs" and for algebraic formulas instead of algebraic circuits), and was subsequently studied in \cite{HT12}. 

\textit{PI-proofs}, as originally introduced in \cite{HT08}, denoted \PC\ (and \PC(\G) when we wish to be explicit about the ring  \G),
are sound and complete proof systems for the set of polynomial identities of \G, written as equations between algebraic circuits. A PI-proof starts from axioms like associativity, commutativity of addition and product, distributivity of product over addition, unit element axioms, etc., and derives new equations between algebraic circuits $F=G$ using rules for adding and multiplying two previous identities. The axioms of \PC\ express reflexivity of equality, commutativity and associativity of addition  and product,  distributivity, zero element, unit element, and true identities in the field.

Algebraic circuits in PI-proofs are treated as purely syntactic objects (similar to the way a propositional formula is a syntactic object in propositional proofs). Thus, simple computations such as multiplying out brackets, are done explicitly, step by step.     



\begin{definition}[\textbf{\textit{PI-proofs}}; System \PC(\G), \cite{HT08,HT12}] \label{def:PI-proofs} The system \PC$(\G)$ proves equations of the form $F=G$ over the ring \G, where $F,G$ are algebraic circuits over $\G$. The inference rules of $\PC$ are (with $F,G,H$ ranging over  algebraic circuits, and where an equation below a line can be derived from the one above the line):
\newcommand{\mez}{\qquad}
\begin{align*}
&{\rm R1}&   \frac{F=G}{G=F} \mez &~~~~~{\rm R2}&  \frac{F=G\qquad G=H}{F=H} && \\
&{\rm R3}&  \frac{F_1=G_1 \qquad  F_2=G_2}{F_1+ F_2= G_1+ G_2}&~~~~~ {\rm R4}& \frac{F_1=G_1\qquad F_2=G_2}{F_1\cdot F_2= G_1\cdot G_2} .&&
\end{align*}
The axioms are equations of the following form, with $F,G,H$ circuits:

\noindent\HRule 

A1 ~~ $F=F$  

A2 ~~~$F+G=G+F$ 

A3 ~~  $F+(G+H) = (F+G)+H$

A4 ~~ $F\cdot G=G\cdot F$

A5 ~ ~$ F \cdot (G\cdot H)= (F\cdot G)\cdot H$ 

 A6~~~   $F\cdot (G+H) = F\cdot G+F\cdot H$    

A7 ~  ~$F+0 = F$  

A8 ~ ~$F\cdot 0 = 0$

A9 ~~  $F\cd 1 = F$   

A10 ~~ $a=b+c\,, ~ a^\prime=b^\prime\cdot c^\prime\, $ ~ (if $a,b,c,a^\prime,b^\prime,c^\prime\in  \G$,  

~~~~~~~~~~~~~~~~~~~~~~~~are such that the equations hold in \G);

C1 ~~$F\cplus G=F+G$\, 

C2  ~~$F\ctimes G=F\cdot G$\, 

\noindent\HRule

\smallskip 

A \emph{\PC(\G)-proof                                                                                                                                                                                                                                                                                                                                                                                                                                                                                                                                                                                                                                                                                                                                                                                                                                                                                                                                                                                                                                                                                                                                                                                                                                                                                                                                                                                                                                                                                                                                                                                                                                                                                                                                                                                                                                                                                                                                                                                                                                                                                                                                                                                                                                                                                                                                                                                                                                                                                                                                                                                                                                                                                                                                                                                                                                 } is a sequence of equations, called \textbf{proof-lines},  $F_1=G_1,\, F_2=G_2,\dots, F_k=G_k$, with $F_i, G_i$ circuits, such that every equation is either an axiom or was obtained from previous equations by one of the inference rules. The \textbf{size} of a proof  is the total size of all circuits appearing in the proof. The \emph{number of steps} in a proof is the number of proof-lines in it.
\end{definition}

A PI-proof can be easily verified for correctness in deterministic  polynomial-time (assuming the field (or ring) has efficient representation; e.g., the field of rational numbers or the the ring \Z), simply by syntactically checking that each proof line is derived from previous lines by one of the inference rules.  

\newcommand{\demph}[1]{\textbf{\emph{#1}}}

\subsection{Circuits and Proofs with Division}\label{sec:circ_proof_with_divison}
We denote by $\G(X)$ the field of formal rational functions in the variables $X, $ where a formal rational fraction is a fraction of two formal polynomials with coefficients from \G. In this work we will consider \G\ to be the ring of integers \Z. We will not be interested in `inverse elements' in \Z\ (excluding the element 1),  nor much  in the completeness or soundness of proof systems for rational functions (like $\PIZ$ described below), because the theory will only prove \emph{syntactical} properties of these proof systems (hence, no actual `division' is performed over the integers).

 It is possible to extend the notion of a circuit so that it computes rational functions in $\G(X)$ (\cite{HT12}). This is done in the following way:  a \demph{circuit with division} $ F $ is an algebraic circuit which may contain an additional type of gate with fan-in $1$, called an \emph{inverse} or a \emph{division} gate, denoted $(\cd)^{-1}$. A division gate $v^{-1}$ (i.e., a division gate whose incoming circuit is $v$) computes the rational function $1/\widehat{v}\in \G(X)$, assuming $v$ does not compute the zero polynomial. If the circuit with division $F$ contains some division gate  $v^{-1}$ such that $v$ computes the zero polynomial, then we say that the circuit $F$ is \emph{not well-defined}, and is otherwise \emph{well-defined}. 
Note, for instance, that the circuit  $(x^{2}+x)^{-1}$ over $ \GFTwo$ is well-defined, since $x^{2}+x$ is not the zero polynomial   (although it vanishes as a function over \GFTwo, for example).
\TODO{Check if things are okay here with \G\ a ring not a field.}

\newcommand{\wh}{\widehat}

We define the system $\PI(\G)$, operating with equations $F=G$ where $F$ and $G$ are circuits with division \cite{HT12}, as follows: 
first, we extend the axioms of $\PC(\G)$ to apply to well-defined circuits with division. Second, we add
 the following new axiom: 
\[
 {\rm D}~~~~~~~~~ F\cdot F^{-1}= 1 \,,
~~\mbox{provided that $F^{-1}$ is well-defined.}
\]
Note that if $F^{-1}$ is well-defined then both $F$ is well-defined and $F\neq 0$. We sometimes call the \PI-system \emph{PI-proof} as well (although it operates with rational functions and not merely polynomial).    

We say that a \PI-proof is \emph{\textbf{syntactically correct}} if it is a correct \PI-proof except that in the axiom D above $F^{-1}$ is \emph{not} necessarily well-defined. Since we do not know how to check in uniform \NCTwo\ that a circuit is well-defined, we do not know how to express the full correctness of \PI-proofs in the \VNCTwo. For our purposes it is sufficient that \VNCTwo\ expresses only the syntactic correctness  of \PI-proofs.

The \emph{syntactic-degree of a circuit $C$ with division} is defined as  
$$\deg(C):=\deg(\Num(C))+\deg(\Den(C)).$$

\section{Carrying the Proof in the Theory: Overview}\label{sec:carry-proof-in-theory}


Here we provide a detailed overview of the proof of  the determinant identities in the theory, as highlighted before in Section \ref{sec:overview}. \iddo{}

We assume all polynomials are over the ring of integers $\Z$. We reason inside \VNCTwo\ (and \VZ) about $\PIZ$- and $\PCZ$-proofs (Definition \ref{def:PI-proofs} and 
Section \ref{sec:circ_proof_with_divison}). We use the following \emph{reflection principle}, stating that if an equation has a proof then the equation is true:

%

\begin{theorem}[\PCZ-reflection principle; In \VNCTwo]\label{thm:ovr:sound-of-arithmetic-proofs}
Let $\pi$ be an $O(\log^2 n)$-depth $\PCZ$-proof of the  equation $F=G$. Then $F=G$ is \emph{true} in \Z; that is, $\forall \alpha \in\Z^n(F(\alpha)=G(\alpha))$.
\end{theorem}
Theorem \ref{thm:ovr:sound-of-arithmetic-proofs} is proved as follows. We define the \emph{evaluation function} for $O(\log^2 n)$-depth algebraic circuits over \Z\ as the function that receives an integer assignment $A$ and an $O(\log^2 n)$-depth algebraic circuit $C$. The algorithm then converts $C$ into a layered and monotone \textit{Boolean} \NCTwo\ circuit, where the inputs are the bit-strings corresponding to  $A$. And then evaluates the Boolean circuit using evaluation of \NCTwo\ circuits (\SigOneB-definable in \VNCTwo), and finally outputs the result (see  Section \ref{sec:algebraic_circuit_value_problem}).   

We also need to show in \VNCTwo\ that the rules and axioms of $O(\log^2 n)$-depth \PCZ\ are sound with respect to the above evaluation function. This is proved by inspection of each of the axioms and rules.

Note that we do \emph{not} know how to prove the soundness of \PI-proofs in \VNCTwo. This is because the division axiom $F\cdot F^{-1}= 1$ requires that $\widehat F\neq 0$, and we do not know how to check in \NCTwo\ that a circuit does not compute the zero polynomial. However, we observe that for our purposes it is enough to show that given a specific object (a \PI-proof for the determinant identities) it is possible to obtain from this object a new legitimate \PC-proof of the determinant identities. See the example in the  \textit{Eliminating division gates} part below. 

\TODO{Should be the same Theorem as above!}
\medskip


\para{The determinant function \Det\ in the theory} We now describe slightly informally the (uniform-\NCTwo) determinant function \Det\ defined in the theory. Essentially, each step in the algorithm corresponds to a (more involved) step in the construction of the final PI-proof of the determinant identities in the theory (as described after the algorithm below). 
\bigskip
\iddo{}

\noindent\hrule 
\smallskip
 
\noindent\textsc{Algorithm  \Det} (in \VNCTwo)
\smallskip

\noindent\hrule 
\medskip

\noindent\textbf{Input}: an $n\times n$ integer matrix $A.$

\noindent\textbf{Output}: $z\in\Z$, where $z$ is the determinant of $A$.

\smallskip

\begin{enumerate}
\item Write down an unbalanced algebraic circuit $\cdeti(X)$ with division  that computes the symbolic \nbyn\ determinant polynomial, over the variables $X=\{x_{ij}\}_{i,j\in[n]}$. This circuit captures the standard recursive formula for computing the determinant of matrices, using Schur complement (intuitively, it captures the Gaussian elimination procedure). For details see Section \ref{sec:circuit-with-div-for-det}.
\item Consider the circuit $\cdeti(I_n+zX)$ as computing a univariate polynomial in the new variable $z$. Using this circuit, construct a new circuit $\taydet(X)$ computing the  $n$th term of the Taylor expansion of $\cdeti(I_n+zX)$ around $z=0$.
This is a circuit  \emph{with }a division gate, of exponential syntactic-degree, that computes the determinant \emph{as a polynomial}. For details see Section \ref{sec:PI-proof-from-det-div-to-Taylor}

\item Convert the circuit $\taydet(X)$ into a syntactic homogeneous circuit without division of  syntactic-degree $n$, denoted $\taydetsharp(X)$. For details see Section \ref{sec:Reducing-the-Syntactic-Degree-of-the-Determinant-Polynomial}.
 \iddo{}

\item Make sure that  constant leaves in $\taydetsharp(X)$, when treated as if they are variables, do not lead up to nodes of exponential syntactic-degree in $\taydetsharp(X)$. This is done by constructing the circuit $\taydetsharpprime(X)$, based on $\taydetsharp(X)$. For details see Section \ref{sec:taking-care-of-high-degubp-nodes}.




\item Balance $\taydetsharpprime(X)$ via a (uniform) balancing algorithm, to yield a polynomial size and  $O(\lgsq)$-depth circuit without division denoted $\cdetbl(X)$ that computes the determinant polynomial. For details see Section \ref{sec:balancing-algebraic-circ-in-theory}. 

\item Evaluate the circuit $\cdetbl(X)$ with the input assignment $A$, using the algebraic  circuit evaluation function for $O(\lgsq)$-depth circuits, and output the resulting integer in binary.
For details see Section \ref{sec:applying-the-reflection}. 

\end{enumerate}
\medskip 

\hrule 
\bigskip

\TODO{The algorithm above has inaccuracies probably: do we need to first substitution bi-yi instead of $x_i$'s? Does Taylor depends on the non-nullifying assignment being 0 assignment?}

Since we show that all the parts in the algorithm above are $\SigOneB$-definable functions in \VNCTwo, the determinant function as defined above is $\SigOneB$-definable in the theory (namely, totally recursive).\medskip

Given the function \Det\ we now sketch the proof in \VNCTwo\ of the two determinant identities  (\ref{eq:1}), (\ref{eq:2}) below.
\mar{***change -1 to div  or something ,in notations *** }\medskip

\TODO{Something is not right: when do we do the (b-y) substitution? We need to do it before the Taylor expansion I believe.}



\para{Step 1: Existence of \PIZ-proofs with division gates} We show in \VZ\  a \SigZB-definable function that given a natural number $n$ outputs a $\PC^{-1}(\Z)$-proof \(\pi_0 \) of the following equations 
\medskip

\noindent\fbox{%
    \parbox{\textwidth}{%
\begin{equation}\label{eq:1}
\cdeti(X)\cd\cdeti(Y)=\cdeti(XY) 
\end{equation}
\begin{equation}\label{eq:2}
\cdeti(Z)=z_{11}\cdots z_{nn}, 
\end{equation}
    }%
}

\medskip
\noindent for $X,Y$ symbolic \nbyn\ matrices; that is, the $(i,j)$th entry of $X$ and $Y$ are the variables $x_{ij}$ and $y_{ij}$, respectively, and $Z$ a lower (equivalently, upper) triangular symbolic matrix in which the variable $z_{ij}$ is the $(i,j)$th entry of $Z$ iff $i\ge j$, and 0 otherwise.  

These are equations between algebraic circuits over \Z. This is a proof in which circuits have exponential syntactic-degrees (though the theory cannot express this fact). The circuits in the proof are not necessarily homogeneous, and have division gates. The theory can also only prove that the \PIZ-proof is syntactically correct (Section \ref{sec:circ_proof_with_divison}).  
Note that $\cdeti(X)$ computes the determinant as a rational function and not as a polynomial.
The construction of the proofs uses only the \SigZB-\textbf{COMP} axiom and thus is done already in \VZ. See Section \ref{sec:P-C-inverse-proof-of-1-2} for  details.
\TODO{The same as the problem in the DET algorithm: shouldn't we substitute (bi-yi) before Taylor step?}

\para{Step 2: From the determinant polynomial to a rational function}  For technical reasons  relating to eliminating both division gates and high syntactic-degrees, we will need to construct in the theory a $\PC^{-1}(\Z)$-proof of the determinant identities in which the   determinant circuits appearing in the  identities that are proved are, firstly, written as polynomials and not as rational functions, that is, as circuits without division; and secondly, have small syntactic-degree. Nevertheless, note that some intermediate  \PIZ\ proof-lines will contain the determinant written with division gates and having high syntactic-degree. The first task is achieved in the current step, and the second task in the next step.

 Let $F=F(\overline{x},z)$ be a circuit with division of syntactic-degree $d$. Similar to \cite{HT12}, we  define $\Coef_{z^{k}}(F)$ as a circuit
in the variables $\overline {x}$, computing the coefficient of $z^{k}$ in $F$, when $F$ is written as a power series at $ z=0$. In other words, $\sum_{i=0}^d\Coef_{z^i}(F)\cd z^i$ are the first $d+1$ terms in  the Taylor expansion of $F$ at $z=0$. 

Let $\taydet(X):=\Coef_{z^n}(\cdeti(I_n+zX))$ be\iddo{}  the circuit computing the $n$th term of the Taylor expansion of $\cdeti(I_n+zX)$ about  $z=0$. 
It is easy to see that  $\taydet(X)$ computes the determinant function: since every variable $x_{ij}$ is multiplied by $z$, the coefficient of $z^n$ is precisely the determinant. 

By construction, $\cdeti(I_n+zX)$ will compute the determinant \emph{as a polynomial} and will have only one division gate (this is where we differ from \cite{HT12}; due to the fact that we cannot simply substitute division gates $u^{-1}$ that compute 1 by the node 1, because the theory needs to express the correctness of this substitution in some way). 
Furthermore, since we work over \Z\ we need to make sure that the only invertible ring element needed to be used is the element 1.

We then show in \VZ\ the existence of a function that given a natural positive number $n$   outputs a $\PC^{-1}(\Z)$-proof of $\taydet(X)=\cdeti(X)$, for the \nbyn\ symbolic matrix $X$. Combined with the previous step, \VZ\ proves the existence of a \PIZ-proof, denoted $\pi_1$,  of the determinant identities \eqref{eq:1}, \eqref{eq:2}, in which the determinant circuit in \eqref{eq:1}, \eqref{eq:2} is replaced by $\taydet$. See Section \ref{sec:PI-proof-from-det-div-to-Taylor}. 




\para{Step 3: Reducing the syntactic-degree of the determinant polynomial} The circuit $\taydet(X)$ has exponential syntactic-degree (here we once more differ from \cite{HT12}, since we do not know how to formulate and prove the correctness of an \NCTwo-algorithm that  eliminates 0 nodes in general algebraic circuits, or nodes of high syntactic-degree that compute the zero polynomial).\iddo{} However, for  the next step, we need  $\taydet(X)$  to have a \textit{polynomial} syntactic-degree. We show in \VZ\ that there exists a $\PIZ$-proof of $\taydet(X)=\taydetsharp(X)$, where $\taydetsharp(X)$ has syntactic-degree $n$ and no division. This is done simply by a direct construction of such  a proof using the \SigZB-\textbf{COMP} axiom, and thus is carried out  in \VZ.

\TODO{check syntactic-degree assertion}  
Therefore, by previous steps, \VZ\ proves the existence of a \PIZ-proof of the determinant identities \eqref{eq:1}, \eqref{eq:2}, where the determinant  in these two equations is replaced by $\taydetsharp$ which is an algebraic circuit with no division gates and of syntactic-degree $n$. Denote this \PIZ-proof by $\pi_2$ (this proof does contain the determinant circuit written as \cdeti\ and $\taydet$ but only in intermediate proof-lines). See Section \ref{sec:Reducing-the-Syntactic-Degree-of-the-Determinant-Polynomial}. 
  
\para{Step 4: Bringing division gates to the top}
We say that a circuit $C$ \emph{has a division at the top} whenever $C$ is of the form $F\cd (G)^{-1}$ or $(G)^{-1}\cd F$, for two circuits $F,G$. If $F,G$ do not have division gates we say that $C$ \emph{has a single division gate at the top}. We need our circuits to have a single division gate at the top, because in the next step we need to   replace division gates by an ``approximating'' power series, but we do not know how to do it with nested divisions.      

We devise an \FACZ\ algorithm that takes an algebraic circuit with division, \emph{of any depth}, and outputs an algebraic circuit computing the same rational function   that has a \emph{single} division gate at the top.
Using this algorithm, we show in \VZ\ how to convert the  $\PIZ$-\emph{proof} $\pi_2$\TODO{Den(.) was already needed for the Taylor proof? Is this okay?} into a proof in which every circuit has  a  single division gate at the top or is division free. Denote the resulted proof by $\pi_3$.
\iddo{}
This step is shown in  Section \ref{sec:bring-div-to-top}.
\para{Step 5: Eliminating division gates}  
We now wish to eliminate the division gates from the $\PIZ$-proof $\pi_3$. 
Standard division elimination by Strassen  \cite{Str73} requires finding a total assignment to the variables, such that no division gate in the circuit equals zero under this assignment. However, we do not know how to uniformly find such assignments in uniform-\NCTwo, and so we do not know how to uniformly eliminate division gates from general algebraic circuits in \VNCTwo. We solve this by working out the division elimination only for \underline{those circuits in \(\pi_3\)}.   

In fact, the only properties of the proof-sequence $\pi_3$, as well as the proof-sequence $\pi_4$ constructed during this stage, that we need to express in the theory are about the proof-sequences having a good ``local" behaviour, namely that proof-lines in the resulted \PCZ-proof obtained after division elimination, are derived syntactically correct from previous lines according to the rules of \PCZ.  

We start with a simple example to illustrate the main idea in this and the next step, and then describe the current step in more detail.

\begin{example}
Recall that   the theory expresses only the syntactic correctness of \PIZ-proofs (due to the fact that we cannot verify that division gates do not compute the zero polynomial). In  \VZ\ we can reason  as follows about division elimination.  Start with the following \PIZ-proof: $x\cd x^{-1} = 1$. Then, apply the linear transformation $x\mapsto\ 1-x$ which yields $(1-x)\cd(1-x)^{-1}$. Substitute $1-x$ by the term $\inv_k(1-x)$ defined as: 
$$
\inv_k(1-x):=1+x+x^2+\dots+x^k,
$$ 
which serves to ``approximate'' the inverse of $1-x$ up to degree $k$, in the sense that $(1-x)\cd\inv_k(1-x)=1-x^{k+1}$. For a circuit $F$ denote by  $F^{(i)}$ the syntactic-homogeneous component of degree $i$ of $F$, which computes the sum of all (syntactic-)degree $i$ monomials in $F$. Then, \VZ\  can prove the following statement: 
\begin{quote}
\textit {``Let $k\ge 1$ be a natural number. Given $x\cd x^{-1}=1$, substitute  $1-x$ for $x$, and then substitute the circuit $\inv_k(1-x)$ for $(1-x)^{-1}$. Assuming $\hc{(1-x)}{0}=1$ has a \PCZ-proof, there exist \mbox{\rm \PCZ}-proofs of syntactic-degree at most $k$ for the following equations:\vspace{-0pt} 
\begin{align*}
\left((1-x)\cd\inv_k(1-x)\right)^{(0)}&=1,\\
\left((1-x)\cd\inv_k(1-x)\right)^{(i)}&=0, ~~~\hbox{for $1\le i\le k$.}"
\end{align*}}
\end{quote}
\end{example}
\iddo{}
\bigskip 

We now describe how to eliminate division gates in more detail.
Similar to  $\Coef(F)$, the use of $\inv_k(F)$, for a circuit $F$, involves using the inverse of the \emph{constant term of $F$}, namely, the inverse of $F^{(0)}$. This is why we need to make sure that the only invertible ring element to be used is 1 (and thus it has an inverse in \Z).
For this purpose we show that the assignment of identity matrices to the matrix-entry variables $X=\{x_{ij}\},Y=\{y_{ij}\}$ and $Z=\{z_{ij}\}$, for $i,j\in[n]$, in $\pi_3$ will result in  all  division gates $u^{-1}$ computing polynomials with a constant term 1 (though this statement is not expressed in the theory). 

Assuming for simplicity that $r_i$ (for $i\in J$) are all the variables appearing in $\pi_3$ and  $b$ is the assignment of identity matrices to the variables in $\pi_3$, substitute in $\pi_3$ the term $(b_i-w_i)$ for each $r_i$ (for all $i \in J$) denoting the obtained proof by $\pi_3'$. Then, by our assumption about identity matrices assignment, the all zero assignment $\overline 0$ to the $w_i$ variables in  $\pi_3'$ does not nullify any division gate  in $\pi_3'$. Furthermore, we show that under this assignment every division gate provably in \PIZ\ computes the polynomial $1$. Therefore, in the theory, we  construct this \PIZ-proof $\pi'_3$ (which is simply a substitution instance of $\pi_3$).

As exemplified above, let $\inv_n(H)$ be the truncated power series of $H^{-1}$ over the point determined by the identity matrices to the variables of the entries of the matrices $X,Y,Z$. Loosely speaking, this truncated power series serves as the inverse polynomial of $H$ ``up to the $n$th power". Specifically, $\widehat{H}\cd\widehat{{\rm Inv}_n(H)}=1+\hbox{[\small terms of degree $>n$]}$ (note again that  \VNCTwo\ cannot necessarily prove this equality, since general evaluation of (unrestricted depth) algebraic circuits is not known to be defined in the theory).
For every circuit $C$ with a top division gate $H^{-1}$, \VZ\ proves there exists a corresponding division-free circuit $C'$, obtained by replacing the division gate $H^{-1}$ in $C$ by $\inv_n(H)$.

Let $\pi_4$ be the corresponding division-free proof-sequence obtained from $\pi'_3$ by replacing every circuit with the corresponding division-free circuit as above. By itself $\pi_4$ is not a legal $\PCZ$-proof, since the axiom of division  in $\PC^{-1}(\Z)$ does not translate into an axiom in $\PCZ$. In other words, the axiom D of division: $F\cd F^{-1}=1$ (provided that for every division node $u^{-1}$ in $F^{-1}$, including $F^{-1}$ itself, $u$ does not compute the zero polynomial; see Definition \ref{def:PI-proofs}), 
translates into $F\cd\inv_n(F) = 1$, which is neither a legal axiom, nor  a true identity (since $F\cd\inv_n(F) = 1+\hbox{[\small terms of degree $>n$]}$). We fix this problem as follows: first, we break this equation into its lower $n+1$ syntactic homogeneous components, thereby also getting rid of terms of syntactic-degree $>n$. Second, we will need to construct explicitly \PCZ-proofs of  $(F\cd\inv_n(F) )^{(0)}= 1$ (using the notion  of \emph{provably good division gates}; see Section \ref{sec:division-elim-from-proofs}, as well as Lemma \ref{lem:inv}).



Step 5 is shown in details in section \ref{sec:division-elim-from-proofs}.


\para{Step 6: Eliminating high degrees} Here we eliminate the high syntactic-degree ($>n$) parts in the circuits appearing in $\pi_4$ (these high syntactic-degree circuits appear in the middle of the proof, and not in the identities proved). This is done by homogenizing the proof $\pi_4$. Specifically, we show a \SigZB-definable function in \VZ\ that receives an algebraic circuit $G$ of syntactic-degree $k$ and converts it into a sum of $k+1$ syntactic-homogeneous circuits $\sum_{i=0}^kG^{(i)}$ (computing the same polynomial), in which every node is labeled with an \emph{upper bound }on its syntactic-degree. We show that for our purposes it is enough to work with upper bounds on syntactic-degrees rather than the syntactic-degrees themselves. 

More generally,  we show a \SigZB-definable function in \VZ\ that given a \PCZ-proof of an equation $F=G$ of syntactic-degree $n$, decomposes the proof into $n+1$ ~\PCZ-proofs of $F^{(i)}=G^{(i)}$, for $i=0,\dots,n$, each proof having  syntactic-degree at most $i$. Combining these proofs gives a low syntactic-degree version of $\pi_4$.

This also fixes the problem caused by division elimination described at the end of  the previous step. 
We thus obtain a $\PCZ$-proof, denoted $\pi_5$, of 
equations (\ref{eq:1}) and (\ref{eq:2}), where in these two equations the determinant is written as $\taydetsharp$.

See Section~\ref{sec:homogenization-of-proofs} for more details.
\iddo{}

\para{Step 7: Balancing  algebraic circuits  in the theory} 
We \SigOneB-define in \VNCTwo\  a function that receives an algebraic circuit  $C$ with size $s$ and a number $d$ which stands for an upper bound on the syntactic-degree of $C$,  and outputs a circuit  denoted $[C]$ computing $\widehat C$ with depth $O(\log s\cd\log d+\log ^2 d)$ and size $\poly(s,d)$. \TODO{**CHECK***} As mentioned before this \FNCTwo-algorithm provides an \FACZ-implementation of most parts of the classic Valiant \textit{et al}.~\cite{VSB+83} algorithm, combining it with ideas from the  Miller \textit{et al}.~\cite{MRK88} algorithm and usages of matrix powering  (which then entails working in \VNCTwo).  

%
More generally, we show a \SigOneB-definable function in \VNCTwo\ that receives a \PCZ-proof of $F=G$ with syntactic-degree $d$, and outputs a \PCZ-proof of $[F]=[G]$ in which every circuit is of depth $ O(\log s\cd\log d+\log ^2 d)$ and the size of the proof is $\poly(s,d)$. 


Applying this function to $\pi_5$, we obtain a \SigOneB-definable function in \VNCTwo, which given $n$ in unary outputs a depth $O(\log^2 n)$ $\PC(\F)$-proof $\pi_6$ of the determinant identities  (\ref{eq:1}), (\ref{eq:2}), where the determinant in the two identities is replaced by the appropriate balanced division free  circuit of syntactic-degree $n$ computing the determinant, denoted \cdetbl.
Note that the theory will now express the fact that the PI-proof obtained is indeed a legitimate PI-proof.

See Section \ref{sec:balancing-algebraic-circ-in-theory} for more details.

\para{Step 8: Applying the reflection principle}
We now reason in \VNCTwo\ as follows: for every $n$ and every pair of  matrices $A, B$ over \Z\ of dimension \nbyn, by the definition of the function \Det\ in the theory, $\Det(AB)$, $\Det(A)$ and $\Det(B)$ equals the value of applying the evaluation function to  the circuit $\cdetbl$ with the input assignment  $AB, A, B$, respectively.

\mar{Check if the integer quantifiers here are open or bounded?}

By the arguments above, there exists a depth $O(\log^2 n)$ $\PCZ$-proof of $\cdetbl(XY)=\cdetbl(X)\cd \cdetbl(Y)$ for the two symbolic matrices $X,Y$ of dimension $n\times  n$. But by the reflection principle for depth $O(\log^2 n)$ $\PCZ$-proof from  Theorem \ref{thm:ovr:sound-of-arithmetic-proofs} this means that for every input matrices over \Z,  $\cdetbl(AB)=\cdetbl(A)\cd \cdetbl(B)$. We therefore  conclude   $\Det(AB)=\Det(A)\cd\Det(B)$. 
  
  The same argument  applies to the proof of the determinant identity \eqref{eq:2}, when using a symbolic triangular (lower or upper) matrix and then using Theorem \ref{thm:ovr:sound-of-arithmetic-proofs}.\iddo{}





\section{Encoding  Circuits and PI-Proofs in the Theory}\label{sec:Encoding Circuits and PI-Proofs}
Here we explain how to encode algebraic circuits and PI-proofs in the theory. Specifically, in Section \ref{sec:circuit-with-div-for-det} we describe the circuit \cdeti, namely a circuit with division for the determinant. In Section \ref{sec:vz-function-for-constructing-DET-circuit} we explain how to construct \cdeti\ in \VZ. Finally, in Section \ref{sec:P-C-inverse-proof-of-1-2} we complete Step 1  (in Section \ref{sec:carry-proof-in-theory}) in the construction of the proof in the theory: we  construct a \PIZ-proof of the determinant identities, where the determinant is written as \cdeti. The construction in the theory is syntactic in nature, and  is done  in \VZ\ using the \SigZB-\textbf{COMP} axiom.

\subsection{Encoding  Circuits}\label{sec:encoding-circuits}
In order to talk about algebraic circuits, Boolean circuits and PI-proofs in the theory we need to fix an encoding scheme for these objects. Basically, \VNCTwo\ (in fact, already \VZ) is rich enough to let us encode syntactic objects in a rather natural way. Since every uniform \ACZ\ function is definable in \VZ\ we can assume basic encoding functions to be defined in the theory. 

We show below how to construct  \cdeti\ in the theory.  Encoding and constructing PI-proofs in the theory follows similar lines, we will not always define all the encoding details explicitly when these objects are already constructible in \VZ. 

\subsubsection{Encoding of Algebraic Circuits in the Theory}\label{sec:subsec-for-encoding-circuits}

Algebraic circuits are encoded using strings in the theory as follows: (i) a string of nodes $V$ (in which we assign natural numbers to nodes; this is convenient for our encoding schemes); (ii) a string of gates $G$, where each gate is a natural number interpreted as a pair of natural numbers $(v,t)$ (using the pairing function) where $v$ is a node in $V$ and $t$ is a natural number that  expresses that the gate $v$ is either  $+,\times$ or $(\cd)^{-1}$ (plus gate, times gate or a division gate, respectively) or is the \textit{i}th input (the first two connectives are binary and the third is  unary); (iii) a two-dimensional string of input gates $I$, where, if the first (lsb) of the $i$th string is 0, the $i$th string of $I$ encodes a variable $x_j$, and otherwise it is a binary string representing an integer scalar---the index $j$ of an input variable $x_j$ is represented using the binary representation of $j$; %
\iddo{}%
and finally (iv) a string of directed edges $E$ between two nodes, where $(u,v)\in E$ means that there is an incoming edge to $v\in V$ emanating from $u\in V$.  

We show in Section \ref{sec:algebraic_circuit_value_problem} that there is a \SigOneB-definable function in \VNCTwo\ that converts a non-layered algebraic circuit into a layered circuit in \VNCTwo, where a layered circuit is a circuit in which each node belongs to a specific layer and nodes in layer $i$ have an outgoing edge only to nodes in layer $i+1$. This will enable us to convert algebraic circuits to layered Boolean circuits as  required by the evaluation axiom of \VNCTwo\ (Definition \ref{def:vnctwo}).




\subsubsection{Circuit with Division for the Determinant}\label{sec:circuit-with-div-for-det}

First we need to define the determinant circuit \emph{with division} denoted \cdeti. Similar to \cite{HT12}, this is done using Schur complement and can be viewed as performing a block Gaussian elimination: by considering the symbolic matrix $X= \{x_{ij}\}_{i,j\in [n]}$, consisting of $n^{2}$ distinct variables, defining the matrix inverse $X^{-1}$ of $X$ and then, by partitioning $X$ into blocks, we formulate a recursive definition of the determinant, using matrix inverse.

Formally, we define an $n \times n$ matrix $X^{-1}$ whose entries are circuits with divisions, computing the inverse of $X$, as follows:

\vspace{-5pt} 
\begin{enumerate}
\item If $n=1$, let $X^{-1}:=(x_{11}^{-1})$.
\item If $n>1$, write  $X$ as follows:
\begin{equation}
X= \left(
                \begin{array}{l r}
                        X_{1} & v_{1}^{t}\\
                        v_{2} & x_{nn}
                \end{array}
        \right)\,,
\label{eq: X}
\end{equation}
where $X_{1}=\{x_{ij}\}_{i,j\in [n-1]}$, $v_{1}=(x_{1n},\dots, x_{(n-1)n})$ and $v_{2}=(x_{n1},\dots, x_{n(n-1)}) $.
Assuming we have constructed $X_{1}^{-1}$, let the \emph{Schur complement} be defined as \vspace{-5pt} 
\begin{equation}\label{eq: dd}
\dd(X):= x_{nn}- v_{2}X_{1}^{-1}v_{1\\ }^{t}\,.
\end{equation}
Since $\dd(X)$ computes a single non-zero rational function, $\dd(X)^{-1}$ is well-defined.
Finally, let
\begin{equation}
X^{-1}:=
\left(\begin{array}{l r}
\!\! X_{1}^{-1}\left(
            I_{n-1}+\dd(X)^{-1}v_{1}^{t} v_{2 }X_{1}^{-1}
            \right) & -\dd(X)^{-1}X_{1}^{-1}v_{1}^{t} \\
\!\!       -\dd(X)^{-1} v_{2} X_{1}^{-1} & \dd(X)^{-1}
\end{array}
\right)\,.
\label{eq:def-inverse}
\end{equation}
\end{enumerate}

The circuit $\cdeti(X)$ is defined as follows:
\begin{enumerate}
\item If $n=1$, let $\cdeti(X):= x_{11}$.
\item If $n>1$, partition  $X$ as in (\ref{eq: X}) and let $\dd(X)$ be as in (\ref{eq: dd}).
Let 
\begin{equation}\label{eq:det-min-one-definition}
\cdeti(X):=
 \cdeti(X_{1})\cdot \dd(X)= \cdeti(X_1)\cd(x_{nn}- v_{2}X_{1}^{-1}v_{1}^{t})\,.
\end{equation}
\end{enumerate}

The definition in (\ref{eq:def-inverse}) should be understood as a circuit with $n^{2}$ outputs which takes $X_{1}^{-1}, v_{1}, v_{2}, x_{nn} $ \emph{as inputs} and moreover, such that the inputs from $X_{1}^{-1}$ \emph{occur exactly once}. Altogether, we obtain a polynomial-size circuit for $X^{-1}$ and the determinant function of $X$. The circuits obtained are unbalanced, have division gates and are of exponential syntactic-degree (see Definition \ref{def:syntactic-degree}). The fact that $ \cdeti(X) $ indeed computes the determinant (as a rational function) stems, e.g., from the fact (shown in this work, or in \cite{HT12}) that  $ \PIZ $ can prove the two identities that characterize the determinant. That $X^{-1}$ computes  matrix inverse is also proved in the theory. 
\TODO{Should we say that a matrix is in fact an abbreviation for $n^2$ circuits? Where do we use it/if at all, in the \PI-proof?}

\subsubsection{Constructing the Circuit \cdeti\ in \VZ}\label{sec:vz-function-for-constructing-DET-circuit} Here we show a \SigZB-definable in \VZ\ function, denoted $\outix(n)$, that outputs the multi-output circuit \ix\ (\eqref{eq:def-inverse} above) given as input a unary integer $n$.\iddo{} 
From $X^{-1}$, in a similar manner we can construct (the single-output circuit) $\cdeti(X)$ using \eqref{eq:det-min-one-definition} above. Note that the definition in \eqref{eq:def-inverse} is implicitly a construction that uses \SigOneB-induction (that is, the number induction axiom as in Proposition \ref{prop:number-induction} in which we use  \SigOneB\ instead of \SigZB): given that there exists a circuit for $X_1^{-1}$ of dimension $(n-1)\times(n-1)$, we construct $X^{-1}$ of dimension \nbyn. 
However, since we do not have in \VZ, nor in  \VNCTwo, the number induction axiom for $\SigOneB$-formulas we will need to construct the circuit ``syntactically'' using only the \SigZB-\textbf{COMP} axiom by utilizing a  natural encoding scheme. This idea and similar encoding is then used in the sequel to construct all the \PIZ-proof in the theory. For getting the final division free \PCZ-proofs using homogenization and balancing  we need to consider different arguments, including the axioms of \VNCTwo, e.g., to be able to compute matrix powering (see Sections \ref{sec:homogenization-of-proofs}, \ref{sec:balancing-algebraic-circ-in-theory}).  \iddo{}

The circuit for \ix\ is encoded as follows. It is a multi-output circuit. The string $V$ encodes the nodes in the circuit, as natural numbers, where a node number is interpreted as a tuple of natural numbers as shown below (using the \SigZB-definable in \VZ\ tupling number function). For each \emph{inductive level} $d=1,\dots,n$ in the inductive definition of \ix\ in \eqref{eq:def-inverse}, corresponding to the construction of a \dbyd\ inverse matrix, we  have a set of nodes $(d, (i,j), \ell)\in V$, each interpreted as a three-tuple of numbers where the second number is a pair of numbers in itself. In $(d, (i,j), \ell)\in V$, the pair $(i,j)$, for $i, j\in[d]$, is an entry in a \dbyd\ matrix, meaning that the node $(d, (i,j),\ell)$ is part of a sub-circuit of \ix\ that computes the $(i,j)$th entry in the  $d$th inductive-step; $\ell$ is the running index of the nodes in that part, where $\ell=0$ iff the node is what we consider an  \emph{output node of the given level  $d$ and the given entry $(i,j)$}. Nodes of the form $(0,(i,j),0)$ stand for the \textit{input node} corresponding to the variable $x_{ij}$ (or scalar) in the input string $I$.

For example, $(1,(1,1),0)$ is the node computing $x_{11}^{-1}$, because the first coordinate $d=1$ refers to   ``inductive" level  1 in \eqref{eq:def-inverse}, the second coordinate is $(1,1),$ meaning the $(1,1)$-entry from the circuit computing the inverse of $x_{11}$, and the last coordinate is 0, meaning this is the \emph{output }node of the inverse of $x_{11}$. Note that we use the numbers on the nodes in $V$ to denote information on the structure of the circuit, namely information about the edges in $E$ and  whether a gate is an input node (this information is expressed also in $G$). This makes the construction of the corresponding $E$ and $G$ easier.

Additionally, we have a string $G$ of natural numbers, each interpreted as a four-tuple encoding the gate-type of each node in $V$, excluding the input nodes $(0,(i,j),0)$. That is, $(d,(i,j),\ell,g)\in G$ means that node $(d,(i,j), \ell)\in V$ is of type $+$ if $g=0$, $\times$ if $g=1$ and division $(\cd)^{-1}$ if $g=2$, and an input variable $x_{ij}$ if $g=(i,j)$\TODO{don't use $\langle x,x\rangle$. Change also in definition above.}, where, again, $(\cd,\cd)$ is the pairing function 
(note that the pairing function (cf. \cite{CN10}) is monotone increasing and that $(1,1)>2 $, so we can distinguish between the case of an arithmetic gate and an input gate). Finally, the string $E$ encodes the edges between nodes in the circuit. That is, $(d,(i,j),\ell,d',(i',j'),\ell')$ means that there is a directed edge from node $(d,(i,j),\ell)$ to node $(d',(i',j'),\ell')$.

Using the above  encoding scheme it is possible now to bit-define the string function \outix\ as a \SigZB-definable function in \VZ. We  only need to construct, given some level  $d,(i,j),$  the sub-circuits whose nodes will be $(d,(i,j),\ell)$, for  some $\ell$, according to the definition in \eqref{eq:def-inverse}. We will use the following notation and functions in the theory.  

\para{Notations and basic functions for constructing sub-circuits} Let $F$ be some ``primitive'' arithmetic function,  such as inner product of two  $n$-element vectors over the integers, or one of the functions in \eqref{eq:def-inverse} used to define a minor or the matrix inverse $\ix$, such as $\dd(X)^{-1}$ (we use the term \emph{minor} to refer to a sub-matrix). We will denote by $\encd{F}(n,d,\ell,\overline I,\overline O)$ the following string function: the  input to this string function are $\overline I$, serving as the input nodes to the circuit and $\overline O$ the output nodes of the circuit for $F$, $d$ is the index ``level'' (used to record the induction-level of the inductive  circuit constructions  as in \eqref{eq:def-inverse}) and $\ell$ is the ``running index'' of a node in a given level $d$, and  $n$ stands for the ``dimension'' of the operation defined by $F$ (e.g., inner product of vectors of size $n$, or matrix product of two \nbyn\ matrices has dimension $n$). The output is a string, but we abuse notation and assume it is \emph{three} separate strings encoding the (output) circuit, for simplicity, as follows: $E,V,G$ as described above. \TODO{maybe don't need $n$?}

More formally, we define $\encd{F}(n,d,\ell,\overline I,\overline O)=(E,V,G)$ 
as follows (similar to the above notation): $V$ is a string describing the vertices in an algebraic circuit. $E$ is a string describing the edges between vertices in $V$. $G$ is a string describing the gate-types of vertices in $V$. Every vertex is of the form $(d,(i,j),\ell)$ with $d$ the recursive level in the definition of $\ix$ in \eqref{eq:def-inverse}, $(i,j)$ means that the node is in the $(i,j)$'s part of the definition of \ix, and $\ell$ is the running index of  nodes in the same level $d$ and same part $(i,j)$, where $\ell=0$ iff the node is an output node of \emph{that level} $d$ (it is not necessarily the output node of the whole circuit). Assume that $F(\overline I)$ is some algebraic  function with $m_0$ integer inputs $\overline I$ and $m_1$ integer outputs $\overline O $. Then, we supply $\encd{F}(n,d,\ell,\overline I,\overline O)$  with the node indices (as encoded in $V$) to be used as input nodes and output nodes for the (sub-)circuit computing $F$.
Here is an example of the input and output nodes of $F_1$.

\begin{example} Consider the multi-output circuit  $F_1:={X_{1}^{-1}(I_{n-1}+\dd(X)^{-1}v_{1}^{t} v_{2 }X_{1}^{-1})}$ from \eqref{eq:def-inverse}. We want to construct the circuit $F_1$ in \VZ. Note that $F_1$ is a recursive function in the sense that it uses as inputs the outputs $X_1^{-1}$ which are computed in the previous recursive level $d-1$, together with the ``new'' nodes in row $d$ and column $d$ in $X$. Therefore, the inputs of $F_1$ are the following nodes: $(d-1)^2$ input nodes for $X_1^{-1}$, $2(d-1)$ input nodes for $v_1^t$ and $v_2$, and finally one input node $x_{dd}$ (needed for computing $\dd(X)^{-1}$), which sums up to $d^2$ input nodes in total. The number of output nodes for  $F_1$ is $(d-1)^2$, as it defines a $(d-1)\times(d-1)$ minor of $\ix$. Therefore, in our encoding scheme, the input nodes for $F_1$ (viewed as a \dbyd\ matrix) are: 
\[\!\!\!\!\!
\begin{pmatrix}
(d-1,(1,1),0)  & \dots & (d-1,(1,d-1),0) & (0,(1,d),0)   \\
\vdots & \ddots &\vdots & \vdots \\
(d-1,(d-1,1),0) &  \dots & (d-1,(d-1,d-1),0) & (0,(d-1,d),0)   \\
 (0,(d,1),0)   & \dots & (0,(d,d-1),0) & (0,(d,d),0)
\end{pmatrix}   
\]
and the output nodes (viewed as a $(d-1)\times(d-1)$ matrix)  are:
\[
\begin{pmatrix}
(d,(1,1),0) & \dots & (d,(1,d-1),0)   \\
\vdots & \ddots & \vdots  \\
(d,(d-1,1),0) & \dots  & (d,(d-1,d-1),0) 
\end{pmatrix}   \,.
\]
%
\end{example}
\bigskip 

Let $F_2,F_3,F_4$ be the other three functions used in the definition of \ix\ \eqref{eq:def-inverse} (for the other simpler three minors). We shall define similarly $\encd{F_i}$ functions for these  $F_i$'s. \TODO{define them explicitly}


To  show that \outix\ is a \SigZB-definable \iddo{} function in \VZ\ we need to demonstrate how to bit-define  this function using a \SigZB-formula (see Definition \ref{def:bit-graph-of-string-func} in Section \ref{sec:Definability-in-Bounded-Arithmetic} for bit-definability). In our case we need to show how to bit-define $\encd{v\cd u}$ using a \SigZB-formula, given two $n$-element vectors of integers $v,u$  representing \textit{nodes} in the circuit. This is quite easy to do: simply output a binary tree with the appropriate plus and products nodes, and plug the input nodes $v,u$ to the leaves accordingly. We provide a proof of this in the Appendix \ref{sec:construct-binary-tree}.
 Here we denote the nodes in the circuit computing the inner-product $v\cd u$ in level $d$ using the running index: every node excluding the  output nodes of this level $d$ (which are unique  for every fixed $d$ and $(i,j)$) has a different running index $\ell>0$, namely has the tuple $(d,(i,j),\ell)$ associated with level $d$ and the $(i,j)$ entry in the matrix computed at level $d$.    


Similarly, we have \SigZB-formulas for constructing other formulas like $\encd{vA}$ and $\encd{Av^t}$, given the input nodes for an \nbyn\ matrix $A$, and the input nodes for an $n$-elements vector $v$. Also, given a node $z$ it is immediate to output a circuit computing $z^{-1}$ or $-z$, and given two matrices $A,B$ (i.e., $2n^2$ nodes) it is easy to give a \SigZB\ bit-definition of $\encd{A+B}$ in \VZ. \bigskip 

\newcommand{\hi}[1]{\ensuremath{{\rm Inp}_{F_{#1}}}}
\newcommand{\ho}[1]{\ensuremath{{\rm Out}_{F_{#1}}}}

Now that we set up the notation and the functions for constructing sub-circuits, we can bit-define with a \SigZB-formula  \outix\ in \VZ\ as follows. First, for $i=1,\dots,4$, define $\hi{i}(d)$ and $\ho{i}(d)$ to be  the string functions that output the sequence of input and output nodes of the $d$th recursive level of \ix\ for each of the $F_i$'s, respectively, as shown for $F_1$ in the example above. They are all \SigZB-definable string-functions in \VZ. The bit-definition of  \outix\ is\iddo{}
\vspace{-5pt}      
\begin{multline*}
\outix(n)(i) \equiv   \\ \exists 2\le d\le n  
\left(\exists 1\le j\le 4, \outixd\left(n,d,1,\hi{j}(d),\ho{j}(d)\right)(i)\right)\\ \lor
 \encd{x_{11}^{-1}}\left(n,1,0,((0,(1,1),0)), ((1,(1,1),0))\right)(i)\big)\,,  
\end{multline*}
where $\outixd(n,d,\ell,\overline I,\overline O)$ outputs $(E,V,G)$ encoding a (sub-)circuit that is  the $d$th inductive level of $\ix$, and $\encd{x_{11}^{-1}}\big(n,1,0,((0,(1,1),0)), ((1,(1,1),0))\big)$ is the string function that outputs the encoding of the circuit ``$x_{11}^{-1}$''.  
\medskip 

In the sequel we will be less formal about encoding  in \VZ\  circuits in the  \PIZ-proofs in the theory. 
 

\subsection{Encoding and Witnessing PI-proofs}
\label{sec:encoding-PI-proofs}\iddo{}

Recall that \PI\ is a PI-proof system with division gates and with the division axiom D added (Section \ref{sec:circ_proof_with_divison}).
Also, recall from Step 5 (Section \ref{sec:carry-proof-in-theory})\mar{is this the right step?} that we are not going to express in the theory the full correctness of \PIZ-proofs in the sense that the theory will express only the syntactic-correctness of a \PIZ-proof (see definition in Section \ref{sec:circ_proof_with_divison}). 

\PIZ- and \PCZ-proofs are encoded as a two dimensional array $S$ (that is, a string encoding an array of strings), in which the $i$th string $S^{[i]}$, also called \emph{the $i$th row of $S$}, is the $i$th equation in the proof, written as a pair of circuits with division (and where circuits encoding is done as described in Section \ref{sec:vz-function-for-constructing-DET-circuit}).
 
Furthermore, the encoding of PI-proofs will always consist of additional  \emph{\textbf{witnesses for syntactic correctness}}, as follows:
\begin{enumerate}
\item  Each row $S^{[i]}$ specifies whether it is an axiom, and if not specifying  the proof-lines from which it was derived as well as the rule by which it was derived.

\item For the four rules R1-R4, we have the following convention to witness the correctness of applying the rule: 
the encoding of the circuits $F,G,H$ and $F_1,F_2,G_1,G_2$  in the antecedent and consequence of the rules are identical, that is, with the same node numbers in their respective sets of nodes $V$. In other words, the respective strings encoding $F, G, H, F_1,F_2,G_1,G_2$ in the antecedent and consequence are identical.\label{it:pi-proof-encode-circuits-identical}

\item  For the axioms A1-A9, and the axiom D in \PI, the circuits $F,G,H$ in both sides of the equations are encoded identically, as in part  \ref{it:pi-proof-encode-circuits-identical} above. 

\item  The scalar axioms A10 is encoded as a circuit with scalar inputs as usual. Only that we will not verify their correctness, as this will not be needed.

\item  The axioms C1, C2 needs a special treatment. Consider $F_1\oplus F_2=F_1+F_2$, and let $V$ be the set of node (numbers) belonging to $F_1\oplus F_2$. Every node $u\in V $, excluding the plus at the root, occurs as two different nodes $u_1, u_2$ in $F_1+F_2$. 
To witness this rule we add a string that stores (as an array of number pairs) the mapping from the nodes of $F_1$ in $F_1+F_2$ to the nodes of $F_1$ in $F_1\oplus F_2$, and similarly for $F_2$.   
Given such a witness it is  immediate to verify (with a \SigZB-formula) that the C1 axiom is applied correctly\footnote{One can also use an \NL\ algorithm, formalizable in \VNCTwo\ (since $\NL\subseteq\NCTwo$), to verify that both sides of the axiom C1 are different representation of the same circuit. However, it will not be easy to prove for our \PIZ-proofs that they are correct with such a predicate of correctness.}.
C2 is treated similarly.  
\end{enumerate}

When we talk about a PI-proof in the theory, unless otherwise stated, we assume that the proof encoding includes its witness for syntactic correctness as  above. When we talk about \PIZ-proofs specifically, we shall say that ``the theory proves the existence of a syntactic correctness \PIZ-proof'' to mean that the proof is encoded with the witness as above, only that we emphasize that the proof and witness only ensure syntactic correctness (since division by zero may occur in such proofs).

\subsubsection{Existence of Proofs with Division for the Determinant Identities} \label{sec:P-C-inverse-proof-of-1-2}

\iddo{}

Here we complete Step 1 of the argument (Section \ref{sec:carry-proof-in-theory}), by demonstrating  that there is a \SigZB-formula with $n$ as a number parameter that  defines the \PIZ-proof of the two determinant identities \eqref{eq:1}, \eqref{eq:2}.\iddo{}  



\begin{proposition}[in \VZ]\label{prop:write_PC-1-proof}
Given a positive natural number $n$ there exists a syntactic correct $\PIZ$-proof of the determinant identities \eqref{eq:1} and \eqref{eq:2} for \nbyn\ matrices, where the determinant in \eqref{eq:1} and \eqref{eq:2} is written as the division free circuit $\cdeti(A)$, for $A$, the \nbyn\ symbolic matrix $X$ or $Y$, or their product $XY$, or a symbolic triangular matrix $Z$. 
\end{proposition}

\begin{proof}
It is enough to show that there exists a \SigZB-formula denoted \writepci, that holds iff $W$ is the syntactically correct  $\PIZ$-proof of the identities \eqref{eq:1} and \eqref{eq:2} for \nbyn\ matrices. 
We construct this \SigZB-formula defining the polynomial-size  \PIZ-proof demonstrated in \cite[Section 7.1]{HT12} for the identities \eqref{eq:1}, \eqref{eq:2}. 

Given a pair of $ n\times n$ matrices $X,Y$,  the expressions $XY= A$ in the context of \PI\ is an abbreviation of  a sequence of $n^2$ equalities between the appropriate entries.
Note however that whereas before we treated $X^{-1}$ as a single multi-output circuit, here $X^{-1}$ is a set of $n^2$ separate circuits for each of the entries in $X^{-1}$ (this is achieved simply by taking the same single multi-output circuit for $X^{-1}$ as before, and duplicating each of the output gates together with its sub-circuit).  



We exemplify our \VZ-construction with the proof of $X\cdot X^{-1}=I_{n}$ below. This can potentially be constructed by induction on $n$. However, similar to the construction of \cdeti\ in the theory (Section \ref{sec:vz-function-for-constructing-DET-circuit}), we cannot use (number) induction on \SigOneB-formulas in \VNCTwo, and thus we need to work out an encoding of the proof that can be constructed using a \SigZB-formula. 

If $n=1$, we have $x_{11}\cdot x^{-1}_{11}=x^{-1}_{11}\cdot x_{11}= 1$ which is a $\PI$ axiom.  Otherwise, let $n>1$ and $X$ be as in (\ref{eq: X}). We want to construct a polynomial-size proof of $X\cdot X^{-1}=I_{n}$ from the assumption $X_{1}\cd X_{1}^{-1}=I_{n-1}$. 

Denote $a:= \dd(X)$, and $A:=I_{n-1}+a^{-1}v_{1}^{t}
                            v_{2} X_{1}^{-1} -                             a^{-1}v_{1}^{t} v_2                                X_{1}^{-1}$, and $B:= v_2    X_{1}^{-1}+                             a^{-1}(v_2    X_{1}^{-1}v_1^t                           -x_{nn} )v_2     X_{1}^{-1}$. Taken verbatim from \cite[Proposition 7.2]{HT12}, using some rearrangements, and the definition of $a$, we have:
\begin{align*}
 X\cdot X^{{-1}} & = 
        \left(
                \begin{array}{l r} X_{1} & v_{1}^{t}\\
                                     v_{2} & x_{nn}
                \end{array}
        \right)\cdot 
        \left(
                \begin{array}{l r}                 X_{1}^{-1}(I_{n-1}+a^{-1}v_{1}^{t} v_{2 }X_{1}^{-1})     & ~~~~-a^{-1}X_{1}^{-1}v_{1}^{t} \\
    -a^{-1} v_{2} X_{1}^{-1} & a^{-1}
                \end{array}
        \right)\\
        &=\left(
                     \begin{array}
                        {l r} A& - a^{-1}v_1^t
                                + a^{-1}v_1^t                          \\
                                B                             &      ~~~~~a^{-1}(-v_2                        X_{1}^{-1}v_1^t     + x_{nn})
                        \end{array}
                \right) \\
        &=\left(
                        \begin{array}{l r}
                                I_{n-1}& 0\\
                                v_2    X_{1}^{-1}-
                                a^{-1}a v_2   
                                X_{1}^{-1}& ~~~a^{-1}a
                        \end{array}
                \right) \\
        &=  \left(
                             \begin{array}{l r}
                                I_{n-1}& 0\\
                                 0 & 1
                                \end{array}
             \right).
\end{align*}
To encode this we do the following: for every $k=1,\dots,n$, we use  \SigZB-\textbf{COMP} to define the above \PIZ-proof sequence, where $k$ replaces $n$, namely, we construct the above \PIZ-sequence for  $X\cd X^{-1}$ of dimension $k\times k$ for each $k=1,\dots,n$. Such a sequence can be \SigZB-defined similar to the encoding shown in Section \ref{sec:vz-function-for-constructing-DET-circuit}: since we simply need to compose primitive constructions of matrix multiplications (which are written as $k^2$ separate equations for each entry), dot products, plus and minus, construction of the identity matrix of dimension $k$, and $X_1^{-1}$ (which we encoded explicitly in Section \ref{sec:vz-function-for-constructing-DET-circuit}). The proof sequence for $k$ uses the proof of $X_1\cd X_1^{-1}=I$, where $X_1$ has dimension $(k-1)\times (k-1)$, and this is done  by specifying $X_1\cd X_1^{-1}=I$ as a previous proof-line from which we derive our new proof-line (in particular, no (string-)induction is needed for this). 


\iddo{}
Having \PIZ-proofs of $X\cdot X^{-1}=I_{n}$ and $X^{-1}\cdot X= I_{n}$, we can now proceed to construct the \PIZ-proofs of \eqref{eq:1}, \eqref{eq:2} in the theory. This is done in exactly the same manner, by formalizing directly in \VZ\ the \PIZ-proof shown in \cite[Section~7.1]{HT12}.
However, for constructing in the theory the \PIZ-proof of \eqref{eq:2} in \cite[Section~7.1]{HT12} we make sure that the diagonal elements of the triangular matrix $Z$ are in fact  \emph{variables} and not scalars.  This is needed for us to be able to eliminate division gates later: to eliminate division gates we need every division gate to be provably good under the identity matrix assignment (see Section \ref{sec:division-elim-from-proofs}), and specifically, these gates cannot compute the zero polynomial. When the diagonal elements are not definable, namely contain division gates that compute the zero polynomial, we will not be able to establish that they are provably good.      
\iddo{} 
\end{proof}


\section{From a Rational Function to the Determinant as a Polynomial}\label{sec:From-Rational-Function-to-the-Determinant-as-a-Polynomial}

Here we complete Step 2  (Section \ref{sec:carry-proof-in-theory}) in the construction of the proof in the theory: we  construct a \PIZ-proof of the determinant identities \eqref{eq:1}, \eqref{eq:2}, where the determinant is written as \taydet, which is  a polynomial-size algebraic circuit without division gates for the determinant (of exponential syntactic-degree). We first need rst to provide the preliminaries for division elimination and truncated Taylor expansions that we shall use in this and the next sections.  

\subsection{Preliminaries for Division Elimination}\label{sec:prelim-div-elim}
\newcommand{\pow}{{\rm pow}}
%
Let $F$ be a \emph{division-free} circuit and let $F^{(0)}$ be the syntactic homogeneous division free circuit that computes the constant term of $F$ at the point $\overline  0$;  we show that such a syntactic homogeneous circuit can be constructed in \VZ\ in Section \ref{sec:homogenization-of-proofs}. We  define the  circuit $\inv_k(F)$ that will serve as an inverse of $F$ modulo high degree monomials, in the sense that 
\begin{equation}\label{eq:F-times-invF}
F\cd\inv_k(F)=1+\hbox{[\begin{small}monomials of degree greater than $k$\end{small}]}\,.
\end{equation}  
Note that because we work over \Z\ the only way for \eqref{eq:F-times-invF} to hold (for $k>1$) is when $\widehat {F^{(0)}}=1$, since if the constant term in $F$ is not 1, no product of $F$ can compute the constant term 1. In general, the division elimination (as in Strassen \cite{Str73}) needs to work over a field and have an inverse element for $\widehat {F^{(0)}}$. 
But in our application inside the theory we will always have $F^{(0)}$ as a circuit without division, possibly with variables, that computes the polynomial  1, and hence we do not require an inverse of $\widehat {F^{(0)}}$. 


Assume that $F$ is good under the zero assignment, in symbols $\widehat {F^{(0)}}=1$, and define the circuit $\inv_{k}(F)$ as 
\begin{align*}
        \inv_{k}(F):= \left(1+ (1-F)+ (1-F)^{2}+\dots+ (1-F)^{k}\right), \end{align*}
where powers like $(1-F)^k$ are abbreviation of $(1-F)\cdots (1-F) $, $k$  times (written for instance as a logarithmic in $k$ depth circuit; since $k$ will always be polynomial in our application this will be enough for our purposes).
\mar{Put this abbreviation in preliminiaries or somewhere else.}%
%
%
The following  lemma demonstrates that $\inv_k(F) $ can \emph{provably} serve as the inverse polynomial of $ F $ ``up to the $ k $-th degree''.

\begin{lemma}[in \VZ]\label{lem:inv}
Let $F$ be a size $ s $ circuit without division, let $\eta$ be a \PCZ-proof of $F^{(0)}=1$,  
\iddo{}
and let $k\ge 1$ be a natural number. Then there exists   $\PCZ$-proofs of size $s\cdot \poly(k)$ of the following equations, in which every node in every circuit in the proofs appears with its syntactic-degree upper bound (see Definition \ref{def:appears-with-syntactic-degree-upper-bound}):
\begin{eqnarray}
(F\cdot \inv_{k}(F) )^{(0)}&=& 1 \label{eq:FtimesInvF-zero}\\
(F\cdot \inv_{k} F )^{(i)}&=& 0, \,\, ~~\hbox{for}\,\, 1\leq i\leq k\,.\label{eq:FtimesInvF-k}
\end{eqnarray}
\end{lemma}

\begin{proof}
\iddo{}

Denote $a=F^{(0)} $. 
We construct the following simple  \PCZ-proof sequences. We have $a(1-(1-F))=aF$, and by assumption that $a=1$ has a \PCZ-proof $\eta$, we get a \PCZ-proof of $F=(1-(1-F))$. 

By definition $\inv_{k}(F)= (1+(1-F)+(1-F)^{2}+\dots +(1-F)^{k})$.
By elementary rearrangement, we prove in $\PCZ$:
\begin{align}\notag
F\cdot \inv_{k}(F) & = (1-(1-F))\cd \left(1+(1-F)+(1-F)^{2}+\dots+ (1-F)^{k}\right) \\\notag
& = 1+(1-F)+\dots +(1-F)^{k}\\\notag 
& ~~~~~~~~~~~~~~~~~~-(1-F)\cd\left(1+(1-F)+(1-F)^2+\dots+ (1-F)^{k}\right) \\\notag
& = 1+(1-F)+\dots+ (1-F)^{k}\\\notag 
& ~~~~~~~~~~~~~~~~~~~~~~~~~~~~~~~~~~~~~~~~~~~~-(1-F)-(1-F)^2-\dots- (1-F)^{k+1} \\
& = 1-(1-F)^{k+1}\,. \label{eq:cancel-this}
\end{align} 

From \eqref{eq:cancel-this} and Theorem \ref{thm:homogenize-proofs} in Section \ref{sec:homogenization-of-proofs}, we construct the \PCZ-proof of 
$$
(F\cdot \inv_{k}(F))^{(0)}= 1 - ((1-F)^{k+1})^{(0)},
$$
wherein every node in every circuit in the proof appears with its syntactic-degree upper bound.
Specifically, from Theorem \ref{thm:homogenize-proofs} we have \PCZ-proofs of   $\left((1-F)^{k+1}\right)^{(0)}=\left((1-F)^{(0)}\right)^{k+1} = (1-F^{(0)})^{k+1}$, and using $\eta$ we have that the rightmost term is  $(1-1)^{k+1}=0$, and we conclude \eqref{eq:FtimesInvF-zero}.

To conclude \eqref{eq:FtimesInvF-k} we proceed as follows.  From \eqref{eq:cancel-this} and Theorem \ref{thm:homogenize-proofs}  we  construct \PCZ-proofs  $(F\cdot \inv_{k}(F))^{(i)}=  (1-(1-F)^{k+1})^{(i)}$, for all $1\le i\le k$, with all nodes appear with their syntactic-degree upper bounds. From Lemma \ref{lem:bounded degree} in Section \ref{sec:homogenization-of-proofs}, we prove $(1-(1-F)^{k+1})^{(i)}=\hci{1}-\hci{((1-F)^{k+1})}=0-\hci{((1-F)^{k+1})}$. 

 To construct the \PCZ-proof of $((1-F)^{k+1})^{(i)}=0$, for $1\le i\le k$, in the theory we use again Lemma \ref{lem:bounded degree}. We omit the details. (Note   that since $F^{(0)}=1$ by assumption, we have $(1-F)^{(0)}=0$, meaning that all monomials in $(1-F)$ are of positive total degree. Therefore $(1-F)^{k+1}$ can only have monomials of degree greater than $k$, and so $\hci{((1-F)^{k+1})}=0$ is a true identity for all $i\le k$.)\iddo{}
\end{proof}

\subsubsection{Extracting the Numerators and Denominators of Circuits with Division}\label{sec:NUM-DEN-definition}
We also need to show how to extract the denominator and numerator of circuits with divisions. For every node $v$ in a circuit $F$ with division we introduce two nodes $\Den(v)$ and $\Num(v)$ that  will compute \emph{as polynomials} (that is, they will be circuits with no division) the numerator and denominator of the rational function computed by $v$, respectively, as follows:

\begin{enumerate}
\item If $v$ is an input node of $F$, let $\Num(v):= v$ and $\Den(v):= 1$.
\item If $v=u^{-1}$, let $\Num (v):= \Den(u)$ and $\Den(v):=\Num(u)$.
\item If $v=u_{1}\cdot u_{2}$, let $\Num(v):=\Num(u_{1})\cdot \Num(u_{2})$ and $\Den(v):=\Den(u_{1})\cdot \Den(u_{2})$.
\item If $v=u_{1}+ u_{2}$, let $\Num(v):=\Num(u_{1})\cdot \Den(u_{2})+ \Num(u_{2})\cdot \Den(u_{1})$ and $\Den(v):=\Den(u_{1})\cdot \Den(u_{2})$.
\end{enumerate}

Let $\Num(F)$ and $\Den(F)$ be the circuits with the output node $\Num(w)$ and $\Den(w)$, respectively, where $w$ is the output node of $F$. In Section \ref{sec:bring-div-to-top} we show that given a circuit with division $F$, there is a  \SigZB-definable function in \VZ\ that constructs the circuit $\Num(F)\cd(\Den(F))^{-1}$.

\subsection{A PI-Proof Reducing the Determinant from Rational Function to  Polynomial}
\label{sec:PI-proof-from-det-div-to-Taylor}

Like \cite{HT12}, in order to write the determinant as a polynomial instead of a rational function, we will write the determinant as the coefficient of $z^n$ in the Taylor expansion of a certain circuit with division. The coefficients of a Taylor expansion are defined as  follows:

\begin{definition}[Taylor expansion]\label{def:taylor}
Let $F=F(\overline{x},z)$ be a circuit with division. Define $\Coef_{z^{k}}(F)$ as a circuit
in the variables $\overline {x}$, computing the coefficient of $z^{k}$ in $F$, when $F$ is written as a Taylor power series at $ z=0$, in the following way:
\HalfSpace

\case 1 Assume  that no division gate in $F$ contains the variable $z$. Then we define $\Coef_{z^{k}}(F)$ by induction on the structure of $F$ as follows:
\begin{enumerate}
\item $\Coef_{z}(z):=1$ and $\Coef_{z^{k}}(z):=0$, if $k>1$.
\item If $F$ does not contain $z$, then $\Coef_{z^{0}}(F):=F$ and $\Coef_{z^{k}}(F):=0$, for $k>0$.
\item $\Coef_{z^{k}}(F+G):= \Coef_{z^{k}}(F)+\Coef_{z^{k}}(G)$.
\item $\Coef_{z^{k}}(F\cdot G):= \sum_{i=0}^k\Coef_{z^{i}}(F)\cd\Coef_{z^{k-i}}(G)$.
\end{enumerate}

\case 2 Assume  that $z$ occurs in the scope of some division gate in $F$. We let
$F_0$ be the denominator of the rational function computed by $F$ when $z=0$:\mar{make sure substitution notation is defined}\iddo{}
\[F_{0}:=\big(\Den(F)\big)(z/0).\]
\TODO{Here we drop the non redundancy condition---need t overify this.}
Note that $F_0$ is not necessarily a constant, as it may contain variables different from $z$. If $\wh{F_{0}}=0$ then $\Coef$ is \emph{undefined}.
Assume that  $\wh{F_{0}}\not=0$ and denote by $G$ the circuit $(1-F_{0}^{-1}\cd\Den(F))$. We let
\begin{equation}\label{eq:def-of-Coef}
\Coef_{z^{k}}(F):= 
F_{0}^{-1} \cdot \Coef_{z^{k}} \left(\Num(F)\cd
\left(1+G+G^2+\dots+G^k \right)\right).
\end{equation}
\end{definition}


Note that $z$ does not occur in any division gate inside $\Num(F)\cd
\left(1+G+G^2+\dots+G^k \right)$, and so $ \Coef_{z^{k}}(F) $ is well-defined.%
\iddo{}
\iddo{} 
In our applications, when using  $ \Coef_{z^{k}}(F)$  we will need to make sure that  $\widehat F_0=1$, and that we can prove in \PIZ\ that $F_0^{-1}=1$.
\iddo{}



\medskip

The following are the main properties of $\Coef_{z^{k}}$ that we  can use already in \VZ\ (similar to \cite{HT12}). 

\begin{lemma}[in \VZ]\label{lem:Taylor-basic-properties}
\begin{enumerate}


\item If $F_{0},\dots, F_{k}$ are circuits \emph{with} division not containing the variable $z$, then $\Coef_{z^{j}}\left(\sum_{i=0}^{k}F_{i} z^{i}\right)= F_{j}$ has a syntactically correct  $\PIZ$-proof, for each $j\leq k$.
\iddo{}

\item
Assume that $F,G$ are circuits \emph{with division} such that $F=G$ has a syntactically correct $\PIZ$-proof $\pi$ of size $s$. Then, $\Coef_{z^k}(F)= \Coef_{z^{k}}(G)$
has a syntactically correct $\PIZ$-proof of size $s\cdot \poly(k)$, for every natural number $k$.
\iddo{} 

\item Let  $F$ be a circuit \emph{without} division, together with a witness for  the syntactic-degree of all nodes in $F$, where $\deg(F)=d$. Then ${F= \sum_{i=0}^d \Coef_{z^{i}}(F)\cdot  z^{i}}$ has a $\PCZ$-proof. \footnote{The only place where we need this part is the  proofs of  the Cayley-Hamilton theorem in Section \ref{sec:corollaries}. For witnesses for syntactic-degree see the appendix Section \ref{sec:witnessing-syntactic-degrees}.}
\end{enumerate}

In both parts 1, 2 above the inverse elements of an integer $a$ is written as $a^{-1}$ (i.e., with an explicit division gate. This will not be a problem for us since in our case \VZ\ can prove that $a=1$).\iddo{}\mar{make sure all syntactically correct \PI-proofs have been put in appropriate places} \iddo{}
\end{lemma}
\begin{proof}
\iddo{}
The proofs  of parts 1, 2 are  almost identical to the proof of Theorem \ref{thm:homogenize-proofs} in Section \ref{sec:homogenization-of-proofs} for eliminating high syntactic-degrees (namely, homogenizing \PCZ-proofs), and we omit the details.
\iddo{}  The proof of part  3 is given  in the appendix (Lemma \ref{lem:Annex-F=sum-F-k}). 
\end{proof}

%

%
%
%
%

\para{The determinant as a polynomial}We are now ready to define the circuit computing the determinant as a polynomial. Let 
$$\taydet(X):=\Coef_{z^n}(\cdeti(I_n+zX)).$$
As explained \iddo{}in Section \ref{sec:carry-proof-in-theory}, $\taydet(X)$ is a circuit of polynomial-size in $n$ that computes the determinant polynomial. This is because every variable from $X$ in the circuit $\cdeti(I_{n}+zX)$ occurs in a product with $z$, and thus $\Coef_{z^{n}}(\cdeti(I_{n}+zX))$ computes the  $n$th homogeneous part of the determinant of $I_{n}+X$, which is simply the determinant of $X$.
By the definition of $\Coef_{z^{n}}$, the circuit $\Coef_{z^{n}}(\cdeti(I+zX))$ \emph{contains exactly one inverse gate}, namely the inverse of $\Den(\cdeti(I_{n}+zX))$ at the point $z=0$.  $\Den(\cdeti(I_{n}+zX))(z/0)$ has an exponential syntactic-degree  which will be dealt with  below in Section \ref{sec:Reducing-the-Syntactic-Degree-of-the-Determinant-Polynomial}.  

We can now use Lemma \ref{lem:Taylor-basic-properties} to construct the desired \PIZ-proof:

\begin{lemma}[in \VZ]\label{lem:PI-proof-DET-taylor-equals-DET-circ-div} Let $A$ be an  \nbyn\ symbolic matrix $X$ of distinct variables or the product of two symbolic matrices $XY$ or a triangular \nbyn\ symbolic matrix $Z$. Then there exists a polynomial-size $\PIZ$-proof of $\taydet(A)=\cdeti(A)$.
\end{lemma}
\begin{proof}
Using Lemma \ref{lem:Taylor-basic-properties} parts 1 and 2 we construct in \VZ\ the \PIZ-proof demonstrated in \cite[Proposition 7.9]{HT12} (see also Lemma \ref{lem:Taylor-good-assignment}).
(Since we do not use part 3  of Lemma  \ref{lem:Taylor-basic-properties}
we do not need to construct a witness for the syntactic-degrees of any circuit in this case.)
\iddo{}
\iddo{}

\end{proof}

%
%

%
%
%
%
%
%
%
%

\subsection{Reducing the Syntactic-Degree of the Determinant Polynomial}\label{sec:Reducing-the-Syntactic-Degree-of-the-Determinant-Polynomial}

We need to reduce the syntactic-degree of $\taydet(X):=\Coef_{z^n}(\cdeti(I_n+zX))$, which is exponential as we now explain. Similar reasoning applies to $\taydet(A)$, where $A$ is an \nbyn\ symbolic matrix $X$ of distinct variables or the product of two symbolic matrices $XY$ or a triangular \nbyn\ symbolic matrix $Z$.

Let $F(X)=F(x_1,,\dots,x_n)$ be a circuit with division in the displayed input variables. Assume that every input variable is now multiplied by a new variable $z$, to get $F(zx_1,\dots,zx_n)$, which we denote by $F'$. Consider $\Coef_{z^{k}}(F'):= F'^{-1}_0 \cdot \Coef_{z^{k}} \left(\Num(F')\cd \left(1+G+G^2+\dots+G^k \right)\right)$ for $G:=(1-F'^{-1}_{0}\cd\Den(F'))$, and $F'_{0}:=\big(\Den(F')\big)(z/0)$. By induction on the size of  $F'$, inspecting the construction of $\Coef_{z^k}(\cd)$,  it is easy to show that every subcircuit in $\Coef_{z^{k}}(F')$ has a syntactic-degree at most $k$, excluding occurrences of the  subcircuit $F'^{-1}_{0}$ that has syntactic-degree $\deg(\Den(F'))$ that  may be greater than $k$.\iddo{}

Accordingly, $\Coef_{z^n}(\cdeti(I_n+zX))$  contains occurrences of $\Den\left(\cdeti\left(I_n+zX\right)\right)(z/0)$, that are of syntactic-degree greater than $n$. To remedy this we define $\taydetsharp(X)$ as the circuit $\taydet(X)$ in which we replace the subcircuit  $\Den\left(\cdeti\left(I_n+zX\right)\right)(z/0)$ by the constant 1 (note that indeed $\Den\left(\cdeti\left(I_n+zX\right)\right)(z/0)$ computes the polynomial 1):
\begin{multline}
\label{eq:huge-term-for-detsharp-spelled-out}
\taydetsharp(X):=1 \cdot \Coef_{z^{n}} \Big(\Num(\cdeti(I_n+zX))\cd \\ \Big(1+(1-1\cd\Den(\cdeti(I_n+zX)))+(1-1\cd\Den(\cdeti(I_n+zX)))^2+ \\
\dots+\left(1-1\cd\Den(\cdeti(I_n+zX))\right)^n \Big)\Big)
\end{multline}

\begin{lemma}[in \VZ]\label{lem:reduce-syntactic-degree-Taylor}
Let $n$ be a positive natural number and $A$ be an \nbyn\ symbolic matrix $X$ of distinct variables, or the product of two symbolic matrices $XY$, or a triangular \nbyn\ symbolic matrix $Z$.\iddo{} Then, there exists a syntactically correct  $\PIZ$-proof of $\taydetsharp(A)=\taydet(A)$.
\end{lemma}

\begin{proof}
We assume that $A=X$. The other cases are similar. 

By previous constructions we can construct the circuits $\taydetsharp(X)$ and $\taydet(X)$ in \VZ.

We first construct a \PIZ-proof of $\Den\left(\cdeti\left(I_n+zX\right)\right)(z/0)=1$. We use Theorem \ref{thm:Den-Num-normalization-proofs} in Section \ref{sec:bring-div-to-top} that will be proved in the sequel. Initially,  construct directly a \PIZ-proof of $(\cdeti\left(I_n+zX\right))(z/0)= \cdeti\left(I_n\right)$, by substituting $z$ by 0 and gradually replacing in the proof $0\cd u$ to $0$ and $0+u$ to $u$, for $u$ any subcircuit. This is done  using \SigZB-\textbf{COMP} and our encoding scheme in Section \ref{sec:encoding-PI-proofs}.\footnote{Note that this is not done for an \emph{arbitrary} circuit, namely, we do not know of an \NCTwo\ algorithm that receives a circuit $C$, such that $\widehat  C\neq 0$, and discards in such a way every $0$ constant in the circuit. We only build a \PIZ-proof that witnesses such a gradual   procedure for discarding 0's from the \emph{specific} circuit $(\cdeti\left(I_n+zX\right))(z/0)$.} 

Now, construct a \PIZ-proof of $\cdeti\left(I_n\right)=1$. This follows from the proof of Lemma \ref{lem:L01}. Another way to construct this \PIZ-proof is as follows: using Proposition \ref{prop:write_PC-1-proof} for constructing  the proof of  $\cdeti(Z)=z_{11}\cdots z_{nn}$, when $Z$ is a symbolic triangular matrix: we can construct a syntactically-correct \PIZ-proof of $\cdeti\left(I_n\right)=1$.\iddo{} Now using  Theorem \ref{thm:Den-Num-normalization-proofs} part (ii), we construct the proof of $\Den\left(\cdeti\left(I_n+zX\right)\right)(z/0)=1$.

Using the  \PIZ-proof of $\Den\left(\cdeti\left(I_n+zX\right)\right)(z/0)=1$ and the \SigZB-\textbf{COMP} axiom we can show the existence of a $\PI(Z)$-proof of  $\taydetsharp(X)=\taydet(X)$. This is done  by constructing a  $\PIZ$-proof that gradually substitutes each occurrence of the subcircuit $\Den\left(\cdeti\left(I_n+zX\right)\right)(z/0)$ in $\taydet(X)$ by the constant 1.\iddo{} 
\end{proof}

From Proposition \ref{prop:write_PC-1-proof} and Lemma \ref{lem:reduce-syntactic-degree-Taylor} we get: 

\begin{corollary}[in \VZ]\label{cor:exist-PI-div-proof-of-det-identities-for-dettaysharp}
Given a positive natural number $n$ there exists a syntactically correct $\PIZ$-proof of the determinant identities \eqref{eq:1} and \eqref{eq:2} for \nbyn\ matrices, where the determinant  in  \eqref{eq:1} and \eqref{eq:2} is written as the division free circuit $\taydetsharp(A)$, for $A$, the \nbyn\ symbolic matrix $X$ or $Y$, or their product $XY$, or a symbolic triangular matrix $Z$. 
\end{corollary}
\iddo{}

\section{Bringing Division Gates to the Top}\label{sec:bring-div-to-top}

Here  we show a \SigZB-definable function in \VZ\ that receives an algebraic circuit $F$ with division gates and \emph{normalizes} it as shown in Section \ref{sec:NUM-DEN-definition}, that is, converts it into a circuit with a single division gate  ``at the top'', of the form $\Num(F)\cd \Den(F)^{-1}$ (formally, the output gate is a product gate with one of its child being a division gate). We need to normalize circuits in such a way in order to be able to eliminate division gates: as seen in  the next section, we  replace the gate $G^{-1}$ with $\inv_k(G)$; but for $\inv_k(G)$ to be defined we need $G$ to be a division free circuit (and can be guaranteed only if there is only one division gate at the top). 



%
%

 We wish to show the following:

\begin{theorem}[in \VZ]
\label{thm:Den-Num-normalization-proofs}
(i) If $F$ is a circuit with division, then $F=\Num(F) \cd \Den(F)^{-1}$ has a $\PI(\F)$-proof. (ii) Let $F,G$ be circuits with division. Assume that $F=G$ has a $\PI(\F)$-proof. Then $\Num(F)\cd\Den(F)^{-1}= \Num(G)\cd\Den(G)^{-1}$ has a syntactically correct  $\PI(\F)$-proof such that every division gate in every circuit in the proof occurs only at the top.  
\end{theorem}

To prove this we first argue that division normalization is doable in \ACZ. This follows the similar reasoning as shown in Section \ref{sec:vz-function-for-constructing-DET-circuit}: although the division normalization procedure as described in Section \ref{sec:NUM-DEN-definition} is defined by induction on the size of the circuit (and thus implicitly needs \SigOneB-induction) it is uniform enough to be defined \SigZB-formula (using \SigZB-\textbf{COMP} axiom) as follows. 

Every internal node $u$ in the input circuit $F$ is duplicated into two copies denoted $\Num(u)$ and $\Den(u)$. We then construct (in parallel) the division normalized circuit by wiring the nodes $\Num(u)$ and $\Den(u)$ for each $u$ in $F$ as in the original definition:  (\textbf{i}) If $v$ is an input node of $F$, let $\Num(v):= v$ and $\Den(v):= 1$; (\textbf{ii}) if $u=v+w$ we construct $\Num(u):=\Num(v)\cd\Den(w)+\Num(w)\cd \Den(v)$ and $\Den(u):=\Den(v)\cd\Den(w)$;
(\textbf{iii}) if $v=u^{-1}$, let $\Num (v):= \Den(u)$ and $\Den(v):=\Num(u)$; and finally (\textbf{iv}) if $v=u_{1}\cdot u_{2}$, let $\Num(v):=\Num(v_{1})\cdot \Num(v_{2})$ and $\Den(v):=\Den(v_{1})\cdot \Den(v_{2})$.

Since the nodes $\Num(u)$ and $\Den(u)$ are known in advance for every node $u$, this construction is doable in parallel, and specifically in \FACZ.        






The proof of Theorem \ref{thm:Den-Num-normalization-proofs} is similar to  the proof of Theorem \ref{thm:homogenize-proofs} that is presented in details in  Section \ref{sec:homogenization-of-proofs}. We thus omit the details. 




From Corollary \ref{cor:exist-PI-div-proof-of-det-identities-for-dettaysharp} and Theorem \ref{thm:Den-Num-normalization-proofs} we get:

\begin{corollary}[in \VZ]\label{cor:exist-normalized-div-PI-div-proof-of-det-identities}
Given a positive natural number $n$ there exists a syntactically correct $\PIZ$-proof of the determinant identities \eqref{eq:1} and \eqref{eq:2} for \nbyn\ matrices, where the determinant  in  \eqref{eq:1} and \eqref{eq:2} is written as the division free circuit  
$\taydetsharp(A)$, for $A$, the \nbyn\ symbolic matrix $X$ or $Y$, or their product $XY$, or a symbolic triangular matrix $Z$
in which every circuit with division has only a single division gate at the top.\iddo{}
\end{corollary}

\newcommand{\shift}{\ensuremath{\rst\sigma}}

\section{Eliminating Division Gates}\label{sec:division-elim-from-proofs}
In this section  we show in detail step 5 (in Section \ref{sec:carry-proof-in-theory}) which eliminates in the theory  division gates from the \PIZ-proof we obtained in Corollary \ref{cor:exist-normalized-div-PI-div-proof-of-det-identities}. 

A \PIZ-proof is said to be \emph{division axiom free} if it does not use the axiom D of division: $F\cd F^{-1}=1$. The \emph{identity matrices assignment $\rho$ i}s defined to be the assignment of 0 and 1 elements to the variables in \eqref{eq:1}, \eqref{eq:2} such that $X=Y=Z=I_n$. We say that a division gate $u$ is \emph{\textbf{provably good under an assignment $\rho$}} to its variables whenever $u\rst\rho=1$ has a division axiom free \PIZ-proof. In this case we  also say that \emph{$\rho$ is provably good for $u$}. Accordingly, if $\rho$ is provably good for all the division gates in a circuit $C$ we say that \emph{$\rho$ is provably good for $C$}. 

To eliminate division gates we first  make sure that every division gate $u^{-1}$ that appears in the $\PIZ$-proof is provably good under $\rho$.  
This will allow us to carry out the proof in the theory as follows: %
 by induction on circuit size it is easy to see that if every division gate in a circuit $F$ is good under $\rho$ then  $\Den(F)$ (Section \ref{sec:NUM-DEN-definition}) is good under $\rho$ as well (however, we cannot use such \SigOneB-induction in the theory to conclude this fact; see below).
We then proceed as follows. After normalizing division gates (Section \ref{sec:bring-div-to-top}) in   the \PIZ-proof obtained in Corollary \ref{cor:exist-normalized-div-PI-div-proof-of-det-identities}, every circuit with division in the \PIZ-proof has a single division gate and is of the form $\Num(F)\cd\Den(F)^{-1}$. We first linearly shift the variables  by $\rho$, that is, replace every variable $r_{i}$ in the proof by $\rho(r_i)-w_i$ (the \PIZ-proof is still correct under this shift). Denote the shift $r_i \mapsto \rho(r_i)-w_i$ by $\sigma$. Then we replace every subcircuit $\Den(F\shift)^{-1}$ in the shifted proof  with $\inv_k(\Den(F\shift))$, for a suitable $k$. 
We now need to simulate the axiom $F\cd F^{-1}=1$ that was replaced by $F\shift\cd\inv_k (F\shift)=1$ in order to obtain a correct \PCZ-proof. This is done by using in Section \ref{sec:homogenization-of-proofs}  Lemma \ref{lem:inv}; for this lemma we need to  show that $\hc{(\Den(F\shift))}{0}=1$ is \PCZ-provable (note that $\hc{(\Den(F\shift))}{0}$ computes precisely the value of $\Den(F)$ under the identity matrices assignment, and that since   $\Den(F)$ is good under $\rho$, as we mentioned above, ${\hc{(\Den(F\shift))}{0}}$ computes the polynomial 1).

The construction in \VZ\ of the   \PCZ-proofs of $\hc{(\Den(F\shift))}{0}=1$, for every  $F$ as above, is done as follows. As a first attempt, observe that  we can construct this \PCZ-proof by induction on the size of $F\shift$: for every gate $u$ in $F\shift$ we construct the proof of $\hc{\Den(u)}{0}=1$, using the fact that every division gate in $F\shift$ evaluates to 1 under the zero assignment (this is because $\sigma$ shifts the variables to $\rho$, and the fact that  every division gate is good under $\rho$). However, since we do not have the $\SigOneB$-induction axiom in \VZ, in Lemma \ref{lem:from-all-Den-gates-proofs-to-top-Den} we will use a parallel construction of the \PCZ-proof, as follows: using the \SigZB-\textbf{COMP} axiom, for every gate $u$ in $F\shift$, construct a \PCZ-proof of $\Den(u)^{(0)}=1$. This is done simultaneously  for all gates $u$, using pointers to ``previous'' proof-lines and direct constructions of \PIZ-proofs of $\hcz{u}=1$ that do not use the division axiom D, for every gate $u$ in $F\shift$.  

\iddo{} 


\subsection{Identity Matrices are Provably Good Assignments}

Here we explicitly inspect all the division gates in our \PIZ-proofs, making sure that they are all provably good under  $\rho$. We first make sure that this is sufficient for our purpose of eliminating division (that is, applying Lemma \ref{lem:inv}).


\begin{lemma}[in \VZ]\label{lem:from-all-Den-gates-proofs-to-top-Den}
Let $F$ be a circuit with division and assume that for every division gate $u^{-1}$ in $F$ there exists a division axiom free  $\PIZ$-proof of $u\rst\rho=1$. Then, there exists a \PCZ-proof of $\hcz{\Den(F\shift)}=1$.
\end{lemma}
\begin{proof}
Simultaneously, for every node $v$ in $F$ we construct  a $\PCZ$-proof of ${\Den(v\rst\rho)}=1 $, by which we conclude that $\Den(F\rst \rho)=1$ is provable in \PCZ. 
Before showing this construction we show how to get a \PCZ-proof of $\hcz{\Den(F\shift)}=1$ from a \PCZ-proof of $\Den(F\rst \rho)=1$. 

By the definition of $\sigma$ and by basic rearrangements in \PCZ\ we have a \PCZ-proof of $\Den(F\rst\rho) = (\Den(F\shift))\rst \overline 0$ (where $C\rst\overline  0$ means substituting $0$ for each variable in $C$). By the construction of syntactic-homogeneous circuits (Section \ref{sec:homogenization-of-proofs}) the underlying graph of a (division free) circuit $C\rst 0$ and the structure of $\hcz{C}$ are \emph{identical}: constant gates stays the same; variables $x_i$ turns into  0 nodes; plus gates $w=v_1+v_2$ add the zero copies, $[v_1,0], [v_2,0]$  of $v_1,v_2$, respectively; and product gates $w=v_1\cd v_2$ multiply the zero-copies $[v_1,0], [v_2,0]$  of $v_1,v_2$, respectively (in particular, no plus gates are added to the circuit). Hence, $(\Den(F\shift))\rst \overline 0=\hcz{(\Den(F\shift))}$ is \PCZ-provable.\smallskip 

We now come back to the construction of a $\PCZ$-proof of ${\Den(v\rst\rho)}=1$, for every node $v$ in $F$. This is done by cases as follows.

\case 1 $v$ is an input node. Hence $\Den(v)=1$ by definition of $\Den$.

\case 2 $v=w_1\circ w_2$, for $\circ\in\{+,\cd\} $, then $\Den(v\rst\rho):=\Den(w_1\rst\rho)\cd\Den(w_2\rst\rho)$. Thus we construct the proof of $\Den(w_1\rst\rho)\cd\Den(w_2\rst\rho)=1$ by pointing to the proofs of $\Den(w_1\rst\rho)=1$ and $\Den(w_2\rst\rho)$ and using basic rules of \PCZ\ such as $1\cd 1=1$.

\case 3 $v=u^{-1}$. This is the relatively more difficult case. We need to use the following claim:

\begin{claim}[in \VZ]\label{cla:special-den-eq-num}
If there exists a \PIZ-proof of $u\rst\rho=1$ in which we do not use the axiom $D$, then there exists a \PCZ-proof of $\Num(u\rst\rho)=\Den(u\rst\rho)$. 
\end{claim}

\begin{proofclaim} This is proved in a similar way to the division gates normalization Theorem \ref{thm:Den-Num-normalization-proofs}.  We omit the details.
\end{proofclaim}\iddo{}

By definition $\Den(v)=\Num(u)$.  Since we have by assumption a \PIZ-proof of $u\rst\rho=1$, by the claim we have a \PCZ-proof of $\Num(u\rst\rho)=\Den(u\rst\rho)$, and from that, using pointers to the \PCZ-proof of $\Den(u\rst\rho)=1$, we get a  \PCZ-proof of $\Den(v\rst\rho)=1$.    
\end{proof}

We now need to inspect and provide appropriate division axiom free \PIZ-proofs for all division gates in the  \PIZ-proof  of the determinant identities, denoted  $\pi_3$ (according to the notation in Section \ref{sec:carry-proof-in-theory}). The proof $\pi_3$ consists of: (\textbf{i}) the proof obtained from Proposition \ref{prop:write_PC-1-proof}---we deal with this part of the proof in Lemma \ref{lem:L01}; (\textbf{ii}) the proof obtained from Theorem \ref{lem:PI-proof-DET-taylor-equals-DET-circ-div}, using the transformation from Lemma \ref{lem:Taylor-basic-properties}---we deal with this in  Lemma \ref{lem:Taylor-good-assignment}; (\textbf{iii}) the proof obtained from Proposition \ref{lem:reduce-syntactic-degree-Taylor} in Section \ref{sec:Reducing-the-Syntactic-Degree-of-the-Determinant-Polynomial}---which is dealt with in Lemma \ref{lem:small-good-assgn-forgot}. 
 
\begin{lemma}[in \VZ]\label{lem:L01}
All division gates $F^{-1}$ in the proof obtained from Proposition \ref{prop:write_PC-1-proof} are provably good under the identity matrices assignment. 
\end{lemma} 

\mar{Define identity matrices assignment. Make sure that being a good division gate is enough for all cases of division gates} 
\begin{proof}
This is proved by direct construction of the \PIZ-proofs.

We need to show that every division gate  $F^{-1}$ in this proof is  provably good under $\rho$.
Let $\mathcal U $ be the set of all the subcircuits  $ u $ such that $ u ^{-1} $ occurs in some circuit in $ \pi $. 
We show  that: 
\newcommand{\syntProperty}{{\large $\divideontimes$}~}
\begin{quote}
\syntProperty
\emph{An assignment is provably good for  $\mathcal U$ iff the assignment is provably good for the equations \eqref{eq:1}, \eqref{eq:2}.}
\end{quote}
By Claim \ref{cla:1-2-defined-under-identity-matrix} $\rho$ is provably good for equations  \eqref{eq:1} and \eqref{eq:2} which concludes the proof.
%
\mar{do I need to use simple gates u?}
%
%

\medskip 

We  only check \syntProperty for the part of the proof $\pi$ that proves equation \eqref{eq:2}, where $C$ is a triangular matrix  (see \cite[Proposition~7.6 (i)]{HT12}). The inspection of the other parts is similar.
\TODO{Need to make sure that this holds also for the other proofs.}

Assume that  $U$ is  an (upper or lower)  triangular \(n\times n\) matrix. We  show that in the   $\PIZ$-proof $\pi$ of    
\begin{equation}\label{eq:C-triang-main}
\cdeti(U)=u_{11}\cdots u_{nn}\,,
\end{equation}
all division gates $w^{-1}$ are provably good under the identity matrices assignment  iff $u_{ii}^{-1}$ in \eqref{eq:C-triang-main} are provably good under this assignment, for all $i\in[n]$.



For a matrix $A$ with entries $ a_{ij} $ we let
$A[k]:=\{a_{ij}\}_{i,j\in[k]}$. Assume that $U$ is a lower triangular matrix (the case for an upper triangular matrix is similar). Denote by $\overline 0$ the zero vector of the appropriate dimension (depending on the context). We write \(U\) as   
\begin{displaymath}
U:=
\begin{pmatrix}U[n-1] &  \overline 0^t \\
v & u_{nn} \\
\end{pmatrix}\,.
\end{displaymath}       
 The \PIZ-proof of (\ref{eq:C-triang-main}) (in \cite[Proposition~7.6.(i)]{HT12}) proceeds, using the  construction of $\cdeti(U)$ as follows (the explicit proof is omitted in \cite{HT12}):\iddo{}
\begin{align*}
\cdeti(U) & = \cdeti(U[n-1])\cd (u_{nn}-v (U[n-1])^{-1} \overline 0^t)\\
                                                                & = \cdeti(U[n-1])\cd u_{nn}\\
                                                                & = \cdeti(U[n-2])\cd (u_{n-1,n-1}-v '(U[n-2])^{-1}\ \overline 0^t)\cd u_{nn}\\
                                                                & = \cdeti(U[n-2])\cd u_{n-1,n-1}\cd u_{nn} \\
                                                                & = \dots = \prod_{i=1}^n u_{ii}
\end{align*}
where $v'$ is $v$ excluding the  $(n-1)$th coordinate.

Inspecting the proof sequence above we see that all the  division gates in the proof appear in the circuits computing the matrices $(U[n-1])^{-1},\dots,(U[1])^{-1}$. It thus suffices to show that all division gates appearing in the circuits  $ (U[n-1])^{-1},\dots,(U[1])^{-1}$ are provably good under $\rho$ iff all   $ u_{ii}^{-1}, $ for $i\in[n]$, are.  Specifically, we show the following claim, from which it is immediate to conclude that all these division gates $u$ have division axiom free \PIZ-proofs  of $u\rst\rho=1$ (using the \PCZ-axiom $0\cd F=0$):

\iddo{}
\begin{claim}\label{cla:good-div-gates-in-proof}
Every division gate in the circuits  $(U[n-1])^{-1},\dots,(U[1])^{-1}$  is of the form 
$$(u_{ii}-\overline 0\cd \overline h)^{-1},$$
for some (possibly empty) vector of circuits with division $\overline h$ and some $i\in[n].$ 
\end{claim}
\begin{proofclaim}
For brevity, let $U_i:=U[i]$, for $ i=1,\dots,n-1$, and let $v_{(i)}:=(v_1,\dots,v_i)$, and denote by $I_n$ the identity matrix of dimension $n$. Denote 
\begin{gather*}
A= U_{i-1}^{-1}
\left(I_{i-1}+\delta(U_i)^{-1}v_{(i-1)}^{t}
\overline 0  U_{i-1}^{-1}\right)  .
\end{gather*}
Then, by construction
\begin{equation}
U_i^{-1}= 
\left(\begin{array}{l r}
    A     & -\delta(U_i)^{-1}U_{i-1}^{-1}v_{i-1}^{t} \\
    -\delta(U_i)^{-1} \overline 0 U_{i-1}^{-1} & \delta(U_i)^{-1}
    \end{array}
    \right)\,
\label{eq:inverse}
\end{equation}

and
\begin{equation}\label{eq:delta-division-elim}
\delta(U_i)= u_{ii}- \overline 0 U_{i-1}^{-1}v_{(i-1)}^{t}\,.
\end{equation}

We proceed by induction on $i$ to show that all division gates in $ (U_{n-1})^{-1},\dots,(U_1)^{-1}$ are of the required form. The base case $i=1$ holds by definition, since $(U_1)^{-1}=u^{-1}_{11}$.

The induction step $(U_i)^{-1}$ is proved as follows. By 
(\ref{eq:inverse}) all the division gates in $U_{i}$ appear  either in    $(U_{i-1})^{-1}$
or in  $\delta(U_i)^{-1}$. By induction hypothesis the division gates in $(U_{i-1})^{-1}$ are all of the form $(u_{jj}-\overline 0\cd \overline h)^{-1} $, for some vector of circuits $\overline h$ and some $ j=1,\dots,i-1$. Moreover, the single-output circuit $\delta(U_i)^{-1}$ contributes the outermost division gate $(u_{ii}-\overline 0\cd \overline h)^{-1}$, to $\overline h= U_{i-1}^{-1}v_{(i-1)}^{t}\,$, which is of the correct form; as well as all the division gates in $\overline h$ itself. But all the division gates in $\overline h$ appear in $U_{i-1}^{-1}$, which by induction hypothesis  all have the required form. 
\end{proofclaim} 
\bigskip 

It remains to show the following: 

\begin{claim}[in \VZ]\label{cla:1-2-defined-under-identity-matrix}
All division gates in equations \eqref{eq:1}, \eqref{eq:2} are provably good under $\rho$.\iddo{} \iddo{}
\end{claim}
\begin{proofclaim}
For  equation (\ref{eq:2}) above, this is immediate from the above, since every inverse gate that appears in the proof (and equivalently  in equation (\ref{eq:C-triang-main})) is one of $(u_{ii}-\overline 0\cd \overline  h)^{-1}$, for some $ i\in[n]$, which is provably good under the identity matrices assignment.

For equation \eqref{eq:1}, namely, $\cdeti(X)\cdeti(Y)=\cdeti(XY),$ by  definition and notation in \eqref{eq: X} we have $\cdeti(X):=
 \cdeti(X_{1})\cdot \dd(X)= \cdeti(X_1)\cd(x_{nn}- v_{2}X_{1}^{-1}v_{1}^{t})\,$, and (by induction on $n$ the dimension of $X$\iddo{}), similarly to Claim \ref{cla:good-div-gates-in-proof} above, we can verify that every division gate in $\cdeti(X)$ is provably good under $\rho$. Similar reasoning applies to $\cdeti(Y)$ and $\cdeti(XY)$. \iddo{} We omit the details.  
\end{proofclaim}

This concludes the proof of Lemma \ref{lem:L01}.
\end{proof}
\bigskip

\iddo{}

\iddo{}

\iddo{}   

\iddo{}

\iddo{}

\iddo{}

\begin{lemma}\label{lem:Taylor-good-assignment}
Every division gate $F^{-1}$ that appears in the  $\PIZ$-proof in Lemma \ref{lem:PI-proof-DET-taylor-equals-DET-circ-div} is provably good under $\rho$. 
\end{lemma} 

\begin{proof}

We wish to show that every division gate $F^{-1}$ that appears in the  $\PIZ$-proof in Lemma \ref{lem:PI-proof-DET-taylor-equals-DET-circ-div} of
\begin{equation}\label{eq:det-taylor-eqls-dte-circ-div}
\cdeti(X)=\taydet(X),
\end{equation}
where $\taydet(X):=\Coef_{z^n}(\cdeti(I-zX))$, is provably good under $\rho$. The cases where $A$ in Lemma \ref{lem:PI-proof-DET-taylor-equals-DET-circ-div} is $XY$ or $Z$ are similar. 

We shall  characterize all division gates in this proof. The goal of this \PIZ-proof as shown in \cite[Proposition 7.9]{HT12} is to show that $\cdeti(X)$ is equal to $z^n\cd\taydet(X)+\sum _{i< n}z^i\cd Q_i$ and then apply Lemma \ref{lem:Taylor-basic-properties} to extract the equality between the coefficients of $z^n$. For this purpose, the proof of \eqref{eq:det-taylor-eqls-dte-circ-div}
proceeds by first proving 
\begin{equation}\label{eq:taydet001}
\cdeti(zI_n+X^{-1})=z^n+\sum_{i=0}^{n-1}z^i\cd Q_i
\end{equation}
where the $Q_i$'s are circuits with division, and  such that $z$ does not occur in the $Q_i$'s. We then use the \PIZ-proof of $\cdeti(X)\cdeti(Y)=\cdeti(XY)$ proved in Section \ref{sec:P-C-inverse-proof-of-1-2} to  show that $\cdeti(zI_n+X^{-1})\cd\cdeti(X)=\cdeti(zX+I_n)$. By this and \eqref{eq:taydet001}, we get that 
\begin{equation}\label{eq:taydet002}
\cdeti(zX+I_n)=\cdeti(X)\cd(z^n+\sum_{i=0}^{n-1}Q_i).
\end{equation} 
We  use Lemma \ref{lem:Taylor-basic-properties} to extract the coefficient of $z^n$ in \eqref{eq:taydet002} to get $\Coef_{z^n}(\cdeti(zX+I_n))= \cdeti(X)$, which is our desired identity.
\medskip 

We will first check that the division gates in $\cdeti(zI_n+X^{-1})$ are provably good under $\rho$. The division gates in $\cdeti(zI_n+X^{-1})$  include, among other gates, the gates in the circuit $X^{-1}$. But by the proof of Lemma \ref{lem:L01} the division gates in $X^{-1}$ are all provably good under $\rho$. The other division gates in $\cdeti(zI_n+X^{-1})$ occur also in $\cdeti(X)$. More precisely, let $u^{-1}$ be a division gate in $\cdeti(X)$; then $u^{-1}$ occurs also in the substitution instance   $\cdeti(zI_n+X^{-1})$ of $\cdeti(X)$. Such a division gate $u^{-1}$  in $\cdeti(zI_n+X^{-1})$, 
by definition  \eqref{eq:det-min-one-definition} of $\cdeti$, will occur in the following term $\cdeti(zI_n+X^{-1})=z\cd 1+\delta(X)^{-1}-(-\delta(X)^{-1}v_2 X_1^{-1})\cd (X^{-1})_1^{-1}\cd(-\delta(X)^{-1} X_1^{-1}v_1^t) $ (where $(X^{-1})_1^{-1}$ is $(X^{-1})^{-1}$ without the rightmost column and without the lowest row). Note that in this term \emph{we have only the division gate }$\delta(X)^{-1}$ (terms like $(X^{-1})_1^{-1}$ contain division gates, but they do not stand themselves as division gates, since $X^{-1}$ is a notation for a matrix, not a circuit whose root is a division gate).   

We thus need to construct a free division axiom \PIZ-proof for the fact that  the  division gate $\delta(X)^{-1}$ is 1, under $\rho$. But this already stems from the proof of Lemma \ref{lem:L01} above. 

Precisely the same reasoning shows that for every circuit $F$ in the \PI-proofs of $\cdeti(X)=\taydet(X)$ in Lemma \ref{lem:PI-proof-DET-taylor-equals-DET-circ-div}  we have: if $z$ occurs in the scope of a division gate in $F$ then $\big(\Den(F)\big)(z/0)$ is provably good under $\rho$. \iddo{} 
\end{proof}

\begin{lemma}\label{lem:small-good-assgn-forgot}
Every division gate $u^{-1}$ that appears in the  $\PIZ$-proof in Lemma \ref{lem:reduce-syntactic-degree-Taylor} is provably good under $\rho$.  
\end{lemma}
\begin{proof}
We consider the case where $A$ in Lemma \ref{lem:reduce-syntactic-degree-Taylor} is $X$. The cases for $XY, Z$ are similar.

By Section \ref{sec:Reducing-the-Syntactic-Degree-of-the-Determinant-Polynomial} every division gate in the $\PIZ$-proof of Lemma \ref{lem:reduce-syntactic-degree-Taylor} is of the form $(\Den\left(\cdeti\left(I_n+zX\right)\right)(z/0))^{-1}$.  Identical to the proof of Lemma \ref{lem:reduce-syntactic-degree-Taylor}, start by constructing directly a \PIZ-proof of $(\cdeti\left(I_n+zX\right))(z/0)= \cdeti\left(I_n\right)$, by substituting $z$ by 0 and gradually replacing in the proof $0\cd u$ to $0$ and $0+u$ to $u$, for $u$ any subcircuit. Now,  by the proof of Lemma \ref{lem:L01}, all the division gates in $\cdeti\left(I_n\right)$ are provably good under $\rho$. By Lemma \ref{lem:from-all-Den-gates-proofs-to-top-Den} we get that $(\Den\left(\cdeti\left(I_n+zX\right)\right)(z/0))=1$ has a division axiom free \PIZ-proof (in fact, already a \PCZ-proof).\iddo{}
\end{proof}

In order to be able to apply Lemma \ref{lem:inv} in Section \ref{sec:homogenization-of-proofs} we need to show how to construct in the theory \PCZ-proofs of the fact that the division gates from Corollary \ref{cor:exist-normalized-div-PI-div-proof-of-det-identities}, are equal 1.

\begin{corollary}[in \VZ]\label{cor:what-we-need-for-div-elim} Let $\pi$ be the \PIZ-proof of \eqref{eq:1} and \eqref{eq:2} from Corollary \ref{cor:exist-normalized-div-PI-div-proof-of-det-identities}, in which every division gate in a circuit in  $\pi$ appears at the top as $\Den(F)^{-1}$, for some circuit $F$, and where the determinant in  \eqref{eq:1} and \eqref{eq:2} is written as $\taydetsharp$. For every such $F$ in $\pi$ there exists a \PCZ-proof of $\hc{(\Den(F\shift))}{0}=1$  (where $\sigma$ is the linear shift as defined above).  
\end{corollary}
\begin{proof}
This follows from Lemma \ref{lem:from-all-Den-gates-proofs-to-top-Den} and the fact that all division gates in $\pi$ were shown to be provably good under $\rho$ in the lemmas that follow Lemma \ref{lem:from-all-Den-gates-proofs-to-top-Den}. 
\end{proof}

\subsection{Eliminating Division}
We are now ready to construct the division free PI-proof of the determinant identities (with some restrictions). Note that we will need to use the fact that division gates are provably good in the next section (Section \ref{sec:homogenization-of-proofs}). 

We first define the following:

\begin{definition}[Correct up to degree $k$ $\PCZ$-proof]\label{def:correct-up-to-deg} Let $k$ be a natural number. We say that a $\PCZ$-proof sequence $\pi$ is \emph{correct up to degree $k$} if (i) every proof-line in $\pi$ is an equation between algebraic circuits that was derived by one of the derivation rules of $\PCZ$  from previous lines; or (ii) is a variant of the division axiom {\rm D}, where instead of $F\cd F^{-1}=1$ we have the line  $F\cd\inv_k(F) = 1$; or (iii) is an axiom of $\PCZ$ different from {\rm D}.\iddo{}
\end{definition}

The witness for syntactic correctness of a correct up to degree $k$ $\PCZ$-proof is similar to that in Section \ref{sec:encoding-PI-proofs}.
%
Note that we do not need to witness the syntactic-degree of nodes in circuits in a correct up to degree $k$ \PCZ-proof. In other words, there exists a \SigZB-formula $\psi(Z,k)$ that holds if $Z$ is a correct up to degree $k$ \PCZ-proof (where $Z$ contains also the syntactic correctness witness as in Section \ref{sec:encoding-PI-proofs}). The formula $\Psi(Z,k)$ only needs to verify that in the division axiom D we have $F\cd\inv_k(F) = 1$, and checking whether a circuit is $\inv_k(F)$ is done without the need to witness the syntactic-degree of $F$ or $\inv_k(F)$. \iddo{}   

\begin{corollary}[in \VZ]\label{cor:exists-corr-up-to-deg-n-PC-proof-of-dettaysharp}
Given a positive natural number $n$, there exists a \emph{correct up to degree $2n$} $\PCZ$-proof sequence of the determinant identities \eqref{eq:1} and \eqref{eq:2} for \nbyn\ matrices where the determinant in \eqref{eq:1} and \eqref{eq:2} is written as the division free circuit $\taydetsharp(A)$ of syntactic-degree $n$, for $A$, the \nbyn\ symbolic matrix $X$ or $Y$, or their product $XY$, or a symbolic triangular matrix $Z$. 
\end{corollary}

\begin{proof}
By Corollary \ref{cor:exist-normalized-div-PI-div-proof-of-det-identities}
there exists a syntactically correct \PIZ-proof $\pi_3$ of the determinant identities where the determinant  in \eqref{eq:1} and \eqref{eq:2} is written as $\taydetsharp$, and in which each circuit appearing in the proof has at most one division gate that appears as the root of the circuit\iddo{}. 

The rest of the proof follows the sketch in Step 5 (Section \ref{sec:carry-proof-in-theory}). 
Assume that $r_i$ (for $i\in J$) are all the variables appearing in $\pi_3$ and  $b$ is the assignment of identity matrices to the variables in $\pi_3$ (that is, the assignment of $I_n$ to the (variables belonging to the) matrices $X,Y,Z$ in $\pi_3$). 

Substitute in $\pi_3$ the term $(b_i-w_i)$ for each $r_i$ (for all $i \in J$) denoting the obtained proof by $\pi_3'$. The substitution is performed as follows:
given $\pi_3$ there is a \SigZB-formula $\Psi(\pi_3,v)$ that is true iff $v$ is a node in a circuit in $\pi_3$ that is labeled by a variable (more formally, $v$ is a number in the set of nodes $V$ in a circuit in the proof and $(v,t)\in G$, where $G$ is the string specifying the gates of $V$, and $t$ is a number that specifies one of the input variables; recall the encoding of circuits in Section \ref{sec:subsec-for-encoding-circuits}). Using $\Psi(\pi_3,v)$ we can define a \SigZB-formula $\Phi(\pi_3,i)$ such that $\Phi(\pi_3,i)$ is true iff  $i\in \pi'_3$, where $\pi'_3$ is a string that encodes the desired substitution instance of $\pi_3$. This is done using a bounded number universal quantifier over nodes $v$ in $\pi_3$ and over all indices of circuits in $\pi_3$, so that when $\Psi(\pi_3,v)$ holds, $(v,t)\in G$ is replaced by ``$(b_i-w_i)$'' as follows:  $(v,t)$ is replaced by $(v,t')$ where $t'$ specify a $+$ gate; then we add four new nodes in $V$, new edges in $E$, new  gates in $G$, and new inputs to the input string $I$, encoding that the product gate $v$ has two children: $b_i$ and $-1\cd w_i$.


Accordingly, when we define the substitution instance above we also need to update the witnesses for the proofs. The witnesses all stay the same except the witnesses for the \PI\ axioms C1 and C2, that is, $F_1\oplus  F_2=F_1+ F_2$ and $F_1\otimes  F_2=F_1 \times F_2$, respectively. For these two axioms we specify the new mappings of $F_1, F_2$ in the right hand side of the equation to the circuits in the left hand side of the equation under the substitution. Since the substitution is defined by replacing some leaves in $F_i$ by some circuits then we only need to add new pairs to the original mapping to update it appropriately.\iddo{}   

\iddo{}

By Lemma \ref{cor:what-we-need-for-div-elim} the all zero assignment $\overline 0$ to the $w_i$ variables in  $\pi_3'$ gives every division gate the value  $1$. 

%
For every circuit $C$ with a top division gate $H^{-1}$, let $\inv_{2n}(H)$ be  the truncated power series of $H^{-1}$ over the zero assignment (we use $2n$ because this is the degree of the determinant identities). \VZ\ proves there exists a corresponding division-free circuit $C'$, obtained by replacing the division gate $H^{-1}$ in $C$ by $\inv_{2n}(H)$. This is done like the substitution of variables describes above, where now the nodes we define using a \SigZB-formula $\Psi(\pi_3,v)$ are the inverse nodes $u^{-1}$, and the substitution adds nodes and wirings while also using the subcircuit $u$ itself (as needed to construct $\inv_{2n}(u)$).\iddo{}

Let $\pi_3''$ be the corresponding division-free proof-sequence obtained from $\pi'_3$ by replacing every circuit $C$ with a division gate at the top by $C'$ as above. 


\iddo{}
In $\pi_3''$ every occurrence of the axiom D of division $F\cd F^{-1}=1$ is replaced by $F\cd\inv_{2n}(F) = 1$. Now substitute back the original variables $X,Y,Z$ instead of the terms $b_i-w_i$ .
The resulting sequence is a correct up to degree $2n$\ $\PCZ$ proof sequence of the desired determinant identities.
\iddo{}
%
\end{proof}

\iddo{}

\section{Eliminating High Degrees: Constructing PI-Proofs with Polynomial Syntactic-Degrees}\label{sec:homogenization-of-proofs}

\emph{From now on all our algebraic circuits are division free}. Let $C$ be an algebraic circuit. Recall that $\deg(C)$ is the syntactic-degree of $C$ defined to be the maximal syntactic-degree of a node in $C$  (Definition  \ref{def:syntactic-degree}). Also, recall that  $C^{(i)}$ is a syntactic-homogeneous circuit computing the degree $i$ homogeneous part of $\widehat C$ (see Section \ref{sec:algebraic_circuits}). For a node $v$ in $C$, $C_v$ denotes the circuit rooted at $v$.  

We now show how to eliminate, within  \VZ, all nodes of syntactic degrees exceeding a given number $d$, from circuits in a PI-proof of the equation $F=G$ assuming the syntactic-degree of $F,G$ is at most $d$.
\mar{FAC0 is undefined}

We first note that we do not know of  an \FACZ\ algorithm that  computes the syntactic-degrees of nodes in a given circuit (see also Section \ref{sec:witnessing-syntactic-degrees}).
However, for our purposes it will be sufficient to input an upper bound on the syntactic-degree of the circuit to be homogenized (except for the proof of the Cayley-Hamilton theorem in Section \ref{sec:corollaries} wherein we need to witness explicitly the syntactic-degree of nodes in a proof in order to use Lemma \ref{lem:Taylor-basic-properties} part (3); see Section \ref{sec:witnessing-syntactic-degrees} for the definition of syntactic-degree witness). 

Our algorithms will  receive a circuit $C$ and a number $d$ that will serve as an upper bound on the syntactic-degree of the circuit (the theory does not need to verify the correctness of this syntactic-degree).  

\iddo{}  

\medskip 
\medskip
\hrule  
\medskip 
\noindent \textsc{{Homogenization Algorithm in Uniform \FACZ}}\smallskip \mar{FAC is undefined}
\medskip 
\hrule
\medskip 

\noindent \textbf{Input:} an arithmetic circuit $C$ of size $s$ and a natural number $d$.
\begin{description}
\item[Optional input 1:] A syntactic-degree witness for all the nodes in $C$ (including the root that has syntactic-degree $d$). 

\item[Optional input 2:] A natural number $i$. \end{description}

\noindent \textbf{Output:} \begin{enumerate}
\item 
An arithmetic circuit $C'$ of size $O(d^2\cdot s)$ computing the polynomial $\widehat C$ such that  $C'$ is a sum of syntactic homogeneous circuits $C'=C^{(0)}+\dots+C^{(d)}$ ($C^{(i)}=0$ for $i>\deg(C)$). 
\item If optional input 1 was supplied, then for every gate $v$ in $C'$, the duplicate gate $[v,j]$ for $j>\deg(v)$, is the circuit $0$ (see below for the notation ``$[u,j]$"). 

\item If the input $C$ is (declared to be)\footnote{The algorithm does not check for correctness of $C$ being a (sum of) syntactic homogeneous circuits.} a (sum of) syntactic homogeneous circuits $\sum_{i\in I}C^{(i)}$ for $I\subseteq\{0,\dots,d\}$  then output $C$, augmented with the nodes $[u,j]=0$, for all  nodes $u\in C$ and all $j\in \{0,\dots,d\}\setminus I$.  \label{it:homogeneous-algorithm-part3}

\item
If optional input 2 was supplied, then $C'=C^{(i)}$, namely  the $i$th homogeneous component.
If moreover  the input circuit $C$ is already a syntactic homogeneous circuits $C^{(j)}$ then the output is the circuit $0$ if  $j\neq i$ and is $C^{(i)}$ if $j=i$. 
\end{enumerate}\medskip

\noindent\textbf{Algorithm:} 
We follow the standard Strassen \cite{Str73} algorithm, but instead of building the circuits by induction from leaves to root we construct all nodes simultaneously as follows.

(\textbf{1}) Assume we do not have the witness for syntactic-degrees of all the nodes (namely, the witness was not supplied as an input). 
Every node $v$ is duplicated $d+1$ times into the nodes $[v,0],\dots,[v,d]$. For a node $[v,i]$ we call $i$ the \textbf{\textit{syntactic-degree upper bound of}} $[v,i]$,  denoted as 
$$\degub([v,i]):=i.$$ 
 
 The node  $[v,i]$ is (the root of) a syntactic-homogeneous circuit of syntactic-degree at most $i$ computing \emph{either }$0$ or the degree $i$ homogeneous part of the polynomial $\widehat {C_v}$. 
\mar{homogeneous part is defined?}
The algorithm is doable in \FACZ\ because every new node $[v,i]$ depends only on the copies of the two nodes $u,w$ that goes into $v$, and these nodes are already known from the input circuit, namely, they are \([u,i], [w,i]\), for  \(i=0,\dots,d+1\), where \(v=u+w\) or \(v=u\cdot w\) in \(C\). Hence, the wiring of the new circuit is done in parallel for each of the new nodes as follows:

\case 0 $v$ is a leaf in $C$. If $v$ is a constant $\alpha$, then define $[v,0]=\alpha$, and $[v,i]=0$ for all $i=1,\dots,d$. Otherwise, $v$ is a variable $x_j$, and we define $[v,1]=x_j$, and $[v,i]=0$ for all $1\neq i\in\{0,\dots,d\}$. 

\case 1 $v=u+w$ in $C$. Define $[v,i]:=[u,i]+[w,i]$ for every $i=0,\dots,d$.

\case 2 $v=u\times w$ in $C$. Define $[v,i]:=\sum_{j+k=i \atop j,k=0,\dots,d}[u,j]\times[w,k]$.
\smallskip 
 
\noindent Finally, $C^{(i)}:=r^{(i)}$, for $r$ the root of $C$, for all $i=0,\dots, d$. 

\medskip

 (\textbf{2}) Otherwise, assume that a witness for the syntactic-degree $\deg(v)$ for every node $v$ in $C$ was supplied as an input. In this case the algorithm is the same as above, except that 
the $i$th duplicate $[v,i]$ of a node $v$ is defined to be the circuit 0 whenever $i>\deg(v)$.  More precisely: 
  
 \case 0 $v$ is a leaf in $C$. If $v$ is a constant $\alpha$, then define $[v,0]=\alpha$. Otherwise, $v$ is a variable $x_j$, and we define $[v,1]=x_j$, and $[v,j]=0$, for all $0\le j\le d$ and $j\neq 1$.
 
 \case 1 $v=u+w$ in $C$. Define $[v,i]:=[u,i]+[w,i]$ for every $i=0,\dots,\deg(v)$, and $[v,i]:=0$ for $i=\deg(v)+1,\dots,d$.
 

\case 2 $v=u\times w$ in $C$. Define $[v,i]:=\sum_{j+k=i \atop j,k=0,\dots,d}[u,j]\times[w,k]$,  for every $i=0,\dots,\deg(v)$, and $[v,i]:=0$ for $i=\deg(v)+1,\dots,d$.\footnote{Note that this means that provably in \VZ\ there exists a \PCZ-proof of $[v,i]:=\sum_{j+k=i \atop 0\le j\le \deg(u) ,0\le k\le \deg(k)}[u,j]\times[w,k]$, for every $i=0,\dots,d$, since $[u,j]=0$ for $j>\deg(u)$ and $[w,k]=0$ for $k>\deg(u)$.}

Finally, $C^{(i)}:=r^{(i)}$, for $r$ the root of $C$, for all $i=0,\dots, d$. \smallskip 

\medskip\hrule 

\bigskip

\para{Note on syntactic-degree upper bounds}
%
Notice that if we do not have a witness for syntactic-degrees and assuming we get as input a correct upper bound, that is, $d\ge \deg(C)$, the above algorithm produces a syntactic homogeneous circuit in which each node $[u,i]$ is of syntactic-degree $i$, \emph{except} that for $\widehat {[u,i]}=0$, the syntactic degree of the circuit rooted at $[u,i]$ can be \emph{smaller} than $i$ (but not bigger). This means that the circuit contains in itself \emph{a witness for the upper bound of the syntactic-degree of each node}. 

We will manage to work out our argument without the need to compute syntactic-degrees except for the Cayley-Hamilton theorem shown in Section \ref{sec:corollaries}. 

On the other hand, if we receive a syntactic-degree witness as an input then assuming $[u,i]$ is not the circuit 0, $u$ has syntactic-degree at least $i$. Moreover, if the input to the algorithm is already a sum of syntactic homogeneous circuits $\sum_{i=0}^k \hci{C}$ then $[u,i]=\emptyset$ for $i>k$ for every node $u$ in $C$.

We are going to construct \PCZ-proofs which witness the syntactic-degree of every node in:

\begin{definition}\label{def:appears-with-syntactic-degree-upper-bound}
Given a \PCZ-proof $\pi$ we say that every node in every circuit in the proof \emph{appears with its syntactic-degree upper bound \degub} if every such  node $u$ is a pair of numbers $[u,i]$ for $\degub(u)=i$, according to the construction in the homogenization algorithm above. 
\end{definition}






\iddo{}

We have the following main theorem about homogenization of proofs:


\iddo{}

\begin{theorem}[in \VZ]\label{thm:homogenize-proofs}\mar{Update that now homogenization is in \VNCTwo\ and not \VZ}
Let $d$ be a natural number and assume that  $F=G$ has a correct up to degree $d$ \PCZ-proof. Then, for every $k=0,\dots,d$, 
there exists a $\PCZ$-proof of $F^{(k)}=G^{(k)}$ of syntactic-degree at most $k$, in which every circuit is a syntactic homogeneous circuit in which every node $u$ appears with its  syntactic-degree upper bound, and $\degub(u)\le k$. 

\end{theorem} 
\iddo{}
We need the following lemmas before concluding this theorem.

\begin{lemma}[in \VZ]\label{lem:bounded degree}
Let $F_1\oplus F_2$ and $F_1\oplus F_2$ be two circuits and let $k$ be a natural number\footnote{Recall that we mean here that we pick one such circuit out of all possible circuits of this form.}. The following have $\PCZ$-proofs, in which every circuit is a sum of syntactic homogeneous circuits in which every node $u$ appears with its  syntactic-degree upper bound: \iddo{} 
\begin{enumerate}
\item\label{c:1}  $(F_1\cplus F_2)^{(k)}= F_1^{(k)}+F_2^{(k)}$;
\item\label{c:2}   $(F_1\ctimes F_2)^{(k)}= \sum_{i=0}^kF_1^{(i)}\cdot F_2^{(k-i)}$.
\end{enumerate}
\end{lemma}

\begin{proof}
Using the homogenization \FACZ-algorithm above  we  construct $(F_{1}\cplus F_{2})^{(k)}$. By definition, $(F_{1}\cplus F_{2})^{(k)}:=F_1^{(k)}\cplus F_2^{(k)}$, and by axiom C1,  $ F_1^{(k)}\cplus F_2^{(k)}=F_1^{(k)}+ F_2^{(k)}$. We thus construct this one-line proof, adding to the proof a witness for the application of axiom C1. \mar{make sure to talk about such witnesses above.} Note that given a circuit $A\oplus B$ we can  construct in \VZ\ a witness for the correctness of applying C1 to get $A\oplus B=A+B$. The \iddo{} witness will say that $A$ in $A\oplus B$ is identical to $A$ in $A+B$, by an explicit mapping of nodes from the former to the latter copy; and similarly for $B$.
 Furthermore, note that in this one-line \PCZ-proof every node $u$ in every circuit appears with its syntactic-degree upper bound: in $(F_{1}\cplus F_{2})^{(k)}$ this is true by construction of $\hc{(\cd)}{k}$ and in  $F_1^{(k)}+F_2^{(k)}$ we simply specify every node $u$ in  $F_1^{(k)}+F_2^{(k)}$ to have the same syntactic-degree upper bound as its origin node in $(F_{1}\cplus F_{2})^{(k)}$ (note that this is indeed a  true upper bound on the syntactic-degree of $u$).  

This concludes \ref{c:1}. Part \ref{c:2} is similar using the homogenization algorithm above.
\iddo{}
%
\end{proof}

%
%



We now conclude the proof of the theorem:


\bigskip

\begin{proof}[Proof of Theorem \ref{thm:homogenize-proofs}]


For every $k=0,\dots,d$, we devise a \SigZB-definable function in \VZ\ that produces a \PCZ-proof of $F^{(k)}=G^{(k)}$ with every node $u$ in every circuit in the proof appears with its syntactic-degree upper bound $\degub(u)$ (that is, it appears as $[u,i]$ where $i=\degub(u)$) and such that $\degub(u)\le k$. This is done in a  manner resembling  the algorithm above for homogenizing circuits. %
%
%
%
%
Specifically, for every $k =0,\dots,d$ and every line $S=T$ in $\pi$,  we construct in parallel  a part of the proof of  $\hck{S}=\hck T$ (that taken collectively would amount to a proof of   $\hck{S}=\hck T$).

Though some proof-lines $S=T$ possess syntactic-degree witnesses while some are already syntactic homogeneous, and some proof-lines do not fall into the two former categories, the proof we show works for all these three cases; this stems from the way we defined the homogenization algorithm: this algorithm  constructs the nodes $[u,i]$ for every original node $u$ in its input and every $i=0,\dots,d$ (even in the case where $i$ exceeds the specified syntactic-degree of $u$ in the syntactic-degree witnesses input; and similar for the output of the homogenization algorithm part \ref{it:homogeneous-algorithm-part3}).
\smallskip

\case 1  $S=T$ is an axiom of $\PCZ$ of size $s$. We construct a \PCZ-proof of  $S^{(k)}=T^{(k)}$ with size $s\cdot \poly(k)$ (and syntactic degree $\leq k)$. 

Lemma \ref{lem:bounded degree} gives a proof $(F_1\cplus F_2)^{(k)}=(F_1+ F_2)^{(k)} $ and $(F_1\ctimes F_2)^{(k)}=(F_1\cdot F_2)^{(k)} $, as required for the axioms C1 and C2.
\iddo{}

Axioms A1 and A10 are immediate. For the other axioms, consider for example the axiom $F_1\cdot (F_2\cdot F_3)= (F_1\cdot F_2)\cdot F_3$, where the circuits have size $\leq s$. We have to construct a proof of
\begin{equation}(F_1\cdot (F_2\cdot F_3))^{(k)}= ((F_1\cdot F_2)\cdot F_3)^{(k)}\,.\label{eq:0}\end{equation}
By part (ii) of Lemma \ref{lem:bounded degree} the equations
\begin{eqnarray}
(F_1\cdot (F_2\cdot F_3))^{(k)}&=&\sum_{i=0}^k F_1^{(i)}\left( \sum_{j=0}^{k-i} F_2^{(j)} F_3^{(k-i-j)}\right) \label{eq:F123A}\\
((F_1\cdot F_2)\cdot F_3)^{(k)}& =& \sum_{i=0}^k \left( \sum_{j=0}^{i} F_1^{(j)} F_2^{(i-j)}\right)\cdot F_3^{(k-i)} \label{eq:F123B}\,,
\end{eqnarray}
can be proved in \PCZ. 
In $\PCZ$, the right hand sides of both (\ref{eq:F123A}) and (\ref{eq:F123B}) can be written as
$\sum_{ i+j+l=k} F_1^{(i)}F_2^{(j)}F_3^{(l)}$ by a proof of size roughly $s(k+1)^4$. (This gives the proof of (\ref{eq:0}) of size $s\cdot \poly(k)$.)
\smallskip 
\iddo{}

\case {2}  If the  line $S=T$ is  $S_1\cdot S_2= T_1\cdot T_2 $, and it was derived  using the rule R4 as follows:
\begin{equation}\label{eq:rule-app-homogenization}
\frac{S_1=T_1~~~~~~ S_2=T_2}
{S_1\cdot S_2= T_1\cdot T_2}
\,.
\end{equation}
From previous lines $S_1=T_1$ and $S_2=T_2$, we construct the derivation of $S^{(k)}=T^{(k)}$ by using the  lines $\hci{S_1}=\hci{T_1}$ and $\hci{S_2}=\hci{T_2}$, for all $i=0,\dots,k$, as follows: by Lemma \ref{lem:bounded degree}, we can construct the proofs of $(S_1\cdot S_2)^{(k)}=\sum_{i=0,\dots, k} S_1^{(i)}\cd S_2^{(k-i)}$ and $(T_1\cdot T_2)^{(k)}=\sum_{i=0,\dots, k} T_1^{(i)}\cd T_2^{(k-i)}$. Hence, $(S_1\cdot S_2)^{(k)}= (T_1\cdot T_2)^{(k)}$ can be proved from the assumptions $S_1^{(i)}=T_1^{(i)}, S_2^{(i)}=T_2^{(i)}, i=0,\dots k$. (The proof has size roughly $s\cdot(k+1)^c (k+1)$.)

Note that this is done independently and simultaneously for each proof-line (and specifically, not by induction on proof-length). That is, given the rule application \eqref{eq:rule-app-homogenization}, we construct the (partial) proof of $\hc{(S_1\cdot S_2)}{k}= \hc{(T_1\cdot T_2)}{k}$ using only the lines $\hci{S_1}=\hci{T_1}$ and $\hci{S_2}=\hci{T_2}$, for all $i=0,\dots,k$; and in addition, since we also want to record the information about which line was derived from which previous lines we  add pointers to previous lines. The latter can be  defined via a \SigZB-definable \VZ\ number function given as input the line-numbers of $S_1=T_1$ and $S_2=T_2$. Hence the whole construction is in \VZ. 

%
%

\case 3 $S=T$ is the rule R1-R3. This is similar to the case for rule R4.  

\medskip 

The fact that in the \PCZ-proof we constructed every node $u$ in every circuit appears with its  syntactic-degree upper bound is clear from the construction, since we used the homogenization algorithm to produce the syntactic homogeneous circuits and Lemma \ref{lem:bounded degree}
\end{proof}

\newcommand{\homo}[2]{\ensuremath{\sum_{i=0}^{#2}\hc{(#1)}{i}}}

We need the following  claims:

\begin{claim}[in \VZ]\label{cla:homogeneous-Fii=Fi} Given a syntactic homogeneous circuit $F^{(d)}$ there exists a $\PCZ$-proof of ${\hc{F}{d}= \sum_{i=0}^d \hci{\left(F^{(d)}\right)}}$ in which every node $u$ in every circuit appears with its syntactic-degree upper bound.
\end{claim}

\begin{proofclaim}
This is by the definition of the homogenization algorithm: if the input to the algorithm is the (already) syntactic homogeneous circuit $\hc{F}{d}$ then $\hci{(\hc{F}{d})}=\emptyset$, for all $i\neq d$, and  $\hci{(\hc{F}{d})}=\hc{F}{d}$, for $i=d$. \end{proofclaim}

\begin{claim}[in \VZ]\label{cla:Coeff-zn-homogeneous-comp-n}
Let $F(x_1,,\dots,x_m)$ be a circuit without division, in the displayed input variables. Assume that every input variable is now multiplied by a new variable $z$ to get $F(zx_1,\dots,zx_m)$, which we denote by $F_z$. Then, there exists a \PCZ-proof of ${\Coef_{z^i}(F_z)=F^{(i)}}$ in which every node $u$ in every circuit appears with its  syntactic-degree upper bound.
\end{claim}

\begin{proofclaim}
This is by construction of $\Coef_{z^i}(\cd)$ and $\hci{(\cd)}$. Consider the construction of $\Coef_{z^i}(\cd)$ in \VZ\ as shown in the appendix Section \ref{sec:alg-for-coef}. Then, having each variable $x_j$ multiplied directly by $z$ in the circuit $F'$ means that $\Coef_{z^i}(F')=F^{(i)}$ are syntactically identical except for the bottom level of the circuits, namely $\Coef_{z^1}(z\cd x_j)=1\cd x_j$ (ignoring zero terms, and applying basic rearrangements), and $\Coef_{z^r}(z\cd x_j)=0$ for all $1\neq r\le k$ and  all $j\in[n]$. And accordingly, $\hc{x_j}{1}=x_j$ and $\hc{x_j}{r}=0$, for all $1\neq r\le k$ and  all $j\in[n]$.  
\end{proofclaim}


\begin{corollary}[in \VZ]\label{cor:exist-PC-proof-of-det-tay-sharp-identities}
Given a positive natural number $n$ there exists a $\PCZ$-proof of the determinant identities \eqref{eq:1} and \eqref{eq:2}, where the determinant in \eqref{eq:1} and \eqref{eq:2} is written as the division free circuit $\taydetsharp(A)$, for $A$, the \nbyn\ symbolic matrix $X$ or $Y$, or their product $XY$, or a symbolic triangular matrix $Z$. Moreover, in this proof every circuit is a sum of syntactic homogeneous circuits in which every node appears with its syntactic-degree upper bound $\degub$.
\iddo{} \end{corollary}

\begin{proof} By Corollary \ref{cor:exists-corr-up-to-deg-n-PC-proof-of-dettaysharp} there exists a \emph{correct up to degree $2n$} \iddo{} $\PCZ$ proof-sequence $\pi$ of the determinant identities \eqref{eq:1} and \eqref{eq:2} where the determinant in \eqref{eq:1} and \eqref{eq:2} is written as $\taydetsharp$. 
By Lemma \ref{lem:inv} and Corollary \ref{cor:what-we-need-for-div-elim} we have a $\PCZ$-proof of $\hci{(F\cd\inv_{2n}(F))} = 0$, for every $1\le i\le 2n$, and a $\PCZ$-proof of $\hc{(F\cd\inv_{2n}(F))}{0} = 1$, for every proof-line $F\cd\inv_{2n}(F)=1$ appearing in $\pi$, and these proofs also contain the syntactic-degree upper bound for every node in every circuit. Therefore, by Theorem \ref{thm:homogenize-proofs} part (i) we have a \PCZ-proof of 

\begin{equation}\label{eq:homogeneous-end-single-h-components}
\left(\taydetsharp(X)\cd \taydetsharp(Y)\right)^{(k)}=(\taydetsharp(XY))^{(k)},
\end{equation} 
for every $k=0,\dots,2n$, and similarly for \eqref{eq:2}, wherein every node appearing in the proof has a specified syntactic-degree upper bound.\iddo{}

Finally, we can conclude the corollary by reasoning as follows.


Recall the definition of $\taydetsharp(X)$ from  \eqref{eq:huge-term-for-detsharp-spelled-out}. Then $\taydetsharp(X)$ is a division free circuit such that every $X$ variable in it is a product of $z$. Thus, by Claim \ref{cla:Coeff-zn-homogeneous-comp-n} (and its proof) we can assume that $\taydetsharp(X)$ is specified already as a syntactic homogeneous circuit when given as input to the homogenization algorithm\iddo{}, meaning that the algorithm will output the same syntactic homogeneous circuit it got as an input (or a single homogeneous component as in part \ref{it:homogeneous-algorithm-part3} of the output). Thus, by Claim \ref{cla:homogeneous-Fii=Fi} there exist \PCZ-proofs of:
\begin{gather}\label{eq:homogeneous-penultimate1}
(\taydetsharp(X))^{(i)}=0,\text{ for all $i<n$,}\\
(\taydetsharp(X))^{(n)}=\taydetsharp(X)\,.\label{eq:homogeneous-penultimate2}
\end{gather}
The same argument works for $\taydetsharp(Y)$.
We proceed as follows.

\iddo{}

By Lemma \ref{lem:bounded degree} we have a \PCZ-proof of 
\begin{equation}\label{eq:homogeneous-end1}
\sum\nolimits_{i=0}^{2n} \hci {\left(\taydetsharp(X)\cd
 \taydetsharp(Y)\right)} =   \sum\nolimits_{i=0}^{2n} \sum_{l+j=i \atop 0\le l\le 2n ,0\le j\le 2n} \taydetsharp(X)^{(l)}\cd \taydetsharp(Y)^{(j)}\,. \end{equation}
This equals
\begin{multline}\label{eq:homogeneous-end2}
\sum\nolimits_{i=0}^{2n} \sum_{l+j=i \atop 0\le l\le n ,0\le j\le n} \taydetsharp(X)^{(l)}\cd \taydetsharp(Y)^{(j)} + \\
\sum\nolimits_{i=0}^{2n} \sum_{l+j=i \atop n< l\le 2n ,n< j\le 2n} \taydetsharp(X)^{(l)}\cd \taydetsharp(Y)^{(j)} \,,
\end{multline}
where the rightmost big term is a sum of only zeros, by construction of the homogeneous circuits (since $\taydetsharp(X)$ and $\taydetsharp(Y)$ are specified as already syntactic homogeneous circuits when input to the homogenization algorithm $\taydetsharp(X)^{(l)}=\taydetsharp(Y)^{(l)}=0$, for all $l>n$). We are thus left with the leftmost big sum in \eqref{eq:homogeneous-end2}. We proceed with
\begin{align*}
&\sum\nolimits_{i=0}^{2n} \sum_{l+j=i \atop 0\le l\le n ,0\le j\le n} \taydetsharp(X)^{(l)}\cd \taydetsharp(Y)^{(j)}  \\
& ~~~~~~~~~~~~~~~~~~~~~~~~~~~~~~~~~~~~~~~~~~~~~~ = \sum\nolimits_{i=0}^{n} 
    \taydetsharp(X)^{(i)}\cd\sum\nolimits_{i=0}^{n} \hci{\taydetsharp(Y)}  \\
&~~~~~~~~~~~~~~~~~~~~~~~~~~~~~~~~~~~~~~~~~~~~~~ = \taydetsharp(X)\cd \taydetsharp(Y) ,
\end{align*}
where the first equation is by rearrangement and the second equation is by \eqref{eq:homogeneous-penultimate1} and \eqref{eq:homogeneous-penultimate2}. By summing  \eqref{eq:homogeneous-end-single-h-components} for all $i=0,\dots,2n$ and using similar reasoning for $\taydetsharp(XY)$\iddo{}, we conclude that there exists a \PCZ-proof of  
$$
\taydetsharp(X)\cd \taydetsharp(Y)=\taydetsharp(XY).
$$

\medskip The fact that in the above \PCZ-proof every proof-line is a sum of syntactic homogeneous circuits in which every node appears with its syntactic-degree upper bound $\degub$ stems from the construction.
\iddo{}
\end{proof}

\iddo{}



\section{Balancing Algebraic Circuits and Proofs in the Theory}\label{sec:balancing-algebraic-circ-in-theory}



We start by providing some background and overview of the \NCTwo-algorithm for constructing a balanced circuit, given as an input  an upper bound (in unary) on the syntactic-degree of the input circuit. 



The input to the algorithm is $C,d$ where $C$ is a syntactic-homogeneous circuit and $d$ is an upper bound on the syntactic-degree \iddo{}of $C$ in unary. 
 

The output of the algorithm is a balanced circuit denoted $\tr  C$ computing the polynomial $\widehat C$. That is, if $C$ has size $s$, then the depth of $\tr C$ is $O(\log s\log d+\log^2 d)$ and the size of $\tr C$ is \poly$(s,d)$.

The algorithm itself follows the general scheme of Valiant \emph{et al.} \cite{VSB+83} that proceeds by induction on the logarithm of the degree of the polynomial computed by the circuit, however there are  differences  that help us fit the algorithm in \FNCTwo\ (for a very clear exposition of the \cite{VSB+83} algorithm we refer the reader to  \cite{RY08-balancing} (cf.~\cite{HT12}), though our treatment is self contained). 

Specifically, we use a preprocessing step to record in advance for every pair of nodes $w$ and $v$ if $w$ is in the scope of the circuit rooted by $v$. Furthermore, in the first stage of the algorithm (corresponding to the base case of the Valiant et al.~algorithm) we need to compute the coefficients of certain linear forms computed by possibly non-balanced circuits. We show that both the preprocessing step and the  first stage  of the algorithm are \ACZ-reducible to matrix powering; where matrix powering is known to be computable in \VNCTwo\ (cf.~\cite{CF10}). 
Another difference is that while  Valiant et al.~\cite{VSB+83} use the notion of degree of a node, and Hrube\v s and Tzameret \cite{HT12} use the syntactic-degree of a node, we are going to use the relaxed notion of syntactic-degree upper bound  $\degub(v)$ of a node $v$ introduced in Section \ref{sec:homogenization-of-proofs} and its variant $\degubp(v)$ (defined below).


\medskip



\TODO{I think we need to replace the scalars by variables in the start of the process, namely, before homogenization. But then we need to make sure that the division elimination goes through. The reason is that in the balancing algorithm we need to know the syntactic-degree of nodes \textbf{when scalars are replaced by variables}; But to know the syntactic-degree of nodes (we can't compute it in the theory since computing syntactic-degrees directly is beyond \NCTwo) we need the homogenization process.
\\ 
But actually we can also replace them after the division elimination: we can have a division free circuit and then apply the homogenization algorithm to get the syntactic-degrees of nodes.   }

\TODO{Also, I think it is easier if we replace also 0,1 scalars in the circuit with new variables. Otherwise you need to evaluate 0,1 circuits in \NCTwo\ and I don't know how to do it. I think this won't be a problem in the VSBR algorithm we use, since we use syntactic-degrees and not degrees, and so we don't need to forbid the circuits from containing nodes computing the zero polynomial (in fact, we must allow such polynomials). (On the other hand, we must be careful in the division elimination with nodes computing 0.)}


\begin{notation} Recall that we now only work with division free circuits. We use the following notation throughout this section:  $F,C$ are circuits and $\widehat F, \widehat C$ are the corresponding polynomials they compute.\mar{already defined notation?} For convenience we denote by $f$ the polynomial $\widehat F$. For a node $v$ in $F$ we write $F_v$ to denote the subcircuit rooted at $v$ and $f_v$ denotes the polynomial $\widehat {F_v}$. We write $ u\in F $ to mean that $ u $ is a node in the circuit $F $.
\end{notation}

We will need to construct with an \FNCTwo\ algorithm some linear polynomials computed by $F_v$, whenever $v\in F$ and $\degub(v)\le 1$ as well as the linear polynomials $\partial wf_v$ whenever $0\le \degub(v)-\degub(w)\le 1$ (see below). However, we cannot directly compute the integer  coefficients in these linear polynomials because their (sub-)circuits are not balanced (and we apparently cannot evaluate circuits of high-depth in \VNCTwo). 

We show how to compute the linear polynomials we need in in Lemma \ref{lem:variable-free-nc2-eval}  and Lemma \ref{lem:construct-partial-w-Fv}. To facilitate these lemmas we need to treat scalar nodes  $c\in \Z$ occurring in the circuit as if they are variables (and hence even circuits with only scalars get balanced  throughout the balancing algorithm). Formally, this means defining their syntactic-degree as 1 instead of 0, as follows (so that both variables and scalars are now treated as syntactic-degree 1 circuits).


Denote by $\degubp(\cd)$ the syntactic-degree upper bound defined similar to $\degub(\cd)$, except that scalar nodes are associated with syntactic-degree upper bound 1 (instead of 0) in the algorithm for homogenizing circuits shown in Section \ref{sec:homogenization-of-proofs}. 
Note that any circuit $F_v$ rooted by the node $v$ such that $\degubp(v)\le 1$ cannot contain product nodes (as this would make $\degubp(v)>1$ by definition). In Lemma \ref{lem:variable-free-nc2-eval} we show how to compute in \FNCTwo\ a circuit with no product nodes. 

Also note that it may happen that a node $v$ in a circuit has polynomially bounded  $\degub(u)$ but exponential large $\degubp(u)$, for example in case we have a linear chain of $n$ products computing $(((2)^2)^2\dots^2)=2^{2^n}$.
We deal with this problem 
in Lemma \ref{lem:scalars-substitution-keeps-poly-deg}.

\newcommand{\prt}[2]{f(#1;#2)}

\begin{definition}[Partial derivative polynomial  {\ensuremath{\partial w f_v}}]
\label{def:partialfvw}
Let $ w,v $ be two nodes in $ F $. We define the \emph{partial derivative of $F$ with respect to $w$}, denoted ${\ensuremath{\partial w f_v}}$, as the following \emph{polynomial}:
\begin{equation}\label{eq:partialfvw}
\partial w f_v :=
                                  \left\{
                                                                \begin{array}{ll}
                                                                  0, & \hbox{if $ w\not\in F_v$,} \\
                                                                  1, & \hbox{if $ w=v $,  and otherwise:} \\
                                                                 \partial w f_{v_1} + \partial w f_{v_2}  , & \hbox{$ v=v_1+v_2$;} \\
                                                                  (\partial w f_{v_1}) \cd f_{v_2} , & \hbox{if either $ v=v_1\cd v_2 $} 
~ \hbox{and }  \degubp(v_1)\geq\degubp(v_2) \\                                                                                &\hbox{~~~~~~~~~or $ v=v_2\cd v_1 $ and  $ \degubp(v_1)> \degubp(v_2) $.}                                                                \end{array}
                                  \right.
\end{equation}
\end{definition}

The idea behind this definition is the following: let $ w,v $ be two nodes in $ F $ and assume   that  $\degubp(w) > \frac{\degubp(v)}{2} $ (we will use $\partial w f_v$ only under this assumption\iddo{}).
  Then for any product node $v_{1}\cdot v_{2}\in F_{v}$, $w$ can be a node in at most one of $F_{v_{1}}, F_{v_{2}}$, namely the one with the higher syntactic-degree. If we replace the node $w$ in $F_{v}$ by a new variable $z$ that does not occur in $F$, then $F_{v}$ computes a polynomial $g(z,x_{1},\dots,x_{n})$ which is \textit{linear }in $z$, that is $g(z,x_{1},\dots,x_{n})=h_0\cd z+h_1$ for some polynomials $h_0,h_1$ in the $x_1,...,x_n$ variables, and $\partial w f_{v}=h_0$. Namely, $\partial w f_{v}$ is the standard partial derivative $\partial z g$.\iddo{}




\begin{proposition}\label{prop:deg_fv-deg_w}
Let $ w,v $ be two nodes in a syntactically homogeneous \iddo{}circuit $ F$ such that $\degubp(w)>\frac{\degubp(v)}{2}$. Then the polynomial $ \partial w f_{v}$ has degree\iddo{} at most $\degubp(v)-\degubp(w)$.
\end{proposition}


\begin{proof}
%
By induction on the size of $F_{v}$.
\Base \(F_v\) is a single node. If $F_{v}=w$ then {\ensuremath{\partial w f_v =1}}, and so $\degubp(\partial w f_{v})=0=\degubp(v)-\degubp(w)$ and the claim holds. If $F_{v}\neq w$ then $\partial w f_{v}=0$ and the claim holds similarly.
\induction

\case 1 $v=v_1+v_2$. Then $\partial w f_{v} = \partial w f_{v_1}+\partial w f_{v_2}$, and by induction hypothesis the degrees of $\partial w f_{v_1}$ and $\partial w f_{v_2}$ are at most $\degubp(v_1)-\degubp(w)$ and $\degubp(v_2)-\degubp(w)$, respectively. Since $F$ is syntactically homogeneous $\degubp(v)=\degubp(v_1)=\degubp(v_2)$ and so the degree of $\partial w f_{v} $ is at most $\degubp(v)-\degubp(w)$.

\case 2 $v=v_1\cd v_2$. Assume that $\degubp(v_1)\ge \degubp(v_2)$. Then ${\ensuremath{\partial w f_v}} =  ({\ensuremath{\partial w f_{v_1}}})\cd f_{v_2}$ and by induction hypothesis the degree of {\ensuremath{\partial w f_{v_1}}} is at most \(\degubp(v_1)-\degubp(w)\). Thus, the degree of {\ensuremath{\partial w f_v}} is at most \(\degubp(v_1)+\degubp(v_2)-\degubp(w)=\degubp(v)-\degubp(w)\). The case where \(\degubp(v_1)<\degubp(v_2)\) is similar.
\end{proof}

%

\begin{mycomment} We have defined {\ensuremath{\partial w f_v}} as a \emph{polynomial}. Below we shall construct (polynomial-size and balanced) \emph{circuits} \tr{{\ensuremath{\partial w f_v}}} that compute the polynomial {\ensuremath{\partial w f_v}}. We will make sure that the construction of \tr{{\ensuremath{\partial w f_v}}} is correct in the sense that it computes {\ensuremath{\partial w f_v}} and also that it has a syntactic-degree at most $\degubp(v)-\degubp(w)$. The correctness of the construction follows from  \cite{VSB+83} (see also \cite{RY08-balancing, HT12}) where in our construction  the notion of a syntactic-degree is used, instead of the notion of degree.
\end{mycomment}

\ind\textbf{Overview of the balancing algorithm:}
Let $ F $ be a syntactic-homogeneous  arithmetic circuit of syntactic-degree $ d $.
For every node $ v\in F $ we introduce the corresponding \emph{node} $ [F_v] $ in $ [F] $ (intended to compute the polynomial $\widehat {f}_v $); and for every pair of nodes $ v,w \in F $ such that $ \degubp(w) >\frac{\degubp(v)}{2}$, we introduce the \emph{node} $ [\partial w f_{v} ]$ in $ [F] $ (intended to  compute the polynomial $ \partial w f_{v} $).
 Note that given a syntactic-homogeneous circuit $F$, we can assume that every node comes with a number that denotes its syntactic-degree---this stems from the \FNCTwo\ \iddo{} algorithm for homogenization in Section \ref{sec:homogenization-of-proofs}; but note that  according to this algorithm circuits that compute zero may be assigned \emph{higher} syntactic-degrees than they actually possess. Since we are given an upper bound on the  the syntactic-degree of the circuit in advance this  will not interfere with  the algorithm.\iddo{} 

The algorithm starts with a preprocessing step that determines some properties of the circuit graph.
Then it  proceeds in steps $ i=0,\dots,\lceil \log d \rceil $. In each step \(i\) we construct:
\begin{enumerate}
\item  Circuits computing $ f_v $, for all nodes $ v $ in $ F $ with $ 2^{i-1}<\degubp(v)\le 2^i $;
\item  Circuits computing $ \partial w f_{v}  $, for all pairs of nodes $ w,v $ in $ F $ with $ 2^{i-1}<\degubp(v)-\degubp(w)\le 2^i $ and $ \degubp(w)>\frac{\degubp(v)}{2}$.
\end{enumerate}
 Each step adds depth  $ O(\log s) $ to the new circuit, which at the end amounts to a depth $ O(\log^{2}d+\log d \cd \log s ) $ circuit. Furthermore, each node $ v $ in $ F $ adds $ O(s) $ nodes in the new circuit and each pair of nodes $ v,w $ in $ F $ adds  $ O(s) $ nodes in the new circuit. This   amounts finally to a circuit of  \TODO{check!} size $ O(s^3)$.

The preprocessing step and step $i=0$ are done in \FNCTwo\ as they both use matrix powering (in fact the class \DET, which is the \ACZ-closure of matrix powering, suffices here). Each other stage constructs a group of nodes (namely, a part of the circuit having depth $O(\log s)$). 
Steps $i=1$ to $i=\lceil \log d\rceil$ are done in \FACZ\ by constructing the nodes and wiring simultaneously. Thus overall the balancing algorithm is in \FNCTwo. 


\subsection{Preliminaries for the Balancing Algorithm}

\iddo{}

\iddo{}

\begin{lemma}[in \VNCTwo]\label{lem:variable-free-nc2-eval}
Given a (division free) algebraic circuit $F$ of size $s$ with no product gates, there exists a depth \(O(\log n)\) circuit
computing $\widehat F$ of size $\poly(s)$, for $n$ the number of variables.\footnote{Notice that if the input algebraic circuit $F$ was a formula instead of a circuit, it would be trivial to output the balanced formula computing $\widehat F$: simply build a balanced binary tree whose leaves are all the variables occurring in $F$ (variables that occur more than once should also occur more than once in the resulting formula). Also, notice that although there are no scalars in $C$, a monomial can occur with a coefficient in $\widehat C$ different from 1.}
\end{lemma}

\begin{proof}
\iddo{}By assumption, the circuit $F$ computes a big sum of variables, where a variable can occur with an integer coefficient.  We will now represent circuits with unbounded fan-in as adjacency matrices. We first construct an upper triangular matrix $A=\{A_{ij}\}_{i,j\in[s]}$ that represents $F$: for every $j>i\in[s]$,\, $A_{ij}$ \emph{is labeled with the number of edges from node $i$ to node $j$ in the circuit}. In the initial  stage, $A$ is a 0-1 matrix because every node $i$ can have at most one directed edge to node $j$. 
This construction is done already in \VZ.

Given such a  matrix \(A\) representing $F$, the algorithm simply computes \(A^{s}\). The matrix $A^s$ has $c$ on its $(i,r)$th entry iff the number of different paths from node $i$ to the output gate $r$ is $c$. Thus, we can consider the matrix $A^s$ as corresponding to a depth $1$ circuit:  each leaf $i$ in this circuit represents the input variable $x_i$ or a scalar $k\in\Z$, and is connected to the root $r$ of the original circuit with a single edge labelled with some integer $c$; this integer $c$ is the total number of different paths in the original circuit leading from the input node $x_i$ or a scalar $k\in\Z$ to the root. Thus $cx_i$ or $ck$ is the contribution of the input node $x_i$ or the scalar node $k$ to the linear polynomial $\widehat F$. It  is thus immediate to construct a circuit (of depth $O(\log n)$) that computes  the linear polynomial  $\widehat F$: simply construct a big sum of the $c x_i$'s and $ck$'s.   

The fact that matrix powering is definable in \VNCTwo\ is shown in Cook and Fontes \cite{CF10}.\iddo{}

%
%
\end{proof}

We will also need the following two lemmas:
\begin{lemma}\label{lem:determine-w-in-Fv}
There is a \SigOneB-definable function in \VNCTwo\ for deciding, given a circuit $C$ and two nodes $w,v$ in $C$, if $w$ is in $F_v$.
\end{lemma}

\begin{proof}
This is similar to Lemma \ref{lem:variable-free-nc2-eval} above. We first construct the adjacency matrix $A_C$ of the circuit $C$ as a directed graph: the dimension of  $A_C$ is $s\times s$ with $s$ being the number of nodes in $C$, each entry in $A_C$ is of number sort, and $A_C(w,u)$ is $1$ iff $w$ has a directed edge towards $u$ or $w=u$, and 0 otherwise.

Then, $w$ has a directed path to $v$ iff $A_C^s(w,v)\neq 0$, where matrix powering is definable in \VNCTwo\ as mentioned above. 
\end{proof}

\iddo{}    
\begin{lemma}\label{lem:construct-partial-w-Fv}
There is a  \SigOneB-definable function in \VNCTwo\  whose input is  a (division free) circuit $F$ with $n$ variables and a pair of nodes  $w,v$ in $F$ where  $w$ is in $F_v$ and $0\le \degubp(v)-\degubp(w)\le 1$, and whose output is an $O(\log(n))$-depth circuit computing $\partial w f_v$.
\end{lemma}

\begin{proof}
In case $v=w$ we output the circuit $1$. Otherwise, first note that since $ 0\le \degubp(v)-\degubp(w)\le 1 $, either $\degubp(v)\le 1$ or $\degubp(v)\ge 2$ and $\degubp(w)>\degubp(v)/2$. Hence, by Proposition \ref{prop:deg_fv-deg_w}, the polynomial $ \partial w f_{v} $ is a linear polynomial $ a_1x_1+\dots+a_nx_n+b $. Therefore, it remains to show how to construct the circuit that computes this linear polynomial in \VNCTwo.

\smallskip
\noindent \textbf{Fact 1}: by definition of $\degubp$, for every node $r$ in $F$ we have $\degubp(r)\ge 1$. Hence, for every product gate $u=t\cd s$ we have $\degubp(u)=\degubp(t)+\degubp(s)\ge 2$.
\smallskip 

\noindent \textbf{Fact 2}: 
there cannot be a product gate 
$ u$ in $F_v$ such that $w$ has \emph{two different paths} directed from $w$ to $u$ (recall that edges are directed from leaves to root). 

This is because otherwise  
$\degubp(v)\ge \degubp(u)\ge 2\degubp(w)\ge 2$ (the last inequality is by the fact above), and hence  $\degubp(v)-\degubp(w)\ge \degubp(w)\ge 2 $ in contrast to the assumption. 
\smallskip

\noindent \textbf{Fact 3}: 
Let $\rho$ be a path from $v$ to $w$ (including $v$ and excluding $w$) in $F_v$. Then there exist at most one product gate in $\rho$. 

The reason is as follows: assume there are more than one product gates in $\rho$ (occurring ``above'' $w$). By Fact 1 every such product gate in $\rho$ increases the syntactic-degree upper bound $\degubp$ along $\rho$ by at least 1. Hence, $\degubp(v)\ge \degubp(w)+2$ in contrast to the assumption that $\degubp(v)-\degubp(w)\le 1$. 


We thus conclude that every product gate $u\neq w$ in $F_v$, either does not have $w$ in its scope, or is the only product gate on the path from $w$ to $v$ along $u$. Let $u=t\cd s$ be a product gate in $F_v$ that has $w$ in its scope, and assume without loss of generality that $F_t$ has $w$ in its scope and $F_s$ does not (by Fact 2 it cannot be that both have $w$ in their scope). We argue that $F_s$ \emph{has no product gates}. Otherwise, by Fact 1   $\degubp(s)\ge 2$ 
and so $\degubp(v)\ge \degubp(u)\ge \degubp(w)+\degubp(s)\ge \degubp(w)+2$ in contrast to the assumption $\degubp(v)-\degubp(w)\le 1$.   

Let $U$ be the set of all product gates $u=t_u\cd s_u$  in $F_v$ such that $F_{t_u}$ has  (without loss of generality) in its scope $w$. The above arguments imply that the polynomial  $ \partial w f_{v} =\sum_{u\in U} \widehat F_{s_u}$ and that there are no product gates in the $F_{s_u}$'s. But the set $U$ is easily \SigZB-defined in \VZ. And by Lemma \ref{lem:variable-free-nc2-eval} we can thus construct a $O(\log n)$ depth circuit, for $n$ the number of variables, computing the  sum $\sum_{u\in U} \widehat F_{s_u}$.  
\end{proof}


\iddo{} 
%


\medskip
%


\newcommand{\PCsf}{\ensuremath{\PC(Y_{\Z})}}
\subsubsection{Taking Care of Nodes with High $\degubp$ Measure}\label{sec:taking-care-of-high-degubp-nodes}
In this technical section we make sure that in the circuits we consider nodes have polynomially bounded $\degubp$ measure. For this purpose we show that the identities proved can be assumed to be of polynomial syntactic-degree (irrespective of the other identities between circuits appearing throughout the \PCZ-proofs), and then apply high syntactic-degree eliminations in proofs (Theorem \ref{thm:homogenize-proofs}; now with $\degubp$ replaced for $\degub$).
 
First note the following:

\begin{fact}\label{fac:degubp}
All the statements about proof-construction and transformations that we presented up to this point for $\degub$ also holds true for $\degubp$.
\end{fact} 

Fact \ref{fac:degubp} holds because of the following: let $C$ be a circuit and let $C'$ be the sum of its syntactic homogeneous components with all nodes in $C'$ appearing with their syntactic-degree upper bound. Let $C_Y$ be $C$ in which we substitution every scalar leaf $c\in\Z$ to a new variable $y_c$. Then the output of the homogenization algorithm on the input $C_Y$ results in a sum of syntactic homogeneous components of $C_Y$ such that every node appear with its syntactic-degree upper bound \degubp. If we now substitution back the scalars $c\in\Z$ for the variables $y_c$ we get the circuit $C'$ in which every node appears with its syntactic-degree upper bound $\degubp$ (instead of $\degub$).


\begin{lemma}[in \VNCTwo] \label{lem:scalars-substitution-keeps-poly-deg}
%
Given a positive natural number $n$ there exists a $\PCZ$-proof of the determinant identities \eqref{eq:1} and \eqref{eq:2}, where the determinant in \eqref{eq:1} and \eqref{eq:2} is written as the division free circuit denoted $\taydetsharpprime(A)$, for $A$, the \nbyn\ symbolic matrix $X$ or $Y$, or their product $XY$, or a symbolic triangular matrix $Z$. Moreover,  in this proof every circuit is a sum of syntactic homogeneous circuits in which every node $u$ appears with its syntactic-degree upper bound \degubp, and $\degubp(u)=O(n)$.

\iddo{}  
\end{lemma}

\begin{proof}
Let $\pi$ be the  $\PCZ$-proof in Corollary 
\ref{cor:exist-PC-proof-of-det-tay-sharp-identities}. Every node $u$ in $\pi$ appears with its syntactic-degree upper bound $\degub(u)$. By Fact \ref{fac:degubp} we could have assumed that every node appears with the syntactic-degree upper bound \degubp\ only that we need to make sure that \degubp\ for all nodes in $\pi$ are polynomially bounded. We do this as follows.

By inspection of $\taydetsharp(X)$, we will show that we can construct in \VNCTwo\ a \PCZ-proof of $\taydetsharpprime(X)=\taydetsharp(X)$, where $\taydetsharpprime(X)$ is a division free circuit with all nodes $u$ having $\degubp(u)=O(n)$. We then homogenize this \PCZ-proof using Theorem \ref{thm:homogenize-proofs} to get rid of all nodes with high $\degubp$ measure. Combining this proof with the \PCZ-proof in Corollary \ref{cor:exists-corr-up-to-deg-n-PC-proof-of-dettaysharp} we obtain the desired proof.

\iddo{}
\iddo{}

We start by identifying a property that will help us to determine the $\degubp$ measure of nodes.
\iddo{}

\newcommand{\property}{\ensuremath{\kreuz}}

\begin{claim*}[Not necessarily in \VNCTwo]
Let $F$ be a circuit (possibly \emph{with} division) over the variables $x_1,x_2,\dots$. Assume that $F$ has the following property: 
\begin{quote}
{\rm Property \property}: 
either $F$ is a scalar, or every scalar leaf $\alpha$  in $F$ is a child of a plus gate $u$ where $u:=\alpha+h$ and $h$ contains at least one  $x_i$ variable.\end{quote}
Then, $\degubp(F)= \degub(F)$.    
\end{claim*}
\begin{proofclaim}
Since we are not proving this claim in the theory we can  proceed by induction on the size of $F$. If $F:=x_i$, then we are done. If $F:=\alpha+v$, for some scalar $\alpha$,  since $v$ contains some $x_i$ variables, $\degubp(F)=\degubp(v)$. We also have  $\degub(v)=\degubp(v)$ by induction hypothesis, and $\degub(v)=\degub(F)$ since $v$ contains some $x_i$ variables, and we are done.
If $F:=v+w$ with both $v,w$ different from a scalar, then  the claim follows by induction hypothesis. 
Similarly, if $F:=v\cd w $ then the claim follows by induction hypothesis. 
\end{proofclaim}
\smallskip

We proceed to prove Lemma \ref{lem:scalars-substitution-keeps-poly-deg}. We show that $\cdeti(I_n+zX)$ has property \property, and thus $\degubp(\cdeti(I_n+zX))= \degub(\cdeti(I_n+zX))$. 

Consider the circuit $\cdeti(X)$. The only scalars in $\cdeti(X)$ are the 0-1 constants that occur in the identity matrix $I_{n-1}$ in \eqref{eq:def-inverse}. In \eqref{eq:def-inverse} the scalars 0-1 all appear in a subcircuit of the form $0+h$ or $1+h$ for some $h$ with $x_i$ variables as required by \property. Now consider  $\cdeti(I_n+zX)$, which results by replacing  the variables $x_{ij}$ in $\cdeti(X)$ by the term $0+zx_{ij}$ or by  $1+zx_{ij}$ in case $i=j$. But this substitution preserves the property \property.

We now argue that \VNCTwo\ can construct a \PCZ-proof of ${\taydetsharpprime(X)}=\taydetsharp(X)$. \iddo{}%
Recall  definition \eqref{eq:huge-term-for-detsharp-spelled-out} of $\taydetsharp(X)$ in Section \ref{sec:Reducing-the-Syntactic-Degree-of-the-Determinant-Polynomial}. 
%


\iddo{}

The only scalar in $\taydetsharp(X)$ is 1. In order to deal with nodes that have high \degubp\ values in $\taydetsharp(X)\rst\gamma$ it suffices to deal with nodes in $\Den(\cdeti(I_n+zX))$, since other parts in \eqref{eq:huge-term-for-detsharp-spelled-out} do not increase $\degubp$ more than by a polynomial factor. For this purpose we construct using the \SigZB-\textbf{COMP} axiom a \PIZ-proof that eliminates 1 from products with $1$ in $\Den(\cdeti(I_n+zX))$,  using the axiom $1\cd 1=1$. This is done by pointing specifically to where the $1$'s are in $\Den(\cdeti(I_n+zX))$; these are \SigZB-definable number functions in \VZ. We omit the details. \iddo{} 

%
\iddo{}
\end{proof}

\subsection{Formal Description of the Balancing Algorithm}\label{sec:formal-desc-of-balancing-algorithm}

For a syntactically homogeneous circuit $ G $ and a natural number $m$ let
\begin{equation}
 {\cal B}_m(G) :=           \big\{
 t \in G \;:\; t = t_1\cd t_2, \hbox{ where } ~\degubp(t)>m \mbox{ and } \degubp(t_1),\degubp(t_2)\le m                  \big\}.
\end{equation}
Notice that ${\cal B}_m(G)$ is a \SigZB-definable relation  in \VZ.
\medskip

\begin{note} In the construction of the balanced circuit of $F$, given nodes $v,w \in F$, the notation $\tr {F_v} $ and \tr{{\ensuremath{\partial w f_v}}} stand for \textbf{\textit{nodes}} (and not  circuits). When we write $\tr{F_v}:=C$ for a circuit $C$ we mean that the node $\tr{F_v}$ is defined to be the root of the circuit $C$, where $C$ possibly contains other (previously constructed) nodes like $[F_u]$, for some $u\in F$. In other words, the algorithm simply connects the node
$[F_v]$ to a circuit for which some of its leaves are already constructed nodes.
\end{note}
\medskip



\hrule
\medskip 
\noindent\textsc{\FNCTwo-Algorithm for Balancing a Circuit $F$ (Construction of  $\tr{F}$)}

\medskip
\hrule 

\begin{description}[font=\rm]
\item[\textbf{Input}:]
$F$ where $F$ is a sum of one or more syntactic-homogeneous circuits over the variables $x_1,\dots,x_n$, in which every node $u$ appears with its syntactic-degree upper bound $\degubp(u)$. 

\item[\textbf{Output}:]
A  circuit $\tr  F$ computing the polynomial $\widehat  F$. That is, if $s$ is the size of $F$ then depth$(\tr F)=O(\log s\log d+\log^2 d)$ and the size of $\tr F$ is ${\rm poly}(s,d)$, where $d=\degubp(F)$.
\end{description}


\renewcommand{\para}[1]
{\medskip\noindent\textbf{#1}}

\para{Preprocessing step:} 
For every pair of nodes $w,v$ we prepare a list that determines whether $w$ is in $F_v$. This is done by running in parallel for all pairs $w,v$ in $F$ the \NCTwo-algorithm in Lemma \ref{lem:determine-w-in-Fv} for checking if $w$ is in $F_v$ described 

\medskip

\noindent\underline{\textit{Step} $i=0$}:
\medskip 

\noindent 
\textbf{\textit{Part (a)}}:\footnote{This base case uses an \FNCTwo\ algorithm, but since it is done only in the base case, the whole algorithm still is in \FNCTwo.}
We construct the node $[F_v]$, for all nodes $v\in F$ such that $\degubp(v)\le 1=2^i$.


Let $v\in F$ be such that $\degubp(v)\leq1$. 
\iddo{}
\begin{claim}\label{cla:inside-balance-constrcution}
$\widehat {F_v} =a_1x_1+\dots+a_nx_n+\sum_{c\in J} b_c c$, for  $a_1,\dots,a_n, b_c\in\Z $ and $J\subset \Z$. Furthermore, there exists an \FNCTwo-construction that given $F$ constructs the  depth $O(\log n)$ circuit $a_1x_1+\dots+a_nx_n+\sum_{c\in J} b_c c$.
\end{claim}

\begin{proofclaim}
Since  $\degubp(v)\le 1$, there are no product gates in $F_v$. Thus, $F_v$ is a circuit with only plus gates, which means $\widehat F_v$ is as stated in the claim. By Lemma \ref{lem:variable-free-nc2-eval} (and Theorem \ref{thm:definable_func_of_vnctwo}) we can construct in \FNCTwo\ the circuit $a_1x_1+\dots+a_nx_n+\sum_{c\in J} b_c c$ (we do not evaluate the circuit). 
\end{proofclaim}

\iddo{} 
We define  
\begin{equation*}\label{eq:the-linear-case}                 
[F_v]:= a_1x_1+\dots+a_nx_n+\sum\nolimits_{c\in J} b_c c.
\end{equation*}

\medskip 

\noindent\textbf{\textit{Part (b)}}:
Let $ w,v $ be a pair of nodes in $ F $ with $ 2\degubp(w) >\degubp(v) $:\mbox{}
\QuadSpace

\case 1 Assume $ w $ is not a node in $ F_v $ (this can be checked using the list from the preprocessing step). Define
\[
                                                                [\partial w f_v]:=0.
\]

\case 2 
Assume that $ w $ is in $ F_v $ and $ 0\le \degubp(v)-\degubp(w)\le 1 $. Again, this is checked by the list from the preprocessing step, and since we the input circuit $F$ is assumed to contain the value of $\degubp$ for each node. 

Thus, by Proposition \ref{prop:deg_fv-deg_w}, the polynomial $ \partial w f_{v} $ is a linear polynomial $ a_1x_1+\dots+a_nx_n+b $. Using  Lemma \ref{lem:construct-partial-w-Fv} and similar notation and reasoning as Claim \ref{cla:inside-balance-constrcution} define 
\[
                                                                [\partial w f_v]:= a_1x_1+\dots+a_nx_n+\sum\nolimits_{c\in J} b_cc\,.
\]

\noindent\underline{\textit{Step} $i+1$}:

The construction in this step is done in \VZ, assuming we have the list from the preprocessing step above. 
\medskip 

\noindent
\textbf{\textit{Part (a)}}:
 Assume that for some $ 0 \le i \le \lceil \log(d) \rceil $:
\[
                                                                2^i < \degubp(v) \le 2^{i+1}.
\]
Put $ m=2^i $, and define (recall that here $ [\partial w f_v], [F_{t_1}] $ and $[F_{t_2}]$ are \emph{nodes})

\[
                                                                [F_v] := \sum_{t\in {\cal B}_m(F_v)\atop t=t_1\cd t_2} [\partial t f_{v}]\cd [F_{t_1}]\cd[F_{t_2}]\,.
\]

\medskip 

\noindent
\textbf{\textit{Part (b)}}: Let $ w,v $ be a pair of nodes in $ F $ with $ 2\degubp(w) >\degubp(v) $:\mbox{}
\QuadSpace

Assume that $ w $ is in $ F_v $ and that for some $ 0 \le i \le \lceil \log(d) \rceil $:
\[
2^i < \degubp(v)-\degubp(w)\le 2^{i+1} .
\]
\smallskip

Put $  m = 2^i + \degubp(w)$. Define:
\[
                                                                [\partial w f_{v}] := \sum_{t\in{\cal B}_m(F_v) } [\partial t f_{v}] \cd[\partial w f_{t_1}]\cd [F_{t_2}]\,,
\]
where here for every given $t\in{\cal B}_m(F_v)$, $t_{1}, t_{2}$ are nodes such that $ t=t_1\cd t_2 $ and $\degubp(t_{1})\geq\degubp(t_{2})$, or $t=t_{2}\cdot t_{1}$ and $\degubp(t_{2})<\degubp(t_{1})$.

\smallskip 

Finally, define $\tr{F}$ as the circuit with  output node $\tr{F_{u}}$, where $u$ is the output node of $F$.
\medskip \hrule 

\bigskip

By construction, the algorithm computes the correct output:\ the fact that $[F]$ has the correct depth stems from the construction as explained in the overview of the balancing algorithm above (see also \cite{VSB+83,RY08-balancing,HT12}). The fact that $[F]$ has the correct size stems from the fact that the algorithm is $\SigZB$-definable in \VNCTwo. The fact that $[F]$ computes $\widehat F$ is shown below by constructing in \VNCTwo\ a \PCZ-proof of $F=[F]$ for a syntactic homogeneous circuit $F$ (this stems from Lemma \ref{lem:simulation}; see again \cite{VSB+83,RY08-balancing,HT12}). \medskip 

Note that given an algebraic circuit $F$ with $\degubp(u)$, for all nodes $u$ in $F$, polynomially bounded, our balancing algorithm provides a way to balance $F$ already in \FNCTwo. As mentioned in  Section \ref{sec:overview}, by first balancing an input circuit and then evaluating it (assuming e.g. it is over the integers, as in the next section) this gives rise to an \NCTwo\ evaluation procedure for algebraic circuits of any depth (given as input an upper bound on their syntactic-degree in unary and assuming the syntactic degree \degubp\ of the circuit is polynomial) that is different from the previously known  algorithm by Miller \textit{et al.}~\cite{MRK88} (their algorithm does not require the syntactic-degree as input) and that of Allender \emph{et al.} \cite{AJMV98} (which is implicit in that work but can be extracted from the text \cite{All18}).

\subsection{Balancing  Proofs in \VNCTwo}\label{sec:balancing-proofs}

For balancing \PCZ-\emph{proofs} we need to show the proof-theoretic counterpart of the balancing algorithm. This is similar to the proof-theoretic counterpart of the homogenization theorem shown in Section \ref{sec:homogenization-of-proofs}.  
We start by showing some properties of the constructions of the base cases of the balancing algorithm described in Lemmas \ref{lem:variable-free-nc2-eval} and \ref{lem:construct-partial-w-Fv} that \VNCTwo\ can prove.  
\begin{lemma}[in \VNCTwo]\label{lem:prove-dv=1-case-for-balancing}
\begin{enumerate}
\item[(i)] Let $F$ be a circuit with no product gates and no scalars (and no division gates). Assume that $v=v_1+v_2$ is a node in $F$ such that $\degubp(v)\le 1$. Then, there exists a \PCZ-proof of $[F_v]=[F_{v_1}]+[F_{v_2}]$. 

\item[(ii)] Let $F$ be a circuit with no scalars and syntactic-degree $d$, and a pair of nodes $w,v$ in $F$, such that $w$ is in $F_v$ and $0\le \degubp(v)-\degubp(w)\le 1$. Then, there is a \PCZ-proof of 
\begin{align}
        [\partial w F_{v}]  &= [\partial w F_{v_1}] + [\partial w F_{v_2}], &\mbox{\ \ \ in case $ v = v_1+v_2 $;}\label{eq:wFv+}\\
        [\partial w F_{v}]  &= [\partial w F_{v_1}] \cd [F_{v_2}] ,& \mbox{\ \ in case $ v = v_1\cd v_2 $ and $\degubp(v_1)\ge\degubp(v_2)$ }\nonumber\\
           && \mbox{\ \ or $ v = v_2\cd v_1 $ and $\degubp(v_1)> \degubp(v_2)$.  }
        \label{eq:wFv*}
\end{align}
\end{enumerate}
\end{lemma}

\begin{proof}
Part (i). Consider Lemma \ref{lem:variable-free-nc2-eval}. The circuit  $[F_v]$ is constructed according to this lemma by first computing the integer coefficients of each of the input variables in the linear form computed by $F_v$. It thus suffices to prove (in \VNCTwo) that for every input variable $x_i$ in $F_v$, the coefficient of $x_i$ in $F_v$ equals the sum of the coefficients of $x_i$ in $F_{v_1}, F_{v_2}$. Assuming we can prove this, we can directly construct the \PCZ-proof of $[F_v]=[F_{v_1}]+[F_{v_2}]$. 

We use a result from Cook and Fontes \cite{CF10}, stating that the theory $V\# L$ which is contained in  $\VNCTwo$, \SigOneB-defines the string function ${\rm PowSeq}_\Z(n,s,A)$. This string function receives an \nbyn\ integer matrix $A$ and outputs a string coding the sequence of powers of $A$: $(A,A^2,\dots,A^s)$.

Let $F'_{v_1}:=F_{v_1}\cup\{(v_1,r)\}$ and $F'_{v_2}:=F_{v_2}\cup\{(v_2,r)\}$. That is, $F'_{v_1}$ is the (non-legit) circuit $F_{v_1}$ to which we add the directed edge from $v_1$ to the output node $r$ in $F_v$, and similarly $F'_{v_2}$, so that $F_v=F'_{v_1}\cup F'_{v_2}$. Assume that $A_v,A_{v_1},A_{v_2}$ are the 0-1 adjacency matrices of the circuits $F_v,F'_{v_1},F'_{v_2}$, respectively, where the dimensions of all the matrices all equal $s$, the number of nodes in $F_v$ and the $(u,w)$th entry in all three matrices corresponds to a directed edge from node $u$ to node $w$. Using $\SigZB$-induction on the power $i=1,\dots,s$, and using the strings $(A_v,A_v^2,\dots,A_v^s), (A_{v_1},A_{v_1}^2,\dots,A_{v_1}^s), (A_{v_2},A_{v_2}^2,\dots,A_{v_2}^s)$, we argue that for every input node $u$ in $F_v$
$$
A^i_v[u,r]=A^i_{v_1}[u,r]+A^i_{v_2}[u,r]\,,
$$ 
where $A[u,r]$ denotes the $(u,r)$th entry of the matrix $A$, and as before $r$ is the output node of $F_v$. \iddo{}

\medskip 
Part (ii). Here we  use the construction in Lemma \ref{lem:construct-partial-w-Fv}. 

\case 1  $ v = v_1+v_2 $. According to Lemma \ref{lem:construct-partial-w-Fv}, $[\partial w F_{v}]$ is defined as the  sum  $\sum_{u\in U} \widehat F_{s_u}$, where $U$ is the set of all product gates $u=t_u\cd s_u$  in $F_v$ such that $F_{t_u}$ has  (without loss of generality) in its scope $w$, and where we construct each $ F_{s_u}$ in the sum using Lemma \ref{lem:variable-free-nc2-eval}, similar to part (i). Similar to part (i) we proceed by the number \SigZB-induction on $i=1,\dots,s$, where $s$ is the size of $F_v$ to prove
$$
A^i_v[u,r]=\sum_{u\in U}A^i_{s_u}[u,r]\,.
$$ 
\iddo{}
\iddo{}

\case 2 $v=v_1\cd v_2$. This is similar to case 1. According to Lemma \ref{lem:construct-partial-w-Fv}  and using the terminology of case 1 above,  $[\partial w F_{v}]$ is defined as the sum $\sum_{u\in U} \widehat F_{s_u}$. Only that by assumption, the only product gate that has  $w$ is in its scope must be $v$ itself (because there can be no two nested product gates with $w$ in their scope by assumption $\degubp(v)-\degubp(w)\le 1$). Assume without loss of generality that $v_1$ has $w$ in its scope. Then, $v_2$ does not have $w$ in its scope (by assumption on degree, as explained in the proof of Lemma \ref{lem:construct-partial-w-Fv}). Thus, $[\partial w F_{v}]= \sum_{u\in U} \widehat F_{s_u}=\widehat F_{v_2}=1\cd [F_{v_2}]=[\partial w F_{v_1}]\cd [F_{v_2}]$.\iddo{}
\iddo{}
\end{proof}

Recall that the length number function  $\lceil \log_2(n)\rceil$ is \SigZB-definable function in \VZ\ (see \cite{CN10}). This is the main theorem of this section: 

\begin{theorem}[in \VNCTwo]\label{thm:main-balancing}

\begin{enumerate}

\item \label{bal-jedna} If $F$ is a sum of one or more syntactic homogeneous circuits, of size $s$ and depth $t$, such that $\degubp(F)=d$, then $F=\tr{F}$ has a $\PC$-proof of size $\poly(s,d)$ and depth $O(t+\log s\cd\log d +\log^2 d) $.

\item  \label{bal-dva} Let $\pi$ be a \PCZ-proof of $F=G$ of syntactic-degree at most $d$ and size $s$, and where every circuit is a sum of syntactic homogeneous circuits with every node appearing with its $\degubp$ value. Then, $\tr{F}=\tr{G}$ has a $\PC$-proof of size $\poly(s,d)$ and depth $O(\log s\cd\log d +\log^2 d)$.

\end{enumerate}
\end{theorem}

Theorem \ref{thm:main-balancing} will be proved analogously to Theorem \ref{thm:homogenize-proofs}: 
the proof is similar to the proof of Theorem \ref{thm:homogenize-proofs}, only that instead of using Lemma \ref{lem:bounded degree} we use the analogous Lemma \ref{lem:main simulation lemma} below that demonstrates some essential properties of  $\tr{F}$ that have short \PCZ-proofs.

\iddo{}

\begin{lemma}[in \VNCTwo]\label{lem:main simulation lemma}\label{lem:simulation}
Let $ F_{1}, F_2 $ be syntactic homogeneous circuits of syntactic degree at most $ d $ and size at most $ s $.
 Then, there exist $\PCZ$-proofs of:
\begin{align}
 [F_1\oplus F_2]  &= [F_1]+[F_2]\, , \ \ \ \     \label{eq:plus}\\
[F_1\otimes  F_2] & = [F_1]\cd[F_2] \,,    \label{eq:prod}
\end{align}
such that the proofs have size $\poly(s,d)$ and depth $O(\log d\cd \log s+\log^{2}d)$.
Furthermore, $\tr{z}=z$ has a constant-size proof whenever $z$ is a variable or an integer.
\end{lemma}

The proof of Theorem \ref{thm:main-balancing} is deferred to Section \ref{sec:proof-of-balancing-theorem}.

\subsubsection{Proof of Lemma \ref{lem:simulation}}\label{sec: lemma} We now prove Lemma \ref{lem:simulation}. 
The proof is similar to Lemma 4.4 in \cite{HT12}, except that we use $\degubp(\cd)$ instead of  syntactic-degrees $\deg(\cd)$ and that we construct the \PCZ-proof in \FNCTwo\ instead of by induction on the structure of $F$ (which would have necessitate  using \SigOneB-induction).   

The statement concerning $\tr{z}=z$ is clear: if $z$ is an integer, $\tr{z}$ and $z$ are the same circuit. If $z$ is a variable, $\tr{z}$ is the circuit $1\cd z$. \iddo{}

We need to construct proofs of  equations (\ref{eq:plus}) and (\ref{eq:prod}).


Let $m(s,d)$ and $r(s,d)$ be functions such that for any circuit $ F $ with $\degubp(F)=d$ and size $s$, $\tr{F}$ has depth at most $ r(s,d) $ and size at most $ m(s,d) $. Since the balancing algorithm shown above is \SigZB-definable in \VNCTwo\  we can  choose
\[ m(s,d)=\poly(s,d) \mbox{\ \ and \ \ } r(s,d)=O(\log^{2}d+\log d \cd\log s).\]

\begin{notation}
In the following, $[F_v]$ and $[\partial w F_v]$ will denote \emph{circuits}: $[F_v]$ and $[\partial w F_v]$ are the subcircuits of $[F]$ with output nodes $[F_v]$ and $[\partial w F_v]$, respectively; the defining relations between the nodes of $[F]$ (see the definition of $[F]$  above) translate to equalities between the corresponding circuits. For example, if $v$ and $m$ are as in part (a) Case 2, of the definition of $ \tr{F} $, then, using just the axioms C1 and C2, we can prove
\begin{equation}\label{eq:hosafa basof}
[F_v] = \sum_{t\in {\cal B}_m(F_v)} [\partial t F_{v}]\cd [F_{t_1}]\cd[F_{t_2}]\,.
\end{equation}
Here, the left hand side is understood as the circuit $ \tr{F_v} $ in which $ \tr{\partial t F_v}, \tr{F_{t_1}}, \tr{F_{t_2}} $ appear as \emph{subcircuits}, and so can share common nodes, while on the right hand side the circuits have \emph{disjoint nodes}.
Also, note that if $F$ has size $s$ and degree $d$, the proof of  (\ref{eq:hosafa basof}) has size  $O(s^2m(s,d))$ and has depth $O(r(s,d))$.
We shall use these kind of identities in the current proof.\iddo{}
\end{notation}

Let $ \lambda(s,i) $ be a function such that 
\begin{equation}\label{eq:recurrence}
   \lambda(s,0)=O(s^4) \quad \hbox{and} \quad \lambda(s,i)\le O(s^4\cd m(s,d))+ \lambda(s,i-1).
\end{equation}
Recurrence (\ref{eq:recurrence}) implies $\lambda(s,d)=\poly(s,d)$.

The following proposition (which is similar to Proposition 4.10 in \cite{HT12}) suffices to conclude the lemma (it is enough to take $ F $ in the statement as either $ F_1\oplus F_2 $ or $ F_1\otimes F_2 $, and $ v $ as the root of $ {F} $).
\begin{proposition}[in \VNCTwo]\label{prop:lengthy-lemma}
Let $ F $ be a syntactically homogenous circuit of syntactic degree at most $d $ and size  $ s $. For every $ i=0,\dots,\lceil \log d \rceil $ there exists  a $\PC$ proof-sequence $ \Psi_i $ of size at most $ \lambda(s,i) $ and depth at most $ O(r(s,d)) $, such that the following hold:
\QuadSpace

\ind\textbf{Part (a):}
For every node $ v\in{F} $ with
\begin{equation}\label{eq:od ehad}
 \degubp(v)\le 2^i,
\end{equation}
$ \Psi_i $ contains the following equations:
\begin{align}
        [F_v] &= [F_{v_1}]+[F_{v_2}]\ , \mbox{\ \ \ \ \ in case $ v = v_1+v_2 $,\ \ \ \ \  and} \label{eq:Fv+}\\
        [F_v] &= [F_{v_1}]\cd [F_{v_2}]\ , \mbox{\ \ \ \ \ in case $ v = v_1\cd v_2 $}.\label{eq:Fv*}
\end{align}

\ind\textbf{Part (b):} For every pair of nodes $ w\neq v\in{F} $, where $ w\in F_v $, and with
\begin{align}
   &
\degubp(v)-\degubp(w)\le 2^i \; \mbox{\ \ \ and}  \label{eq:nizkarti1}\\
        & 2\degubp(w)>\degubp(v),\label{eq:nizkarti2}
\end{align}
$ \Psi_i $ contains the following equations:
\begin{align}
        [\partial w F_{v}]  &= [\partial w F_{v_1}] + [\partial w F_{v_2}], &\mbox{\ \ \ in case $ v = v_1+v_2 $;}\label{eq:wFv+}\\
        [\partial w F_{v}]  &= [\partial w F_{v_1}] \cd [F_{v_2}] ,& \mbox{\ \ in case $ v = v_1\cd v_2 $ and $\degubp(v_1)\ge\degubp(v_2)$ }\nonumber\\
           && \mbox{\ \ or $ v = v_2\cd v_1 $ and $\degubp(v_1)> \degubp(v_2)$.  }
        \label{eq:wFv*}
\end{align}
\end{proposition}

\begin{proof}
Similar to previous constructions the idea is to construct all parts of the \PCZ-proof \emph{simultaneously} in \VNCTwo. This is done in an analogue manner to the balancing algorithm above. 
\bigskip 

\ind\underline{\textit{Step $i=0 $}}.
We need to devise the proof sequence $ \Psi_0 $.

\para{Part (a): proof of (\ref{eq:Fv+}).}  Let $ \degubp(v)\le 2^0 $. By definition, $\tr{F_{v}}= \sum_{i=1}^n a_i x_i +b$, where $a_{i}$'s are integers and $b$ is a sum of constant integers. Further, by construction  $\tr{F_{v}}$ does not contain product gates and thus $v=v_{1}+v_{2}$, and we need to prove only (\ref{eq:Fv+}). This stems from Lemma \ref{lem:prove-dv=1-case-for-balancing} part (i).



\para{Part (b): proof of (\ref{eq:wFv+}) and (\ref{eq:wFv*}).}
Similarly to part (a) above, this follows from Lemma \ref{lem:prove-dv=1-case-for-balancing} part (ii).
%
%
\FullSpace

Overall, $\Psi_{0}$ will be the union of all the above proofs, so that $\Psi_{0}$ contains all equations (\ref{eq:Fv+}) (for all nodes $ v $ satisfying (\ref{eq:od ehad})), and all equations (\ref{eq:wFv+}) and (\ref{eq:wFv*}) (for all nodes $ v,w $ satisfying (\ref{eq:nizkarti1}) and (\ref{eq:nizkarti2})).  \iddo{} The proof sequence $ \Psi_0 $ has size $\lambda(s,0)= O(s^4) $ and has depth $O(\log s)$.

\bigskip 

\ind\underline{\textit{Step $i+1$}}:
We wish to construct the proof-sequence $ \Psi_{i+1} $.

\para{Part (a): proof of (\ref{eq:Fv+}) and (\ref{eq:Fv*}).}
 Let $ v $ be any node in $ F $ such that
\[
         2^i < \degubp(v) \le 2^{i+1} .
\]

\case 1
Assume that $ v=v_1+v_2 $. We show how to construct the proof of $ [F_{v}]=[F_{v_1}]+[F_{v_2}]$. Let $ m=2^i $. From the construction of $ [\cd] $ we have:
\begin{align}\label{eq:101}
        [F_v]=[F_{v_1+v_2}]  = \sum_{t\in {\cal B}_m(F_v) } [F_{t_1}]\cd[F_{t_2}]\cd[\partial t (F_{v_1+v_2})] \, .
\end{align}
Since $\degubp(v_1)=\degubp(v_2)=\degubp(v)$, we also have
\begin{align}\label{eq: new1}
        [F_{v_e}]  = \sum_{t\in {\cal B}_m(F_{v_e}) } [F_{t_1}]\cd[F_{t_2}]\cd[\partial t (F_{v_e})], \quad \hbox{for } e\in \{1,2\} \, .
\end{align}

If $t\in {\cal B}_m(F_v)  $ then $ \degubp(t)> m = 2^i $. Therefore, for any $t\in {\cal B}_m(F_v)  $, since $ \degubp(v)\le 2^{i+1} $, we have $ \degubp(v)-\degubp(t) <  2^i$ and $ 2\degubp(t)>\degubp(v) $ and $ t\neq v $ (since $ t $ is a product gate). Thus, by construction,  the proof-sequence $ \Psi_i $ contains, for any $t\in {\cal B}_m(F_v)  $, the equations
\[
        [\partial t (F_{v_1+v_2})]  = [\partial t F_{v_1}]+[\partial t F_{v_2}] ,
\]
and we can compute the positions of these proof-lines in the string encoding of  $\Psi_i$  (using some natural encoding).  
Therefore, pointing to these proof-lines in $ \Psi_i $ as premises, we construct a \PCZ-proof that (\ref{eq:101}) equals:
\begin{equation}\label{eq:201}
\begin{split}
   &       \sum_{t\in {\cal B}_m(F_v) } [F_{t_1}]\cd[F_{t_2}]\cd([\partial t F_{v_1}]+[\partial t F_{v_2}])\\
   & =       \sum_{t\in {\cal B}_m(F_v) } [F_{t_1}]\cd[F_{t_2}]\cd[\partial t F_{v_1}]    +
                                               \sum_{t\in {\cal B}_m(F_v) } [F_{t_1}]\cd[F_{t_2}]\cd[\partial t F_{v_2}].
\end{split}
\end{equation}
If $ t\in {\cal B}_m (F_v) $ and $ t\not\in F_{v_1} $ then $ [\partial t F_{v_1}] = 0 $. Similarly, if $ t\in {\cal B}_m (F_v) $ and $ t\not\in F_{v_2} $ then $ [\partial t F_{v_2}] = 0 $. Hence we can prove
\begin{equation}\label{eq: biggerB}
\sum_{t\in {\cal B}_m(F_v) }[\partial t F_{v_e}]= \sum_{t\in {\cal B}_m(F_{v_{e}}) }[\partial t F_{v_e}],\quad\mbox{for $ e=1,2 $}.
\end{equation}
Thus, using (\ref{eq: new1}) we have that  (\ref{eq:201}) equals:
\begin{equation}\label{eq:421}
\begin{split}
           \sum_{t\in {\cal B}_m(F_{v_1}) } [F_{t_1}]\cd[F_{t_2}]\cd[\partial t F_{v_1}]   +
                                               \sum_{t\in {\cal B}_m(F_{v_2}) } [F_{t_1}]\cd[F_{t_2}]\cd[\partial t F_{v_2}] \\  = [F_{v_1}]+[F_{v_2}].
                                          \end{split}
\end{equation}
The above proof of (\ref{eq:421}) from $ \Psi_i $ has size $ O(s^2\cd m(s,d)) $ and depth $ O(r(s,d)) $.

The proof of Case 2 where  $ v=v_1\cd v_2 $, and  the proofs of Part (b) for equations (\ref{eq:wFv+}) and (\ref{eq:wFv*}) are similar to Case 1 above, and are identical to those cases in the proof of Proposition 4.10 in \cite{HT12}; like Case 1, the  only difference is that we construct with an \FNCTwo\ construction all the \PCZ-proofs $\Psi_i$ together, for every $i=0,\dots,\lceil \log d \rceil$, where in $\Psi_{i+1}$ we point to proof-lines that appear in $\Psi_i$ (whose position can be computed using a reasonable encoding scheme for proof-lines). For self containment we put these cases in the appendix Section \ref{sec:balancing-main-lemma-verbatim}. \end{proof}

This concludes the proof of Proposition \ref{prop:lengthy-lemma}, and hence of Lemma \ref{lem:simulation}.

\subsubsection{Proof of Theorem \ref{thm:main-balancing}}\label{sec:proof-of-balancing-theorem}

\begin{proof}[Proof of Theorem \ref{thm:main-balancing}]

\textbf{Part (i)}.  We use the balancing algorithm above and  Lemma \ref{lem:simulation} to construct for every node $v$ in $F$ a (part of) the proof of $\tr{F_v}=F_v$, simultaneously. We can use the balancing algorithm because  $F$ is a sum of syntactic homogeneous circuits with all nodes appearing together with their associated syntactic-degree upper bound $\degubp$ .  

\case 1 For a leaf $u$ we construct the equation $u=u$, which is correct since  $\tr{u}=u$. 
 
\case 2  For $v=v_{1}\circ v_{2}$, where $\circ\in\{+,\cdot\}$, Lemma 
 \ref{lem:simulation} gives $\tr{F_{v}}=\tr{F_{v_{1}}}\circ \tr{F_{v_{2}}}$. We then point to the equations $\tr{F_{v_{i}}}=F_{v_{i}}$, $i\in \{1,2\}$, which
gives a proof of  $\tr{F_{v}}={F_{v_{1}}}\circ {F_{v_{2}}}=F_{v}$.  

The proof has size $\poly(s,d)$. The depth of the proof never exceeds the depth of $F$ and the depth of the proofs of  $\tr{F_{v}}=\tr{F_{v_{1}}}\circ \tr{F_{v_{2}}}$.

\medskip 

\textbf{Part (ii)}. 
\iddo{}
We assumed that  $\pi$ is a \PCZ-proof of $F=G$ of syntactic-degree at most $d$ and size $s$, in which every circuit is a sum of syntactic homogeneous circuits with every node appearing with its $\degubp$ value. Similar to previous constructions, we are going to simultaneously construct a (part of a) \PCZ-proof of $\tr{F_1}=\tr{F_2}$ using  pointers to previous lines (that we can compute in parallel), for every proof-line $F_1=F_2$ in $\pi$. This resembles the proof structure of Theorem \ref{thm:homogenize-proofs}. Like part (i), we can use the balancing algorithm because by assumption each of the circuits $F_1, F_2$ is given to us as a sum of syntactic homogeneous circuits with all nodes appearing together with their associated syntactic-degree upper bound $\degubp$. 

Let $m_0$ and $k_0$ be  such that (\ref{eq:plus}) and (\ref{eq:prod}) have $\PCZ$-proofs of size at most $m_0$ and depth $k_0$, whenever $F_1\cplus F_2$, respectively, $F_1\ctimes F_2$ have size at most $s'$ and syntactic degree at most $d$. 

\smallskip 

\case 1 $F=H$ is an axiom of $\PCZ$.  Then, $\tr{F}=\tr{H}$ has a $\PC$-proof of size $c_1m_0$ and depth $c_2 k_0$, where $c_1, c_2$ are some constants independent of $s',d$. The axiom A1 is immediate and the axiom A10 follows from the fact that $[F]= \widehat F$, for $F=c$, $c\in\Z$.. The rest of the axioms are an application of Lemma \ref{lem:simulation}, as follows. Axioms C1 and C2 are already  the statement of Lemma \ref{lem:simulation}. For the other axioms, take, for example,
\[F_1\cdot (G_1+G_2)= F_{1}\cdot G_1+F_{1}\cdot G_2\, .\]
We are supposed to give a proof of
\[\tr{F_1\cdot (G_1+G_2)}= \tr{F_1\cdot G_1+ F_1\cdot G_2}\, ,\]
with a small  depth.
By Lemma \ref{lem:simulation} we have a $\PCZ$-proof
\[\tr{F_1\cdot (G_1+G_2)}= \tr{F_1}\cdot\tr{ G_1+ G_2} =  \tr{F_1}\cdot (\tr{G_1}+\tr{G_2})=\tr{F_1}\cdot \tr{G_1}+\tr{F_1}\cdot \tr{G_2}\,.\]
Lemma \ref{lem:simulation} gives again:
\[
        \tr{F_1}\cdot \tr{G_1}+\tr{F_1}\cdot \tr{G_2}
        =      \tr{F_1\cdot G_1}+\tr{F_1 \cdot G_2}
        =       \tr{F_1\cdot G_1+ F_1 \cdot G_2}
.\]
Here we applied Lemma \ref{lem:simulation} to circuits of size at most $s'$, and the proof of $\tr{F_1\cdot (G_1+G_2)}= \tr{F_1\cdot G_1+ F\cdot G_2}$ has size at most $c_0m_0$ and  depth at most $c_1k_0$, for some constants $c_0,c_1$.

\case 2 An application of rules R1, R2 translates to an application of R1, R2. For the rules R3 and R4,  it is sufficient to show the following: if $\pi$ uses the rule \[\frac{F_1=F_2\qquad G_1=G_2}{F_1\circ G_1=F_2\circ G_2},\, \circ\in\{\cdot, +\}, \] then there is a proof of $\tr{F_1\circ G_1}=\tr{F_2\circ G_2}$, of size $c_1m_0$ and depth $c_2k_0$, from the equations $\tr{F_1}=\tr{G_1}$ and $\tr{F_2}=\tr{G_2}$. This is again
 an application of Lemma \ref{lem:simulation}.
\medskip 

Altogether, we obtain a proof of $\tr{F}=\tr{G}$ of size at most $c_1s'm_0$ and depth $c_2k_0$.
\end{proof}
\bigskip

As a corollary of Theorem \ref{thm:main-balancing} and Corollary \ref{lem:scalars-substitution-keeps-poly-deg} we finally obtain the balanced PI-proof of the determinant identities in \VNCTwo. Denote by \cdetbl\ the circuit obtained by applying the balancing algorithm on $\taydetsharpprime(X)$. That is,
\begin{equation}\label{eq:def-of-Det-balanced}
\cdetbl(X):=[\taydetsharpprime(X)].
\end{equation}
\begin{corollary}
[in \VNCTwo]\label{cor:exist-PC-proof-of-det-tay-sharp-identities-fin}
Given a positive natural number $n$ there exists a  $O(\log^2 n)$-depth $\PCZ$-proof of the determinant identities \eqref{eq:1} and \eqref{eq:2}, where the determinant in \eqref{eq:1} and \eqref{eq:2} is written as the division free circuit $\cdetbl(A)$, for $A$, the \nbyn\ symbolic matrix $X$ or $Y$, or their product $XY$, or a symbolic triangular matrix $Z$.
\end{corollary}

\section{Applying the Reflection Principle and Wrapping Up}\label{sec:applying-the-reflection}
Here we conclude the proofs of the determinant identities in the theory by proving and applying the reflection principle for \PCZ-proofs in \VNCTwo\ (Theorem \ref{thm:ovr:sound-of-arithmetic-proofs}). 

\para{The Determinant Function \Det\ in \VNCTwo.}
As presented in the introduction Section \ref{sec:carry-proof-in-theory}, given an \nbyn\ integer matrix $A$, the determinant function $\Det(A)$ in \VNCTwo \ is defined to first construct an $O(\lgsq)$-depth algebraic circuit  for the determinant polynomial of a symbolic \nbyn\ matrix, and then evaluate the circuit under $A$, using the fact that the evaluation of $O(\lgsq)$-depth algebraic circuits over the integers is definable in the theory as shown below.

Formally, the balanced circuit for the determinant constructed by \Det\ is the circuit $\cdetbl(X)$ from \eqref{eq:def-of-Det-balanced}. This construction was shown above (when constructing the PI-proofs).
 


\subsection{Algebraic \NCTwo-Circuit Value Problem}\label{sec:algebraic_circuit_value_problem}
\iddo{}

We  show that there is an \FNCTwo\  algorithm that receives an \textit{algebraic} circuit over \Z\ with $n$ input variables,  size polynomial in $n$ and  
depth $O(\log ^2 n)$, together with an assignment of integers to the variables written as binary strings, and outputs the value of the circuit under the assignment.

The algorithm proceeds as follows: \textbf{i}) convert the input balanced algebraic circuit into a balanced Boolean circuit computing the same polynomial, where integers are written as binary strings; \textbf{ii}) layer the circuit; \textbf{iii}) convert the layered circuit into a monotone circuit; \textbf{iv}) evaluate the balanced monotone Boolean circuit using the evaluation function for such circuits which is \SigOneB-definable in \VNCTwo.



\para{Step (i): from balanced algebraic circuits to balanced Boolean circuits}.
We show how to transform a polynomial-size  $O(\log ^2 n)$-depth algebraic  circuit into a polynomial-size  $O(\log ^2 n)$-depth Boolean circuit with an \FNCTwo\ algorithm. We use the following two facts:

\begin{fact}[By Vinay \cite{Vin91}]\label{fac:vinay}
Given an algebraic circuit of $\poly(n)$-size and $\poly(n)$-degree,  our algorithm (and the original \cite{VSB+83} algorithm) that  balances the circuit into $O(\log ^2 n)$-depth, in fact balances (with straightforward modifications) the circuit into $O(\log n)$-depth \emph{in which the plus gates have unbounded fan-in} (and product remains a binary operation). 
\end{fact}

\newcommand{\stradd}{\ensuremath{\mathsf{StringAdd}}}
\newcommand{\strmult}{\ensuremath{\mathsf{StringMult}}}
\newcommand{\straddf}{\ensuremath{\mathsf{StringAdd}_{fla}}}

\begin{fact}\label{fac:volmer}
The Boolean (multi-valued) function \stradd\ computing the addition of two integers written in binary is in \FO-uniform \FACZ\ (see \cite[p.~85]{CN10}). The Boolean (multi-valued) function \strmult\ computing the product of two integers written in binary is in \FO-uniform \FTCZ\ (see \cite[IX.3.6]{CN10}).
\end{fact}


From Fact \ref{fac:volmer} we conclude that in \VNCTwo\ we can construct  a constant depth fan-in two circuit for addition of two binary integers. Since a plus gate of unbounded fan-in can be simulated by a polynomial-size and $O(\log n)$-depth circuit of plus gates only, we get that in \VNCTwo\ we can construct a polynomial-size fan-in two Boolean circuit of $O(\log n)$-depth for computing iterated addition of binary integers.

\smallskip 

Using Fact \ref{fac:vinay} above, given an \olgsq-depth  algebraic circuit we have the following \SigOneB-definable function in \VNCTwo\ for constructing the corresponding polynomial-size \olgsq-depth (fan-in two) Boolean circuit:
\begin{enumerate}
\item Every unbounded fan-in plus gate is replaced by a polynomial-size fan-in two and depth $O(\log n)$ circuit computing the corresponding iterated sum of integers;
\item Every fan-in two product gate is replaced by a polynomial-size and depth $O(\log n)$ circuit computing the corresponding product of two integers.
\end{enumerate}

The resulting Boolean circuit is thus an $O(\log^2n)$-depth circuit (with a fan-in two) and polynomial-size in $n$. This Boolean circuit is encoded in the same way as algebraic circuits are encoded; namely, via the  encoding scheme in Section \ref{sec:subsec-for-encoding-circuits} (with the obvious modifications: instead of designating $+,\cd$ we designate $\land,\lor,\neg$).  
 \medskip



\para{Step (ii): layering Boolean circuits}.
For the evaluation of Boolean circuits in the theory we need to have circuits that are layered, namely in which every node belongs to a single layer $i$, and nodes in layer $i$ may only go to nodes in layer $i+1$. We can convert within \FNCTwo\ any \olgsq-depth  Boolean circuit from Step (i) above into a layered Boolean circuit, as follows.

\newcommand{\intbin}{{\rm Int}}


%

\medskip 

\hrule\medskip
\noindent\textsc{\FNCTwo-algorithm for layering  balanced Boolean circuits}
\medskip
\hrule
\medskip 

\begin{description}[font=\rm]
\item[\textbf{Input}:]
 A Boolean  circuit $F$ of depth $c\lgsq $, for some constant $c$ (encoded as in Section \ref{sec:subsec-for-encoding-circuits}).

\item[\textbf{Output}:]
A layered Boolean circuit $F'$ computing the same function as $F$. 
\end{description}

\iddo{}

\noindent\textbf{Algorithm}
\begin{enumerate}
\item Let $A$ be the 0-1 adjacency matrices of $F$  where the dimensions of $A$ equal $s$, the number of nodes in $F$ and the $(u,w)$th in $A$, denoted $A[u,w]$ is 1 iff there is a directed edge from node $u$ to node $w$ in $F$. Using the \SigOneB-definable  in \VNCTwo\ string function ${\rm PowSeq}_\Z(n,s,A)$, that receives an \nbyn\ integer matrix $A$ and outputs a string coding the sequence $(A,A^2,\dots,A^s)$ of powers of $A$, we find the shortest length of a directed path from a leaf in $F$ to each of the internal nodes in $F$: the shortest directed length of a path from a leaf $u$ to a node $v$ is the minimal $i$ such that $A^i[u,v]\neq 0$.

\item Let $F'$ be the circuit $F$ in which for every node $u\in V$, for $V$ the set of nodes of $F$, change $u$ to $(u,\ell)$, where $\ell$ is the shortest directed length of a path from a leaf in $F$ to $u$. Hence, $\ell$ will serve as \emph{the layer} of $(u,\ell)$ in $F'$.
\item We now add dummy edges and nodes ``$1\cd u$'' to $F'$, to force every node $u$ to have edges directed only to subsequent layers. Specifically, we scan the nodes of $F'$ from layer 0 to the top layer $c\lgsq$, and for each node $(u,\ell)$ that is connected with a directed edge  $e$ to node $(v,j)$, for $j>\ell+1$, we discard $e$ and add two new nodes and three new edges as follows. Assuming that $(v,j)=(u,\ell)\circ w$, for $\circ\in\{+,\cd\}$, let $(v,j)=((u,\ell)\cd 1)\circ w$, where the new node 1 is on layer $\ell$, the new node $\cd$ is on layer $\ell+1$ and two  new edges are added from $(u,\ell)$ to $\cd$ and from $1$ to $\cd$, and a third edge is added  from $\cd$ to $(v,j)$. After this the node $(u,\ell)$ has a directed edge only to nodes in layer $\ell+1$. Doing this sequentially for all $c\lgsq $ layers we end up with a layered circuit $F'$.   
\end{enumerate}



\smallskip \hrule

\para{Step (iii): convert  layered circuits into a monotone circuits}.
Here we need to apply sequentially De Morgan rules, from top layer to bottom layer, until all negation in the circuit are in the input level. There is no need to add new layers, since the De Morgan rules preserve the number of layers: $\neg (A\land B) \to \neg A \lor \neg B$, $\neg\neg A \to \textsf{true} \land (\textsf{true}\land A)$, $\neg (A\lor B) \to \neg A \land \neg B$.

For balanced circuits this is done in \FNCTwo\  precisely the same way as Part 3 in the algorithm in Step (ii) above, only that we start from the top layer to layer $0$.

\para{Step (iv): evaluation of balanced monotone Boolean circuits}. We define the function $\eac(F,A)$ that receives the string variable $F$ encoding an algebraic circuit over the integers and an assignment of integers to the variables of $F$ written as a two-dimensional array $A$, and outputs the binary string representing the value of the algebraic circuit encoded by $F$ under $A$. 

Let us denote by \mbeac\ the \SigOneB-definable in \VNCTwo\ function that evaluates a layered and monotone Boolean circuit of  depth \olgsq\ as shown in \eqref{eq:LMCV} (Section \ref{sec:vnctwo}).  By Steps (i) to (iii) and using \mbeac\ we  conclude that \eac\ is \SigOneB-definable string function in \VNCTwo. Note that the input Boolean variables are both the  binary strings representing the integers input $A$ and the negation of these binary strings (we need their negation because this is the input to the \emph{monotone} circuit.)

\subsection{Proving the Reflection Principle for \PCZ}\label{sec:proving-the-reflection-principle-for-PCZ}

We shall prove the following reflection principle for \PCZ:

\begin{theorem}(Theorem \ref{thm:ovr:sound-of-arithmetic-proofs} restated; In \VNCTwo)\label{thm:sound-of-arithmetic-proofs}
Let $\pi$ be an $O(\log^2 n)$-depth $\PCZ$-proof of the circuit equation $F=G$. Then $F=G$ is  true in \Z, in the sense that $\forall A\in\Z^n\left(\eac(F,A)=\eac(G,A)\right)$. \end{theorem}

\begin{proof}
The proof proceeds by the \emph{number }induction  (see Proposition \ref{prop:number-induction}) on the number of proof lines in $\pi$, using Lemma \ref{lem:sound-of-PCZ-rules} below. 

Since the evaluation function \eac\ is \SigOneB-definable in \VNCTwo\ we can use this function in the number induction axiom (see Section \ref{sec:introducing-new-definable-functions-in-VZ}).  
Speciffically, consider the \SigZB-formula $Q(n):=\forall i\le n~(\eac(left(\pi^{[i]}),A)=\eac(right(\pi^{[i]}),A))$, where $left(\pi^{[i]})$ and  $right(\pi^{[i]})$ are the left (resp.~right) hand side circuit in the $i$th proof-line in $\pi$. Then the induction states that assuming the first line is true under an assignment $A$, namely, $Q(0)$, and if $Q(n)\to Q(n+1)$ is true, namely if all proof-lines $\le n$ are true under an assignment $A$, then also the $(n+1)$th line is true under $A$---then we finish the argument since we end up with $\forall n\le length(\pi)(Q(n))$. \iddo{}


It thus remains to prove each of the following cases: 1) Axioms of \PCZ. We  show that the evaluation of \PCZ\ axioms under integer assignments is universally true: $\forall \alpha \in \Z^n(F(\alpha)=G(\alpha))$, when $F=G$ is an axiom. For example, $F+0=F$ holds for every integer assignment to $F$; 2) The rules of \PCZ\ are sound under integer assignments. These two cases are proved in Lemma \ref{lem:sound-of-PCZ-rules}.
\end{proof}


%
%


\begin{lemma}(in \VNCTwo)\label{lem:sound-of-PCZ-rules} 
(i) Let $F_1,F_2,G_1,G_2$ be \olgsq-depth circuits, and $A$ an assignment of integers to their input variables. If $\eac(F_1,A)=\eac(G_1,A)$ and $\eac(F_2,A)=\eac(G_2,A)$ then $\eac(F_1\circ F_2,A) = \eac(G_1 \circ G_2,A)$, for $\circ\in\{+,\times\}$.
(ii) Let $F,G$ be \olgsq-depth circuits, $F=G $  an axiom of $\PCZ$ and $A$ an assignment of integers to the input variables of $F,G$. Then, $\eac(F,A)=\eac(G,A)$.
\end{lemma}

\begin{proof}
\iddo{}
Part (i). Let us consider the rule $F_1=G_1$ and $F_2=G_2$ derives $F_1+F_2=G_1+G_2$. We need to prove that  $\eac(F_1+F_2,A) = \eac(G_1+G_2,A)$.
%

Denote by $\straddf(X,Y)$ the \SigZB-\emph{formula} from Fact \ref{fac:volmer} for adding two binary integers (we reserve the symbol \stradd\ to denote the corresponding constant depth Boolean \emph{circuit}). Specifically, we have (see \cite{CN10}):
\begin{definition}[\straddf]\label{def:straddf} The \SigZB-formula for computing carries in a carry-save adder is:
\begin{multline*}
Carry(i,X,Y)\leftrightarrow \exists k<i\big(X(k)\land Y(k)\land
    \forall j<i(k<j\to (X(j)\lor Y(j)))\big).
\end{multline*}
And the \SigZB-defining axiom for \straddf\ is (where $\oplus$ is exclusive or): 
\begin{multline*}
R_+(X,Y,Z)\leftrightarrow \big(|Z|\le |X|+|Y|\land  \forall i<|X|+|Y|(Z(i)\leftrightarrow X(i)\oplus Y(i)\oplus Carry(i,X,Y)))\big).
\end{multline*}
\end{definition}

The definition of the Boolean circuit \stradd\ is similar to  the \FO\ formula in Definition \ref{def:straddf}: $Carry(i,X,Y)$ is defined as above except that  $\exists k<i$ turns into $\lor _{i=0}^{k-1}$ and  $\forall j<i$ turns into $\land _{j=0}^{i-1}$ and $X(i), Y(i)$ are interpreted as the $i$th bits of the input $X,Y$, that is the Boolean variable $x_i,y_i$, respectively. Thus, $\stradd:=X(i)\oplus Y(i)\oplus Carry(i,X,Y)$ (where $\oplus$ here is built from $\land,\lor,\neg$).

\medskip 
Since $\eac(F_1,A)=\eac(G_1,A)$ and $\eac(F_2,A)=\eac(G_2,A)$, \begin{multline}\label{eq:al-lmb}
\straddf(\eac(F_1,A),\eac(F_2,A))= \straddf(\eac(G_1,A),\eac(G_2,A)).
\end{multline}  

Recall that  by construction, $\eac(F_1+F_2,A)$ first converts the algebraic circuit $F_1+F_2$ into a Boolean circuit of the form $\stradd(F'_1,F'_2)$, where $F'_1,F'_2$ are the monotone, layered and Boolean versions of $F_1,F_2$, respectively, as described in the algorithm above (Steps (i) to (iii)), and then  evaluates it using $\mbeac(\stradd(F'_1,F'_2),A)$ (and similarly for $G_1,G_2$) (for simplicity we shall ignore here the fact that also \stradd\ is turned into a monotone circuit). 
\iddo{} Therefore, we can prove:
\begin{gather}\label{eq:al-lmb2}
\eac(F_1+ F_2,A)=\mbeac(\stradd(F'_1,F'_2),A),\\\label{eq:al-lmb2b}
\eac(G_1+ G_2,A)=\mbeac(\stradd(G'_1,G'_2),A).
\end{gather}
By \eqref{eq:al-lmb}, \eqref{eq:al-lmb2} and \eqref{eq:al-lmb2b} it suffices to prove 
\begin{gather}\label{eq:al-lmb3}
\mbeac(\stradd(F'_1,F'_2),A)=\straddf(\eac(F_1,A),\eac(F_2,A))
,\\ \label{eq:al-lmb3b}
\mbeac(\stradd(G'_1,G'_2),A)=\straddf(\eac(G_1,A),\eac(G_2,A)).
\end{gather}

Let us prove \eqref{eq:al-lmb3} (as \eqref{eq:al-lmb3b} is similar). By the discussion above $\straddf(\eac(F_1,A),\eac(F_2,A))=\straddf(\mbeac(F'_1,A),\mbeac(F'_2,A))$. Hence, it remains to prove 
\begin{equation}\label{eq:two-sides1}
\mbeac(\stradd(F'_1,F'_2),A) = \straddf(\mbeac(F'_1,A),\mbeac(F'_2,A)). \end{equation}

First note that the evaluation function \mbeac\ (\eqref{eq:LMCV} in Section \ref{sec:vnctwo}) works the same for \textit{multi-output }circuits. Second, recall that the function \mbeac\ is defined so that given a circuit $F$ it produces an evaluation string (as defined in \eqref{eq:LMCV}) for the \emph{whole} circuit $F$, and then outputs the evaluation string only of the \emph{top layer} (namely, the output nodes). 

The idea of the proof of \eqref{eq:two-sides1} is the following: consider the left hand side of \eqref{eq:two-sides1}. The evaluation string of the circuit $\stradd(F'_1,F'_2)$ produced by \mbeac\ given $A$ is the same as the combination of the separate evaluation strings of $F'_1$ and $F'_2$, excluding the top layers which belongs to the evaluation of (the constant many layers of the circuit) \stradd. Therefore, we can prove that
\begin{equation}\label{eq:the-idea}
\mbeac(\stradd(F'_1,F'_2),A) = \mbeac(\stradd(\mbeac(F'_1,A),\mbeac(F'_2,A)),A)
\end{equation}
(note that in \eqref{eq:the-idea} the rightmost input $A$ on the right hand side does not have any effect, since the circuit $\stradd(\mbeac(F'_1,A),\mbeac(F'_2,A))$ has no variables).

Therefore, to conclude  \eqref{eq:two-sides1}  it remains to show
\begin{multline}\label{eq:two-sides2}
\mbeac(\stradd(\mbeac(F'_1,A),\mbeac(F'_2,A)),A)\\
 = \straddf(\mbeac(F'_1,A),\mbeac(F'_2,A)). 
 \end{multline}



This is done by the number induction on the \SigZB-formula $$
\Psi(i,X,Y):=|X|=|Y|=i \land \mbeac(\stradd(X,Y),A) = \straddf(X,Y)
$$
(note again that $A$ on the left hand side does not have any effect when $X,Y$ are interpreted as constant binary strings (devoid of variables)). 
This is done using the evaluation string produced for the Boolean circuit \stradd\ by \mbeac. The idea is that the \FO\ formula and the Boolean circuit for \stradd\ are almost identical. More generally, we have the following claim that is proved by construction; and formally, by number induction on the depth of the circuit (equivalently, the number of layers in $C'$):

\begin{claim}\label{claim:circuit-and-FO-formula-equiv}
Let $C$ be an \FNCTwo\ function written in the language \LTwoA\ augmented with function symbols for \FNCTwo\ string and number functions. Further, let $C'$ be the corresponding monotone layered Boolean circuit for $C$. Then, \VNCTwo\ can prove that $C(A)=\mbeac(C',A)$.
\end{claim}

\iddo{}

The same reasoning is applied to  \strmult, for dealing with the rule $F_1=G_1$ and $F_2=G_2$ derives $F_1\cd F_2=G_1\cd G_2$.
\smallskip 

Part (ii) is similar to part (i) and we omit the details.
\end{proof}

\subsection{Wrapping Up}\label{sec:wrapping-up}
Using the definition of \Det, Theorem \ref{thm:ovr:sound-of-arithmetic-proofs} and Corollary \ref{cor:exist-PC-proof-of-det-tay-sharp-identities-fin} we are finally in a position to conclude the main theorem.
\iddo{}

\begin{theorem}[Main theorem]\label{thm:main}
The following determinant identities are provable in \VNCTwo: 
\begin{gather}\label{eq:1-2-final-thm}
\forall n\forall A(\matz(A,n)\land \matz(B,n) \to \Det(A)\cd\Det(B)=\Det(AB))\,,\\
\forall n\forall A(\trianmatz(A,n) \to \Det(A)=A[1,1]\cdots A[n,n])\,.
\end{gather}
 \end{theorem}
Where in \eqref{eq:1-2-final-thm} $\matz(A,n)$ means that $A$ is an \nbyn\ integer matrix, with integer entries are encoded by strings as usual, and $\trianmatz(A,n)$ means that $A$ is a lower or upper  \nbyn\ triangular matrix, and $A[i,j]$ is the $(i,j)$th integer entry in $A$ .

\bigskip 

Using the translation between bounded arithmetic theories and propositional proofs as shown in \cite{CN10} we can also extend the result in \cite{HT12} to work over the integers:

\begin{theorem}\label{thm:NCTwo-Frege-proofs}
There are polynomial-size propositional \NCTwo-Frege proofs of the determinant identities over the integers.   
\end{theorem}

In Theorem \ref{thm:NCTwo-Frege-proofs}, \NCTwo-Frege is defined as in \cite{HT12}, namely, these are families of standard propositional (Frege) proofs with size $\poly(n)$ in which every proof-line is a circuit of depth $\olgsq$, and where we augment the system with rules for manipulating circuits  similar to the rules C1, C2 in \PC\ (it is possible to characterize these proofs as  restricted Extended Frege proofs). Moreover, integers in the \NCTwo-Frege proofs are encoded by fixed length binary strings, that is sequences of propositional variables. Note that for every fixed length of binary strings encoding integers, we will have a different propositional proof.

\section{Corollaries}\label{sec:corollaries}
Here we show some further theorems of linear algebra that can be proved in \VNCTwo, using similar arguments as before. Specifically, we show that the Cayley-Hamilton theorem and the co-factor expansion of the determinant are  provable in \VNCTwo, as well as the \textit{hard matrix identities }identified by Soltys and Cook in \cite{SC04}.

The Cayley-Hamilton (C-H) theorem states that for the (univariate) \emph{characteristic} \emph{polynomial} of a matrix $A$ in the variable $z$, defined as \iddo{}
$$
p_A(z):=\det(zI-A),
$$
it holds that $p_A(A)=0$, where $p_A(A)$ is a univariate polynomial in the matrix $A$, product is interpreted as matrix product, and scalar multiplication of a matrix is interpreted as usual, and where the right hand side 0 stands for the all zero matrix. 

\newcommand{\bl}[1]{\ensuremath{\left[#1\right]}}
The characteristic polynomial of a matrix is defined in the theory as follows: we introduce a \SigOneB-definable string function $p(A,n)$ that receives an  \nbyn\ integer matrix $A$ and outputs a division free $O(\log^2 n)$-depth algebraic circuit with $n^2$ input variables, where the coefficient of $z^i$, for $i=0,\dots,n$, in the circuit is computed (as a sub-circuit) by  
$$
\bl{\Coef_{z^i}\left(\taydetsharp(zI_n-A)\right)},
$$ 
namely, the  balanced circuit that extracts the (constant) coefficient of the determinant polynomial of $zI_n-A$; recall that \taydetsharp\ is a division free circuit with a polynomial syntactic-degree. Thus, overall the string function $p(A,n)$ outputs the following circuit, written as an $O(\lgsq)$-depth circuit,
for the characteristic polynomial of $A$:
\begin{equation}\label{eq:characteristic-poly}
p(A,n):=\sum_{i=0}^n \bl{\Coef_{z^i}\left(\taydetsharp(zI_n-A)\right)}\cd z^i.
\end{equation}
\medskip
 

\iddo{}

The C-A theorem is expressed in the theory as follows:
\iddo{} 
\begin{equation}\label{eq:C-H-thm}
\forall n\forall A(\matz(A,n)\to \eval{p(A,n),A}=\mathbf{0}_n)\,,
\end{equation}
where $\mathbf{0}_n$ is the all zero \nbyn\ integer matrix, $\eval{X,A}$ is the string function that evaluates the circuit $C$ under the integer assignment $A$, $\matz(A,n)$ is the relation that holds iff $A$ is an \nbyn\ integer matrix, and $p(\cd)$ is \SigOneB-definable function that receives a matrix and outputs a circuit (in fact a formula) that computes its characteristic polynomial (with a single input variable $z$).   
  
\begin{corollary} \label{cor:C-H-theorem}
The Cayley-Hamilton theorem, expressed as in \eqref{eq:C-H-thm}, is provable in \VNCTwo. 
\end{corollary}
\begin{proof}
This follows  the same line of arguments demonstrated for the \VNCTwo-proofs of the  determinant identities. We first construct using \SigZB-\textbf{COMP} the \PCZ-proof of the C-H theorem shown in  Proposition 9.4 in \cite{HT12} and then use the reflection principle as in Section \ref{sec:applying-the-reflection}.
The only difference is that we need to use part (3) in Lemma \ref{lem:Taylor-basic-properties} (we did not use this part before), and for this we need to supply the witnesses for the syntactic-degrees of the nodes in (\ref{eq:characteristic-poly})\iddo{}. This needs more work, and is shown in the appendix in Lemma \ref{lem:supply-witness-for-syntactic-degrees}. \iddo{}
\end{proof}

Other basic results in linear algebra that are provable in \VNCTwo\ are the cofactor expansion of the determinant and the inversion principle, as follows.

The \emph{inversion principle} is the following formula in \VNCTwo: 
$$
\forall n\forall A,B (\matz(A,n)\land \matz(B,n)\to (AB=I\to BA=I)).
$$
Soltys and Cook \cite{SC04} showed that the inversion principle is equivalent in the theory $LA$ (that can be interpreted in \VNCTwo\ by Cook and Fontes \cite{CF10}\footnote{Though here we have to be careful, because the encoding of matrices and polynomials and the determinant we introduce is different from the encoding of \cite{SC04,CF10}.}) to the following principles they called collectively (including the inversion principle itself) \emph{the hard matrix identities}:
\begin{align*}
& AB=I\land AC=I \to B=C \\
& AB=I \to AC\neq 0 \lor C=0\\
& AB=I \to A^t B^t=I\,.
\end{align*}
\begin{corollary}\label{cor:inv-principle}
The inversion principle is provable in \VNCTwo. \end{corollary}

\newcommand{\adj}{\ensuremath{\mathrm {Adj}}}

 For an \nbyn\ matrix $X$ let $X[i|j]$ be the $(n-1)\times(n-1)$ minor obtained by removing the $i$th row and $j$th column from $X$ (recall that a sum of integer numbers represented in binary is definable in \VNCTwo\ (cf.~\cite{CN10})). 
\begin{corollary}\label{cor:cofactor}
The following \emph{cofactor expansion} of the determinant is provable in \VNCTwo:
$$
\forall n\forall A\left(\matz(A,n)\to\left(\Det(A)=\sum_{j=1}^n(-1)^{i+j}A(i,j)\Det(A[i|j])\right)\right)\,.
$$
\end{corollary}

The proofs of Corollaries \ref{cor:inv-principle}  and \ref{cor:cofactor} are similar to the proof of Corollary \ref{cor:C-H-theorem}.  
It uses the adjoin of a matrix $\adj(X)$ which is defined to be the \nbyn\ matrix whose $(i,j)$th entry is $(-1)^{i+j}\Det(X[i|j])$, where \Det\ is the determinant function  (\SigOneB-defined in \VNCTwo). Then we proceed as in Corollary \ref{cor:C-H-theorem} following Proposition 9.1 and 9.2 from \cite{HT12}.


\section{Conclusions and Open Problems}\label{sec:Conclusions and Open Problems}
We established a  proof of the basic determinant identities and other basic statements of linear algebra such as the Cayley-Hamilton theorem in the weakest logical theory known to date. This answers an open question of, e.g., Cook and Nguyen \cite{CN10}. We achieved this by formalizing in the theory \VNCTwo\ the construction of the PI-proof demonstrated in Hrube\v s-Tzameret \cite{HT12}, and using a reflection principle for PI-proofs in the theory. Due to the central role of linear algebra and the determinant function, these results are expected to be relevant to further basic work in bounded arithmetic.

As mentioned in Section \ref{sec:Note-on-the-choice-of-Theory} the complexity classes $\SharpSACOne\subseteq\TCOne$ that are above \DET\ but below \NCTwo, can compute the required depth reduction and the evaluation of algebraic circuits, and we believe that our construction can be carried out more or less  the same in theories corresponding to these classes (though theories for these classes have not been investigated yet). 

It will be very interesting to establish the same identities in a theory that corresponds to the complexity class \DET\ whose complete (under \ACZ-reductions) problems are the integer determinant itself and matrix powering; such a theory denoted $V\#L$ was introduced in \cite{CF10}. This would necessitate  a completely new argument different from ours  (possibly following Berkowitz' \cite{Ber84} algorithm for the determinant) and may also contribute to the simplification of the proofs. The reason is that our argument utilizes crucially the evaluation of Boolean \NCTwo-circuits in the theory, while it is not expected that such evaluation is doable in  the class \DET.

\section*{Acknowledgements} We thank Pavel Hrube\v s for useful discussions while working on   \cite{HT12},  Eric Allender for very helpful correspondence regarding  \cite{AJMV98} and Emil Je\v rabek for clearing up things about \ACOne\ and \TCOne. An extended abstract of this work  appeared initially at LICS 2017.

\appendix
\

\section*{Appendix}

\bigskip 





%


\section{Definability in Bounded Arithmetic}\label{sec:Definability-in-Bounded-Arithmetic}
Here we give more details on the theories \VZ\ and \VNCTwo.
Specifically, we wish explain now how to extend the language \VZ\ and \VNCTwo\ \ with new function symbols (in a conservative way; see below).

We write $ \exists ! y \varphi $ to denote $ \exists x (\varphi(x) \land \forall y (\varphi(y/x)\to x=y)) $, where $ y $ is a  variable not appearing in $ \varphi $:
\begin{definition}[Two-sorted definability]\label{def:Two-sorted-definability}
Let $\mathcal T $ be a theory over the language $ \mathcal L \supseteq \LTwoA$ and let $ \Phi $ be a set
of formulas in the language $ \mathcal L $. A \emph{number} function $f$ is $\Phi $\textbf{-definable in a theory}
$\mathcal T$ iff there is a formula $ \varphi(\vec x,  y,\vec X)$ in $ \Phi $ such that $\mathcal T$ proves
  \begin{equation*}
        \forall\vec x\forall\vec X\exists!y\varphi(\vec x, y,\vec X)
  \end{equation*}
  and it holds that\footnote{Meaning, it holds semantically
in the standard two-sorted model $ \N_2 $.}
  \begin{equation}
        \label{eq:def axiom num}
        y=f(\vec x,\vec X)\leftrightarrow \varphi(\vec x,y,\vec X).
  \end{equation}
A \emph{string} function $F$ is $\Phi $\textbf{-definable in a theory }$\mathcal T$ iff there is a formula $\varphi(\vec
  x,\vec X,Y)$ in $ \Phi $ such that $\mathcal T$ proves
  \begin{equation*}
        \forall\vec x\forall\vec X\exists!Y\varphi(\vec x,\vec X,Y)
  \end{equation*}
  and it holds that
  \begin{equation}
        \label{eq:def axiom str}
        Y=F(\vec x,\vec X)\leftrightarrow \varphi(\vec x,\vec X,Y).
  \end{equation}
Finally, a relation $R(\vec x,\vec X)$ is $\Phi $-definable in a theory $\mathcal T$ iff there is a formula
$\varphi(\vec x,\vec X,Y)$ in $ \Phi $ such that it holds that
  \begin{equation}
        \label{eq:def axiom rel}
        R(\vec x,\vec X)\leftrightarrow \varphi(\vec x,\vec X).
  \end{equation}
The formulas \eqref{eq:def axiom num}, \eqref{eq:def axiom str}, and \eqref{eq:def axiom rel} are the
{\em defining axioms} for $f$, $F$, and $R$,  respectively.
\end{definition}

\begin{definition}[Conservative extension of a theory]\label{def:conservative extension}
Let $ \mathcal T $ be a theory in the language $ \mathcal L $. We say that a theory $\mathcal
T'\supseteq \mathcal T $ in the language $ \mathcal L'\supseteq \mathcal L $ is \emph{conservative over
$ \mathcal T $} if every $ \mathcal L $ formula provable in $\mathcal T' $ is also provable in $\mathcal
T$.
\end{definition}

We can expand the language $\mathcal L $  and a theory $\mathcal T$ over the language $ \mathcal L $ by adding symbols for arbitrary functions $f$ (or relations $R$) to $\mathcal L $ \textit{and their defining axioms} $A_f$ (or $A_R$) to the theory $\mathcal T$. If the appropriate functions are definable in $\mathcal T$ (according to Definition \ref{def:Two-sorted-definability}) then the theory $\mathcal T +A_f$ ($+A_R$) is conservative over $\mathcal T$. This enables us to add new function and relation symbols to the language while proving statement inside a theory; as long as these function and relation symbols are definable in the theory, every statement in the original language proved in the extended theory (with the additional defining-axioms for the functions and relations) is provable in the original theory over the original language.

However, extending the language and the theory in such a way \emph{\textbf{does not guarantee}} that one can use the new function symbols in the \textit{\textbf{comprehension}} (and induction) axiom schemes. In other words, using the comprehension (and induction) axioms over the expanded language may lead to a theory that is not a  conservative extension. Therefore, definability will not be enough for our purposes. We will show below precisely 
%
%
how to make sure that a function is \emph{both} definable in the theories we work with and also can be used in the corresponding comprehension and induction axiom schemes (while preserving conservativity).

When extending the language with new function symbols we can assume that in \emph{bounded formulas} the bounding terms possibly use function symbols from the expanded language
(because any definable function in a bounded theory can be bounded by a term in the original language \LTwoA\ (cf. \cite{CN10})).




\bigskip 

\subsection{Introducing New Definable Functions in \VZ\ and \VNCTwo}\label{sec:introducing-new-definable-functions-in-VZ}
Here we describe a process (presented in Section V.4. in \cite{CN10}) by which we can extend the language $\LTwoA $ of \VZ\ with new function symbols, obtaining a conservative extension of \VZ\ \textit{that can also prove the comprehension and induction axiom schemes in the extended language, and similarly for \VNCTwo.}

First note that every relation or function symbol has an intended or standard interpretation over the
standard model $ \N_2$ (for instance, the standard interpretation of the binary function ``$ + $'' is that
of the addition of two natural numbers). If not explicitly defined otherwise, we will always assume that a
defining axiom of a symbol in the language defines a symbol in a way that its interpretation in $\N_2$
is the standard one. Note also that we shall use the same symbol $ F(\vec x, \vec X) $ to denote both the function and the \emph{function symbol} in the (extended) language in the theory.

\begin{definition}[Relation representable in a language]
Let $ \Phi $ be a set of formulas in a language $\mathcal L$  that extends \LTwoA. We say a relation $R(\vec x,\vec X)$  (over the standard model) is \textbf{representable} by a formula from $ \Phi $ iff there is a formula $\varphi(\vec
x,\vec X)$ in $ \Phi $ such that in the standard two-sorted model $ \N_2 $ (and when all relation and
function symbols in $ \mathcal L $ get their intended interpretation), it holds that:
  \begin{equation}
        \label{eq:def axiom rel}
        R(\vec x,\vec X)\leftrightarrow \varphi(\vec x,\vec X).
  \end{equation}
\end{definition}
\medskip

We say that a number function $ f(\vec x,\vec X) $ is \emph{polynomially-bounded} if $f(\vec x,\vec X) \le {\rm poly}(\vec x,\vec {|X|}) $. We say that a string function $ F(\vec x,\vec X) $ is \emph{polynomially-bounded} if $|F(\vec x,\vec X) |
\le {\rm poly}(\vec x,\vec{ |X|}) $.

\begin{definition}[Bit-graph]\label{def:bit-graph-of-string-func}
Let $ F(\vec x,\vec X) $ be a polynomially-bounded string function. We define the \textbf{bit-graph} of $ F
$ to be the relation $ R(i,\vec x,\vec X) $, where $ i $ is a number variable, such that
\begin{equation}\label{eq:bit-defn}
        F(\vec x,\vec X)(i) \leftrightarrow  i<t(\vec x,\vec X) \land R(i,\vec x,\vec X)
\end{equation}
holds in the standard two-sorted model,  for some number term $ t(\vec x,\vec X) $.
\end{definition}

\begin{definition}[$\SigZB$-definability from a language; Definition V.4.12.~in \cite{CN10}]
\label{def:sigB1-definability-in-a-langauge} We say that a \emph{number function $ f(\vec x,\vec X)$ is $\SigZB$-definable from a language $ \mathcal L \supseteq \LTwoA $}, if $ f $ is polynomially-bounded and its graph\footnote{I.e., the relation $R(\vec x,\vec X,y)$, such that $f(\vec x,\vec X)=y$ iff $R(\vec x,\vec X,y)$ holds in the standard model.} is represented by a $\SigZB(\mathcal L) $-formula $\varphi $. We call the formula $\varphi $ the \emph{defining  axiom of $ f $}. We say that a \emph{string function $ F $ is $\SigZB$-definable from a language $ \mathcal L \supseteq \LTwoA $}, if $ F$ is polynomially-bounded and its \emph{\uline{bit-graph}} (as in \eqref{eq:bit-defn})
is representable by a $\SigZB(\mathcal L) $-formula $\varphi $. We call the formula $\varphi $ the
\emph{defining  axiom of $ F$} or, equivalently, the \emph{bit-defining axiom of $ F $}.
\end{definition}
\begin{note}
We used the term \emph{defining axiom of a function $ f $} in both the case where $ f $ is defined \emph{from a language} (Definition \ref{def:sigB1-definability-in-a-langauge}) and in case $ f $ is definable \emph{in the theory} (Definition \ref{def:Two-sorted-definability}). In general it is important not to confuse these two  notions. Nevertheless, we will show in the sequel that \emph{for our purposes} these two notions coincide: when we define a function from a language the function will be definable also in the relevant theory, and so the defining axiom of $ f $ from the language will be the defining axiom of $ f $ in the theory (when the theory is possibly  conservatively extended to include new function symbols).

%
\end{note}

The following is a definition of \ACZ\ functions. This definition coincides with the definition of \FACZ\ as \FO-uniform multi-output Boolean circuit families of polynomial-size and constant depth \cite{CN10}.  
\begin{definition}[\FACZ]\label{def:FACZ}
A string (number) function is in \FACZ\ if it is polynomially-bounded and its bit-graph (graph, respectively) is definable by a $\SigZB$-formula in the language \LTwoA.
\end{definition}

\begin{definition}[\ACZ-reduction]
\label{def:ACZ-reduction} A number function $ f $ is \emph{\ACZ-reducible to $\mathcal L\supseteq
\LTwoA$} iff there is a possibly empty sequence of functions $ F_1,\dots,F_k $ such that $ F_i $ is
\TODO{Check if it's $ F_i $ can be only \emph{string} functions, as written in the book}
$\SigZB$-definable from $ \mathcal L \,\cup\,\{F_1,\dots,F_{i-1}\}$, for any $ i=1,\dots,k $, and $ f $ is
$\SigZB$-definable from $ \mathcal L\,\cup\{F_1,\dots,F_{k}\}$.
\end{definition}
\bigskip

We are now finally ready to describe the standard process enabling one to extend a theory $\mathcal T \supseteq \VZ$ over the language \LTwoA\ (and specifically, the theories \VZ\ and \VNCTwo) with new function symbols, obtaining a conservative extension of $\mathcal T $ such that the new function symbols can  be used in comprehension and induction axiom schemes in the
theory (see Section V.4. in \cite{CN10} for the proofs):\iddo{}

\begin{enumerate}
\item[(i)] If the number function $ f $ is $ \SigZB $-definable from $ \LTwoA$, then $\mathcal T $ over the language $ \LTwoA\cup\{ f\} $, augmented with the defining axiom of $ f $,  is a conservative extension of $ \mathcal T $ and we can also prove the comprehension and induction axioms for $\SigZB(f)$-formulas.

\item[(ii)] If the string function $ F $ is $ \SigZB $-definable from $ \LTwoA$, then $\mathcal T $ over the language $ \LTwoA\cup\{ F\} $, augmented with the \emph{bit-defining} axiom of $ F $, is a conservative extension of $ \mathcal T $ and we can also prove the comprehension and induction axioms for $\SigZB(F)$-formulas.

\item[(iii)] We can now iterate the above process of extending the language $ \LTwoA(f) $ (or equivalently, $ \LTwoA(F) $) to conservatively add more functions $ f_2,f_3,\dots$ to the language, which can also be used in comprehension and induction axioms.
\end{enumerate}

\iddo{}
By the aforementioned and by Definition \ref{def:ACZ-reduction}, we can extend the language of a theory with a new function symbol $ f $, \emph{whenever $ f $ is \ACZ-reducible to \LTwoA}. This results in an extended theory (in an extended language) which is conservative, and can prove the comprehension and induction axioms for formulas in the extended language. When defining a new function in \VZ\ or \VNCTwo\ we may simply say that it is \emph{$\SigZB$-definable} or \emph{bit-definable} in the theory and give its $\SigZB$-defining or bit-defining axiom (this axiom can use also previously $\SigZB$-defined (or bit defined) function symbols).

Extending the language of \VZ\ and \VNCTwo\ with new \emph{relation} symbols is simple: every relation $ R(\vec x, \vec X) $ which is representable by a $\Delta^1_1(\mathcal L)$ formula (\cite[Section V.4.1]{CN10}), where $ \mathcal L $ is an extension of the language with new function symbols obtained as shown above, can be added itself to the language. This results in a conservative extension of $ \VZ $ (\VNCTwo, resp.)  that also proves the $\SigZB$-induction and comprehension axioms in the extended language.

\subsection{Some Basic Formalizations in \VZ}\label{sec:basic_formalization_in_ACZ}
\TODO{Move to appendix}In this section we show how to formalize basic objects in \VZ. Most formalizations here are routine (cf.~\cite{CN10,MT10}).


\para{Natural number sequences of constant length} For two numbers
$x,y$ let $\langle x,y\rangle := (x+y)(x+y+1)+2y$ be the
\emph{pairing function}, and let $ \textit{left}(z),\textit{right}(z) $ be the ($\SigZB$-definable in \VZ) projection functions of the first and second element in the pair $ z $, respectively. We also $\SigZB$-define inductively  $\langle v_1,\dots,v_k\rangle := \langle \langle v_1,\dots,v_{k-1}\rangle, v_k\rangle$, for any constant $ k> 2 $. Then $\VZ$ proves the injectivity of the pairing function and enables us handling such pairs in a standard way.

\begin{notation}
Given a number $ x $, coding a sequence of natural numbers of length
$ k $, we write $\langle x \rangle^k_i $, for $ i =1,\dots,k $, to denote the number in the
$ i $th position in $ x $. \TODO{Used?} This is a $\SigZB$-definable function in \VZ\ (defined via $ \textit{left}(x), \textit{right}(x) $ functions).
\end{notation}


%


\para{Natural and integer number sequences}
If we wish to talk about sequences of numbers (whether natural, integers or rationals) where the \textit{length of a  sequence is} \textit{non-constant}, we have to use string variables instead of number variables. Using the number-tupling function
we can encode sequences as sets of numbers (recall that a string is identified with the finite set of
numbers encoding it): a sequence is encoded as a string $ Z $ such that, the $ x $th number in the sequence is $ y $ if the number $ \langle x,y \rangle $ is in $ Z $. Formally, we have the following $ \SigZB$-defining formula for the number function $ {\rm seq}(x,Z) $ returning the $x$th element in the sequence $Z$:
\begin{equation}\label{eq:seq function definition}
  \begin{split}
        y={\rm seq}(x,Z) \leftrightarrow
                & \,(y<|Z|\wedge Z(\langle x,y\rangle)\wedge\forall z<y\,
            \neg Z(\langle
                x,z\rangle))\\
                & \vee(\forall z<|Z|\neg Z(\langle x,z\rangle)\wedge
                y=|Z|).
  \end{split}
\end{equation}
Formula (\ref{eq:seq function definition}) states that the $ x $th element in the sequence coded by $ Z $ is $ y $ iff $ \langle x,y \rangle $ is in $ Z $ \emph{and} no other number smaller than $ y $ also ``occupies the $ x $th position in the sequence'', and that if no number occupies position $ x $ then the function returns the length of the string variable $ Z $.

%

We define the number function $\textit{length}(Z)$ to be the length of the sequence
$ Z $, as follows:
\begin{multline}
\l=\textit{length}(Z) \leftrightarrow 
\SEQ(\l,Z) \land \exists w <|Z|\exists j<|Z|(Z(w)\land w=\langle\l-1,j \rangle)\,.
\end{multline}
The defining axiom of $ \textit{length}(Z) $ states that $ Z $ encodes a sequence and is the lexicographically smallest string that encodes this sequence and that $ \l-1$ is the largest position in the sequence which is occupied (by definition there will be no pair $ \langle a,b \rangle \in Z $ with $ a>\l-1 $).
\newcommand{\RowArray}{{\rm RowArray}}
%
%

\para{Array  of strings}
We wish to encode a \textit{sequence} of strings as an array. We use the function $ \RowArray(x,Z) $ to denote the $ x$th string in $Z$ as follows (we follow the treatment in \cite[Definition V.4.26, page 114]{CN10}).

\begin{definition}[Array of strings]\label{def:array of strings}
The string function $\RowArray(x,Z)$, abbreviated $ Z^{[x]} $, is $\SigZB$-definable  in \VZ\ using the following bit-definition:
\[
        \RowArray(x,Z)(i) \, \leftrightarrow 
            \, (i<|Z| \land Z(\langle x,i\rangle)).
\]
\end{definition}

\para{Matrices}
An \nbyn\ integer matrix is coded as an array of $n$ strings, where each of the $n$ strings is itself an array that represents a row in the matrix, that is an array of $n$ integer numbers. 

\subsection{Binary Tree Construction in \VZ}\label{sec:construct-binary-tree}
Here we demonstrate a \SigZB-construction in \VZ\ of a binary tree encoding. This can be used to construct a formula that computes for example the inner product of two vectors as in Section \ref{sec:vz-function-for-constructing-DET-circuit}.
Specifically, we show that the string function $F(n)$ that receives a number $n$, which we assume is a  power of 2 for simplicity, and outputs a string that describes the edges of a  binary tree with $n$ leaves is \SigZB-definable in \VZ. 

Consider the tree shown in the picture below. Each node in the tree belongs to a single layer $d=1,\dots,\log(n)+1$, and in each layer $d$ the nodes are labeled from $1$ to $2^{\log(n)+1-d}$. The wires of the tree $T$ are encoded by a three-dimensional array, namely a string $E$ such that $E(d,u,v)$ holds iff the output of gate $u$ on layer $d$ is connected to the input of gate $v$ on layer $d+1$.  
\bigskip 

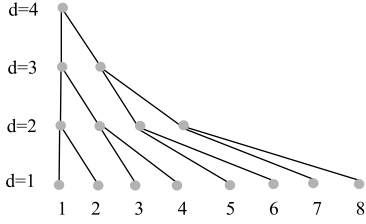
\input{tree.tex}

To show that the string function $F(n)$ \SigZB-definable in \VZ\ (equivalently, \SigOneB-definable in \VZ), according to Section \ref{sec:introducing-new-definable-functions-in-VZ} we need to demonstrate a \SigZB-formula that bit-defines the tree encoding as follows. Let  
\begin{multline}
\Phi(d,u,v)\equiv
n\le du ~\land ~\exists x\le n(u=2x\to 2v=u)\land \\ \exists x\le n (u=2x+1\to 2v=u+1).
\end{multline} 
Then $\Phi(d,u,v)$ is true iff $u$ is a node that occur in the $d$th layer ($n\le du$) and that if $u$ is even then $u$ connects to node $n/2$ in the $(d+1)$th layer, and otherwise it connects to node $\frac{u+1}{2}$ in the $(d+1)$th layer. For the bit-definition of $F(n)$ we introduce the following \SigZB-formula:
$$
\varphi(n,i)\equiv i=(r,u,v)\land r < n ~\land~ u\le n~\land v\le n ~\land \Phi(r,u,v)\,.
$$

\section{Algorithm for \Coef}\label{sec:alg-for-coef}\label{sec:some-technical-lemmas}

The following is similar to the homogenization algorithm from Section \ref{sec:homogenization-of-proofs}.
\bigskip 

\hrule  
\medskip 
\noindent \textsc{{Algorithm for Constructing $\Coef_{z^k}(\cd)$ in Uniform \FACZ}}\smallskip \mar{FAC is undefined}
\medskip
\hrule
\medskip 

\noindent \textbf{Input:} an arithmetic circuit $C$ of size $s$ and a natural number $k$.

\noindent \textbf{Output:} an arithmetic circuit computing $\Coef_{z^k}(C)$.
\medskip

\noindent\textbf{Algorithm:} 
Every node $v$ in $C$ is duplicated $k+1$ times into the nodes $[v,0],\dots,[v,k]$, such that $[v,i]$ is (the root of) a circuit computing the (polynomial) coefficient of $z^i$ in $\widehat {C_v}$. 
\mar{homogeneous part is defined?}
The algorithm is doable in \FACZ\ because every new node $[v,i]$ depends only on the copies of the two nodes $u,w$ that goes into $v$, and these nodes are already known from the input circuit, namely, they are \([u,i], [w,i]\), for  \(i=0,\dots,k+1\), where \(v=u+w\) or \(v=u\cdot w\) in \(C\). Hence, the wiring of the new circuit is done in parallel for each of the new nodes as follows:

\case 0 $v$ is a leaf in $C$. If $v\neq z$ then define $[v,0]=v$, and $[v,i]=0$ for all $i=1,\dots,k$. Otherwise, $v=z$ and we define $[v,1]=1$, and $[v,i]=0$ for all $1\neq i\in\{0,\dots,k\}$. 

\case 1 $v=u+w$ in $C$. Define $[v,i]:=[u,i]+[w,i]$ for every $i=0,\dots,k$.

\case 2 $v=u\times w$ in $C$. Define $[v,i]:=\sum_{j+r=i \atop j,r=0,\dots,k}[u,j]\times[w,r]$.
 
%


%
%
%
%

\medskip\hrule 

\section{Witnessing Syntactic-Degrees}\label{sec:witnessing-syntactic-degrees}

\para{Witnesses for syntactic-degrees}
For most part our work we do not need to witness precise syntactic-degree of nodes, since syntactic-degree upper bounds are enough. However for the Cayley-Hamilton theorem we need to have witnesses for precise syntactic-degrees of nodes. We sketch here how to obtain such witnesses.

Note that computing the syntactic-degree of a node in a circuit is doable   in \NCTwo. This was noted for example by Allender \textit{et al}.~\cite{AJMV98} (replace every scalar gate by $0$, every variable gate by $1$, every product gate  by $+$ and every plus gate by $\max$, and then evaluate the circuit within \NCTwo; e.g., using the algorithm implicit in \cite{AJMV98}, or the algorithm in \cite{MRK88}). However, to actually use this algorithm in the theory we would need also to prove its correctness; this is likely doable (as we essentially show for the \cite{VSB+83} circuit balancing algorithm in Section \ref{sec:balancing-algebraic-circ-in-theory}), but we will opt for a shorter solution: we  simply witness the  syntactic-degrees  of all the specific circuits (and their nodes) we need.

The witness for the syntactic-degrees of nodes in a circuit is a string that stores pairs of numbers $(v,d)$, with $v$ the node label and $d$ its syntactic-degree. We can store each syntactic-degree as a natural number since we will need to witness only circuits with polynomial syntactic-degrees. 

It is easy to formulate a \SigZB-formula $\phi(C,W)$ with $C$ a circuit and $W$ the string that contains all the syntactic-degrees of the nodes in $C$, such that $\phi(C,W)$ holds iff $W$ is correct: for every addition gate $t=v_1+v_2$  it checks that $\deg(t)=\max\{\deg(v_1),\deg(v_2)\}$, and for every product gate $t=v_1\cd v_2$ it checks that $\deg(t)=\deg(v_1)+\deg(v_2)$, and for leaves it checks $\deg(x_i)=1$ and $\deg(c)=0$ for $c\in\Z$.

\iddo{}
\begin{lemma}[in \VZ]\label{lem:supply-witness-for-syntactic-degrees}
There exists a witness for the syntactic-degree of all nodes in $\taydetsharp(X)$.
\end{lemma}

\begin{proof}
We show how to witness in the theory \VZ\ the syntactic-degrees of the nodes in $\taydetsharp(X)$.  Recall the definition of $\taydetsharp(X)$ in \eqref{eq:huge-term-for-detsharp-spelled-out}. 
In order to compute the syntactic-degree of nodes we do the following. 

First we show that there is a \SigZB-definable number function in \VZ\ that computes the syntactic-degree of a node $v$, given $n$ and the node $v$ as inputs, where $v$ is a node in $\Num(\cdeti(I_n+zX))$, assuming the syntactic-degree of $v$ is at most $n$. The case for $\Den(\cdeti(I_n+zX))$ is similar. From this, using \eqref{eq:huge-term-for-detsharp-spelled-out}, we can conclude that  there is a \SigZB-definable number function in \VZ\ that computes the syntactic-degree of a node $v$, given $n$, in  $\taydetsharp(X)$, for all nodes of syntactic-degree at most $n$. 


Recall the encoding scheme for circuits $\cdeti(I_n+zX)$ described in \ref{sec:vz-function-for-constructing-DET-circuit}, and let $d$ denote the ``inductive level'' in the definition of \cdeti\ in \eqref{eq:det-min-one-definition}.  To compute the syntactic-degrees of  nodes that are at most $n$ in $\Num(\cdeti(I_n+zX))$ we wish to compute the \emph{pair} of numbers corresponding to the syntactic-degrees of  $(\Num(v), \Den(v))$ for each node $v$ in $\cdeti(I_n+zX)$. 

Observe that every inductive level $d$ in the circuit $\cdeti(X)$ has a ``base'' syntactic-degree (as a function of $d$), on top of which we add a number that depends on the gate we consider. For example, consider the circuit $F_1:={X_{1}^{-1}(I_{n-1}+\dd(X)^{-1}v_{1}^{t} v_{2 }X_{1}^{-1})}$ from \eqref{eq:def-inverse}. If we know the syntactic-degree of the output nodes in level $n-1$, namely the output nodes of $X_1^{-1}$, then we can easily compute the syntactic-degrees of other nodes in $F_1$. Note however that this cannot be computed inductively in such a way within \VZ, rather we need to show the explicit number functions. Also, notice the the syntactic-degree of some nodes in $F_1$ is exponential because the repeated multiplication of $X_1^{-1}$ by itself, hence we shall need to consider only those nodes whose syntactic-degree is polynomial in $n$. 

It is enough to show that there is a \SigZB-formula that determines the (polynomial-bounded) syntactic-degree pair $(\Num(v),\Den(v))$ of a node $v$, based on the inductive level $d$ in which the node occurs and the type or position of the gate in that level. For example, some gates in $F_1$, for every level $d$, are leaves---for instance, the entries of $I_{n-1}$ correspond to scalar leaves that have syntactic-degree pair $(0,0)$,  and some others are variable leaves---for instance, $v_{1}^{t} v_{2}$ corresponds to an inner product with leaves variables from $v_1,v_2$, having syntactic-degree pair $(1,0)$, for every level $d$. 

\iddo{}

We demonstrate this idea on $\dd(X)^{-1}$ which is the $(n,n)$ entry of $X^{-1}$ of dimension \nbyn. Similar reasoning works for the rest of the entries of $X^{-1}$ as well as  $\cdeti(I_n+zX)$.

For instance, if $n=2$ (note that $\dd(X)$ is defined for $X$ of dimension \nbyn, only for $n>1$), then $\dd(X)^{-1}=(x_{22}-x_{21}\cd x_{11}^{-1}\cd x_{12})^{-1}$. Thus, 
$\Num(\dd(X)^{-1})=\Den(x_{22})\cd\Den(x_{21}\cd x_{11}^{-1}\cd x_{12})=\Den(x_{22})\cd\Den(x_{21})\cd \Den(x_{11}^{-1})\cd\Den(x_{12})=1\cd 1\cd x_{11}\cd 1$. Hence, $\deg(\Num(\dd(X)^{-1}))=1$. %




\iddo{}
\end{proof}

Using witnesses for syntactic-degrees we can prove Lemma \ref{lem:Taylor-basic-properties} part 3, which was used in the proof of Theorem \ref{cor:C-H-theorem}.  
 
\begin{lemma}[in \VZ]\label{lem:Annex-F=sum-F-k} Given a division free circuit $F$ of syntactic-degree $d$ and a witness for the syntactic-degrees of all nodes in $F$, there exists a $\PCZ$-proof of ${F= \sum_{k=0}^d F^{(k)}}$. Moreover, $F= \sum_{i=0}^d \Coef_{z^{i}}(F)\cdot  z^{i}$ has a $\PCZ$-proof.
\end{lemma}

\begin{proof} \iddo{}

We shall prove the first statement (the second is similar).
Note that a big sum is an abbreviation of a sum written as a logarithmic depth tree of plus gates with the summands at the leaves (we also need to use  obvious steps such as applying the associativity and commutativity of addition axioms in \PCZ-proofs of big sums).

For every node $v$ in $F$ we construct \emph{simultaneously} a (partial) \PCZ-proof sequence terminating with
\begin{align}\label{eq:homogeneous-proof-proof}
& F_v= \sum_{k=0}^{\deg(v)} F_v^{(k)} 
\end{align} 
as follows:

\case 1 $v$ is a variable $x_i$. Then we construct a proof of $F_v:=x_i = \sum_{k=0}^{\deg(v)} \hc{F_v}{k}$, which is immediate by construction. Similarly for a constant node. 
\iddo{}   


\case 2 $v = u\oplus w$ and let $d=\deg(v)$. 
Then we use Lemma \ref{lem:bounded degree} to construct the following (partial) \PCZ\ proof-sequence. In the witnesses for this  proof-sequence we add pointers to proof-lines that are constructed in parallel (for nodes that appear closer to the leaf in the tree). We can compute the line numbers to be pointed to just by looking at the current node (hence we can carry out the construction in  \VZ). The pointers are constructed as number-functions by using the nodes (e.g., we can  label line numbers with the nodes in $F$ they correspond to, adding a secondary index to the index of the line). \iddo{}\iddo{}\iddo{}
%
%
\begin{align*}
\sum\nolimits_{i=0}^d \hci F_v  
& = \sum\nolimits_{i=0}^d \hci{(F_u\oplus F_w)} & \text{by assumption}\\
& = \sum\nolimits_{i=0}^d (\hci{F_u}+\hci{F_w}) & \text{by Lemma \ref{lem:bounded degree}} \\
& = \sum\nolimits_{i=0}^d \hci{F_u}+ \sum\nolimits_{i=0}^d \hci{F_w} & \text{rearrangement} \\
& = F_u+F_w &   \text{by ``previous'' lines}\\ 
& & \text{{\small (add explicit pointers to the appropriate  proof-lines)}}\\
& = F_u\cplus F_w = F_v & \text{by axiom C1.}
\end{align*} 
 

\case 3 $v = u\otimes w$ and let $d_1=\deg(u), d_2=\deg(w)$ and $d=\deg(v)=d_1+d_2$. This is similar to the Case 2 only that it is crucial here to use the specified syntactic-degrees of nodes along paths from leaves to the root. 
\iddo{}          
\begin{align*}
\sum\nolimits_{i=0}^d \hci F_v  
& = \sum\nolimits_{i=0}^d \hci{(F_u\otimes F_w)} & \text{by assumption}\\
& =  \sum\nolimits_{i=0}^d  \sum_{l+j=i \atop 0\le l\le d_1 ,0\le j\le d_2} F_1^{(l)}\cd F_2^{(j)}\,.
& \text{by Lemma \ref{lem:bounded degree} part (3)} \\
& = \sum\nolimits_{i=0}^{d_1} \hci{F_u}\cd\sum\nolimits_{i=0}^{d_2} \hci{F_w} & \text{rearrangement} \\
& = F_u\cd F_w &   \text{by ``previous'' lines}\\ 
& & \text{{\small (add explicit pointers to the appropriate  proof-lines)}}\\
& = F_u\ctimes F_w = F_v & \text{by axiom C1.}
\end{align*} 
\end{proof}

\section{Remaining Proof of Proposition \ref{prop:lengthy-lemma}}
\label{sec:balancing-main-lemma-verbatim}

\begin{proof}[Proof of Proposition \ref{prop:lengthy-lemma} continued]
\mbox{}
This is taken almost verbatim from \cite{HT12}, except that we use \degubp\ in instead of the true syntactic-degrees of nodes, and noticing that all predicates we use (like  $t\in {\cal B}_m (F_v)$) are definable in \VNCTwo.
\medskip 

\case 2 Assume that $ v=v_1\cd v_2 $.
We wish to prove $ [F_{v}]=[F_{v_1}]\cd[F_{v_2}]$. Let $ m=2^i $. We assume without loss of generality that $ \degubp(v_1)\ge \degubp(v_2)$. By the definition of $ [\cd ] $, we have:
\begin{align}
[F_v] = [F_{v_1\cd v_2}] & =  \sum_{t\in{\cal B}_m(F_v)} [F_{t_1}]\cd[F_{t_2}]\cd[\partial t F_v].   \notag
\end{align}
If $ v\in {\cal B}_m(F_v) $, then $ {\cal B}_m = \{v\} $ and we have $\tr{F_{v}}=[F_{v_1}]\cd[F_{v_2}]\cd[\partial_{v} F_v]$. Since $[\partial_{v}F_{v}]=1$, this gives $[F_{v}]= [F_{v_{1}}]\cd[F_{v_{2}}]$, and we are done.

Otherwise, assume $ v\not\in {\cal B}_m (F_v) $. Then   $m=2^i<\degubp(v_{1})$ (since, if $ \degubp(v_{1})\le m $, then also $ \degubp(v_{2})\le m $ and so by definition $ v\in {\cal B}_m(F_v) $). Further, because  $\degubp(v_1)\leq 2^{i+1}$, we have
\begin{equation} \label{eq: new2}
[F_{v_{1}}]=                                    \sum_{t\in {\cal B}_m(F_{v_1})} [F_{t_1}] \cd [F_{t_2}] \cd [\partial t F_{v_1}]\,.
\end{equation}

Since $ \degubp(v) \le 2^{i+1} $ and   $ \degubp(t)>m=2^{i} $, for any $ t \in {\cal B}_m(F_v) $, we have
\[
          \degubp(v) - \degubp(t) \le 2^i \quad\text{ and } \quad 2\degubp(t)>\degubp(v).
\]
Since $ v\neq t $, $ \Psi_i $ contains, for any $ t\in {\cal B}_m(F_v) $, the equation:
\begin{equation}
 [\partial t (F_{v_1\cd v_2})] = [\partial t F_{v_1}]\cd[F_{v_2}].\label{eq:076}
\end{equation}
Using (\ref{eq:076}) for all $ t \in {\cal B}_m(F_v) $, we can prove the following with a $ \PC(\F) $ proof of size $ O(s^2\cd m(s,d) ) $ and depth $O(r(s,d))$:
\begin{align}
               \sum_{t\in{\cal B}_m(F_v)} [F_{t_1}] \cd [F_{t_2}] \cd [\partial t F_v]
                                & =        \sum_{t\in{\cal B}_m(F_{v})} [F_{t_1}] \cd [F_{t_2}] \cd  [\partial t (F_{v_1\cd v_2})]   \notag\\
                                & =        \sum_{t\in {\cal B}_m(F_{v})} [F_{t_1}] \cd [F_{t_2}] \cd  ([\partial t F_{v_1}]\cd [F_{v_2}])  \notag \\
                                & =  [F_{v_2}]\cd       \sum_{t\in {\cal B}_m(F_{v})}[F_{t_1}] \cd [F_{t_2}] \cd  [\partial t F_{v_1}]. \label{eq:044}
\end{align}
Since ${\cal B}_{m}(F_{v_1})\subseteq {\cal B}_{m}(F_{v})$, we can conclude as in (\ref{eq: biggerB}) that
 \[
                                        \sum_{t\in {\cal B}_m(F_{v})} [F_{t_1}] \cd [F_{t_2}] \cd [\partial t F_{v_1}]  = \sum_{t\in {\cal B}_m(F_{v_1})} [F_{t_1}] \cd [F_{t_2}] \cd [\partial t F_{v_1}]\, .\]
Using (\ref{eq: new2}),
(\ref{eq:044}) equals $[F_{v_{2}}]\cd[F_{v_{1}}]$.
The above proof-sequence (using $ \Psi_i $ as a premise) has size $ O(s^2\cd m(s,d) ) $ and depth $  O(r(s,d))$.
\FullSpace

We now append $ \Psi_{i} $ with all proof-sequences of $ [F_{v}]=[F_{v_1}]+[F_{v_2}]$ for every $ v $ from Case 1, and all proof-sequences of $ [F_{v}]=[F_{v_1}]\cd[F_{v_2}]$ for every $ v $ from Case 2. We obtain a proof-sequence  $ \Psi'_{i+1} $ of size
\[
        \lambda(s,i+1)  \le O(s^3\cd m(s,d))+  \lambda(s,i),
\]
and depth $O(r(s,d)) $.

In Part (b), we extend $ \Psi'_{i+1} $ with more proof-sequences to obtain the final $ \Psi_{i+1} $.

\para{Part (b): proof of (\ref{eq:wFv+}) and (\ref{eq:wFv*}).}
Let $ v\neq w $ be a pair of nodes in $ F $ such that $ w\in F_v $  and assume that
\[
         2^{i}<\degubp(v)-\degubp(w) \le 2^{i+1} \text{\ \ and\ \ \ }  2\degubp(w)>\degubp(v).
\]
Let
\[  m = 2^i + \degubp(w) .\]

\case 1
Suppose that $ v=v_1+v_2 $.  We need to prove
\begin{equation}\label{eq:stam baribua}
        [\partial w F_{v}]  = [\partial w F_{v_1}] + [\partial w F_{v_2}]
\end{equation}
based on $ \Psi_i $ as a premise. By construction of $ [\partial w F_{v} ] $,
\begin{align}
        [\partial w F_{v}]  & =       \sum_{t\in{\cal B}_m(F_v)} [\partial t F_v] \cd [\partial w F_{t_1}] \cd [F_{t_2}]\nonumber \\
                                & =       \sum_{t\in{\cal B}_m(F_v)} [\partial t (F_{v_1+v_2})] \cd [\partial w F_{t_1}] \cd [F_{t_2}] . \label{eq:072}
\end{align}
Since $\degubp(v_1)=\degubp(v_2)=\degubp(v)$, we also have
\begin{equation}\label{eq: new3}
[\partial w F_{v_e}]   =       \sum_{t\in{\cal B}_m(F_{v_e})} [\partial t F_{v_e}] \cd [\partial w F_{t_1}] \cd [F_{t_2}],\quad\hbox{ for } e=1,2\,.
\end{equation}
Since $ m = 2^i +\degubp(w) $, we have $ \degubp(t)>2^i +\degubp(w) $, for any $ t\in {\cal B}_m (F_v) $. Thus, by $ \degubp(v)-\degubp(w) \le 2^{i+1}$, we get that for any $ t\in {\cal B}_m (F_v) $:
\begin{gather*}
         \degubp(v)-\degubp(t)\le 2^i \text{\ \ \ and \ \ \ }  2\degubp(t)>\degubp(v),\text{\ \ \ and}
\\  t\neq v \text{\ (since $ t $ is a product gate)}.
\end{gather*}
Therefore, for any $ t\in {\cal B}_m (F_v) $, $ \Psi_i $ contains the equation
 \[ [\partial t (F_{v_1+v_2})] = [\partial t F_{v_1}]+[\partial t F_{v_2}].\]
Thus, based on $ \Psi_i $, we can prove that (\ref{eq:072}) equals:
\begin{align}\notag
                                &        \sum_{t\in{\cal B}_m(F_v)} ([\partial t F_{v_1}]+[\partial t F_{v_2}]) \cd [\partial w F_{t_1}] \cd [F_{t_2}] \\
                         = &        \sum_{t\in{\cal B}_m(F_v)} [\partial t F_{v_1}] \cd  [\partial w F_{t_1}] \cd [F_{t_2}] +
                                                      \sum_{t\in{\cal B}_m(F_v)} [\partial t F_{v_2}] \cd  [\partial w F_{t_1}] \cd [F_{t_2}].\label{eq:235}
\end{align}
As in (\ref{eq: biggerB}), using (\ref{eq: new3}) we can derive the following from (\ref{eq:235}):
\[
\begin{split}
                                       \sum_{t\in{\cal B}_m(F_{v_1})} [\partial t F_{v_1}] \cd  [\partial w F_{t_1}] \cd [F_{t_2}] +
                                                      \sum_{t\in{\cal B}_m(F_{v_2})} [\partial t F_{v_2}] \cd  [\partial w F_{t_1}] \cd [F_{t_2}]
                                                \\
                                                =[\partial w F_{v_1}] + [\partial w F_{v_2}] .
\end{split}
\]
The proof of (\ref{eq:stam baribua}) from $ \Psi_i $ shown above has size $ O(s^2\cd m(s,d)) $ and depth $ O(r(s,d)) $.
\FullSpace


\case 2
Suppose that $v = v_1\cd v_2$. We assume without loss of generality that $ \degubp(v_1)\ge\degubp(v_2) $ and show how to prove
\begin{equation}\label{eq:stam bashlishit}
[\partial w F_{v}]  = [\partial w F_{v_1}] \cd [F_{v_2}] .
\end{equation}
By construction of $ [\partial w F_{v} ] $:
\begin{align}\notag
        [\partial w F_{v}]  & =       \sum_{t\in{\cal B}_m(F_v)} [\partial t F_v] \cd [\partial w F_{t_1}] \cd [F_{t_2}] \\
                                & =       \sum_{t\in{\cal B}_m(F_v)} [\partial t (F_{v_1 \cd  v_2})] \cd [\partial w F_{t_1}] \cd [F_{t_2}] . \label{eq:224}
\end{align}
Similar to the previous case, for any $ t\in {\cal B}_m (F_v) $ we have
\[ \degubp(v)-\degubp(t)< 2^i \text{\ \ \ and \ \ \ }  2\degubp(t)>\degubp(v).\]

If $ v\in {\cal B}_m(F_v) $ then $ {\cal B}_m(F_v) =\{v\}$ and so (\ref{eq:224}) is simply $[\partial v F_{v}]  \cd  [\partial w F_{v_1}] \cd  [F_{v_2}] =   [\partial w F_{v_1}] \cd  [F_{v_2}]$ as required. Otherwise, assume that $ v\not\in {\cal B}_{m}(F_{v}) $.
Then  $ \Psi_i $ contains the following equation, for any $t\in {\cal B}_m (F_v) $:
\[
        [\partial t (F_{v_1\cd v_2}) ] = [\partial t F_{v_1}] \cd [F_{v_2}].
\]
Using premises from $ \Psi_i $, we can then prove that (\ref{eq:224}) equals:
\begin{multline}\label{eq:056}
                                      \sum_{t\in{\cal B}_m(F_v)} \left([\partial t F_{v_1}] \cd [F_{v_2}]\right) \cd [\partial w F_{t_1}] \cd [F_{t_2}]
                                 = \left(       \sum_{t\in{\cal B}_m(F_v)} [\partial t F_{v_1}] \cd [\partial w F_{t_1}] \cd [F_{t_2}]\right) \cd [F_{v_2}].\ \ \ \
\end{multline}
As in (\ref{eq: biggerB}), we have
$\sum_{t\in{\cal B}_m(F_v)} [\partial t F_{v_1}] \cd [\partial w F_{t_1}] \cd [F_{t_2}]
= \sum_{t\in{\cal B}_m(F_{v_1})} [\partial t F_{v_1}] \cd [\partial w F_{t_1}] \cd [F_{t_2}]$. Also, since $v_{1}\cd v_{2}=v\not \in {\cal B}_{m}(F_{v})$, we have $\degubp(v_{1})>m= 2^{i}+\degubp(w)$, and so
\begin{equation} \label{eq: new4}
        [\partial w F_{v_1}]    =
                \sum_{t\in{\cal B}_m(F_{v_1})} [\partial t F_{v_1}] \cd [\partial w F_{t_1}] \cd [F_{t_2}]\,.
\end{equation}
Hence by (\ref{eq: new4}),
(\ref{eq:056}) equals
$[\partial_{w}F_{v_{1}}]\cd [F_{v_{2}}]$.

The above proof of (\ref{eq:stam bashlishit}) from $ \Psi_i $ has size $ O(s^2\cd m(s,d)) $ and depth $O(r(s,d)) $.
\FullSpace

We now append $ \Psi'_{i} $ from Part (a) (which also contains $ \Psi_i $) with all proof-sequences of $ [\partial w F_{v}]  = [\partial w F_{v_1}] + [\partial w F_{v_2}] $ in Case 1 and all proof sequences $ [\partial w F_{v}]  = [\partial w F_{v_1}] \cd [F_{v_2}]  $ in Case 2, above.  We obtain the proof-sequence  $ \Psi_{i+1} $ of size
\[
        \lambda(s,i+1)  \le O(s^4\cd m(s,d))+  \lambda(s,i),
\]
and depth $O(r(s,d)) $, as required.
\end{proof}


\singlespacing\small
\small  
\bibliographystyle{alphaurlpp}
\bibliography{PrfCmplx-Bakoma}


\end{document}

%% file: bibmacros.tex
\usepackage{nth}
\usepackage{intcalc}
\usepackage{etoolbox}
\usepackage{xstring}

\newcommand{\ECCCurl}[2]{http://eccc.hpi-web.de/report/\ifnumcomp{#1}{>}{93}{19}{20}#1/#2/}

\newcommand{\shortECCC}[2]{\texttt{\href{http://eccc.hpi-web.de/report/\ifnumcomp{#1}{>}{93}{19}{20}#1/#2/}{eccc:TR#1-#2}}}

\newcommand{\parseECCC}[1]{
\StrSubstitute{#1}{TR}{}[\tmpstring]%
\IfSubStr{\tmpstring}{/}{ 
\StrBefore{\tmpstring}{/}[\ecccyear]%
\StrBehind{\tmpstring}{/}[\ecccreport]%
}{
\StrBefore{\tmpstring}{-}[\ecccyear]%
\StrBehind{\tmpstring}{-}[\ecccreport]%
}%
\shortECCC{\ecccyear}{\ecccreport}}

%% file: tree.tex
\psset{xunit=.6pt,yunit=.6pt,runit=.5pt}
\begin{pspicture}(223.40269089,133.66209984)
{
\newrgbcolor{curcolor}{0 0 0}
\pscustom[linewidth=1,linecolor=curcolor]
{
\newpath
\moveto(34.42299,128.30691184)
\lineto(34.42299,94.02119784)
}
}
{
\newrgbcolor{curcolor}{0 0 0}
\pscustom[linewidth=1,linecolor=curcolor]
{
\newpath
\moveto(34.24442,92.41405584)
\lineto(33.88728,56.87833884)
}
}
{
\newrgbcolor{curcolor}{0 0 0}
\pscustom[linewidth=1,linecolor=curcolor]
{
\newpath
\moveto(36.20871,127.77119884)
\lineto(57.99442,94.73548484)
}
}
{
\newrgbcolor{curcolor}{0 0 0}
\pscustom[linewidth=1,linecolor=curcolor]
{
\newpath
\moveto(59.42299,91.69976984)
\lineto(81.20871,57.23548884)
}
}
{
\newrgbcolor{curcolor}{0 0 0}
\pscustom[linewidth=1,linecolor=curcolor]
{
\newpath
\moveto(35.67299,91.43191184)
\lineto(57.45871,56.96762884)
}
}
{
\newrgbcolor{curcolor}{0 0 0}
\pscustom[linewidth=1,linecolor=curcolor]
{
\newpath
\moveto(60.31293,91.69875684)
\lineto(108.53302,56.34363884)
}
}
{
\newrgbcolor{curcolor}{0 0 0}
\pscustom[linewidth=1,linecolor=curcolor]
{
\newpath
\moveto(33.39586,55.71707884)
\lineto(33.03872,20.18135884)
}
}
{
\newrgbcolor{curcolor}{0 0 0}
\pscustom[linewidth=1,linecolor=curcolor]
{
\newpath
\moveto(58.57443,55.00278884)
\lineto(80.36015,20.53850884)
}
}
{
\newrgbcolor{curcolor}{0 0 0}
\pscustom[linewidth=1,linecolor=curcolor]
{
\newpath
\moveto(34.82443,54.73492884)
\lineto(56.61015,20.27064884)
}
}
{
\newrgbcolor{curcolor}{0 0 0}
\pscustom[linewidth=1,linecolor=curcolor]
{
\newpath
\moveto(61.253,55.00278884)
\lineto(106.43157,20.35992884)
}
}
{
\newrgbcolor{curcolor}{0 0 0}
\pscustom[linewidth=1,linecolor=curcolor]
{
\newpath
\moveto(81.74549,54.36846884)
\lineto(138.17193,20.63820884)
}
}
{
\newrgbcolor{curcolor}{0 0 0}
\pscustom[linewidth=1,linecolor=curcolor]
{
\newpath
\moveto(85.49618,54.00016884)
\lineto(165.31409,22.07794884)
}
}
{
\newrgbcolor{curcolor}{0 0 0}
\pscustom[linewidth=1,linecolor=curcolor]
{
\newpath
\moveto(108.35156,55.27119884)
\lineto(191.56585,22.41405884)
}
}
{
\newrgbcolor{curcolor}{0 0 0}
\pscustom[linewidth=1,linecolor=curcolor]
{
\newpath
\moveto(113.35156,54.73548884)
\lineto(220.31585,21.69976884)
}
}
{
\newrgbcolor{curcolor}{0.7019608 0.7019608 0.7019608}
\pscustom[linestyle=none,fillstyle=solid,fillcolor=curcolor]
{
\newpath
\moveto(38.00620742,91.89844825)
\curveto(37.35932975,90.2727084)(35.47381625,89.46318719)(33.79180145,90.0890467)
\curveto(32.10978664,90.71490622)(31.26636128,92.53983607)(31.90661904,94.16803911)
\curveto(32.54687681,95.79624216)(34.42907084,96.61294576)(36.11360551,95.99349436)
\curveto(37.79814018,95.37404296)(38.64897329,93.55232742)(38.01534535,91.92168639)
}
}
{
\newrgbcolor{curcolor}{0.7019608 0.7019608 0.7019608}
\pscustom[linestyle=none,fillstyle=solid,fillcolor=curcolor]
{
\newpath
\moveto(38.36335243,129.13059684)
\curveto(37.71647476,127.50485699)(35.83096125,126.69533578)(34.14894645,127.32119529)
\curveto(32.46693164,127.94705481)(31.62350628,129.77198466)(32.26376405,131.4001877)
\curveto(32.90402181,133.02839075)(34.78621584,133.84509435)(36.47075051,133.22564295)
\curveto(38.15528518,132.60619155)(39.00611829,130.78447601)(38.37249035,129.15383498)
}
}
{
\newrgbcolor{curcolor}{0.7019608 0.7019608 0.7019608}
\pscustom[linestyle=none,fillstyle=solid,fillcolor=curcolor]
{
\newpath
\moveto(61.93477452,91.98774185)
\curveto(61.28789685,90.362002)(59.40238335,89.55248079)(57.72036854,90.1783403)
\curveto(56.03835373,90.80419982)(55.19492838,92.62912967)(55.83518614,94.25733271)
\curveto(56.4754439,95.88553576)(58.35763794,96.70223936)(60.0421726,96.08278796)
\curveto(61.72670727,95.46333656)(62.57754038,93.64162102)(61.94391245,92.01097999)
}
}
{
\newrgbcolor{curcolor}{0.7019608 0.7019608 0.7019608}
\pscustom[linestyle=none,fillstyle=solid,fillcolor=curcolor]
{
\newpath
\moveto(37.11335334,55.02345384)
\curveto(36.46647567,53.39771399)(34.58096216,52.58819278)(32.89894736,53.21405229)
\curveto(31.21693255,53.83991181)(30.37350719,55.66484166)(31.01376496,57.29304471)
\curveto(31.65402272,58.92124775)(33.53621675,59.73795135)(35.22075142,59.11849995)
\curveto(36.90528609,58.49904855)(37.7561192,56.67733301)(37.12249126,55.04669198)
}
}
{
\newrgbcolor{curcolor}{0.7019608 0.7019608 0.7019608}
\pscustom[linestyle=none,fillstyle=solid,fillcolor=curcolor]
{
\newpath
\moveto(61.3990631,54.66631067)
\curveto(60.75218543,53.04057082)(58.86667192,52.23104961)(57.18465712,52.85690912)
\curveto(55.50264231,53.48276864)(54.65921695,55.30769849)(55.29947472,56.93590154)
\curveto(55.93973248,58.56410458)(57.82192651,59.38080818)(59.50646118,58.76135678)
\curveto(61.19099585,58.14190538)(62.04182896,56.32018984)(61.40820102,54.68954881)
}
}
{
\newrgbcolor{curcolor}{0.7019608 0.7019608 0.7019608}
\pscustom[linestyle=none,fillstyle=solid,fillcolor=curcolor]
{
\newpath
\moveto(86.6133658,54.84487727)
\curveto(85.96648813,53.21913742)(84.08097463,52.4096162)(82.39895982,53.03547572)
\curveto(80.71694501,53.66133523)(79.87351965,55.48626508)(80.51377742,57.11446813)
\curveto(81.15403518,58.74267117)(83.03622921,59.55937477)(84.72076388,58.93992337)
\curveto(86.40529855,58.32047198)(87.25613166,56.49875643)(86.62250372,54.86811541)
}
}
{
\newrgbcolor{curcolor}{0.7019608 0.7019608 0.7019608}
\pscustom[linestyle=none,fillstyle=solid,fillcolor=curcolor]
{
\newpath
\moveto(113.75622264,55.20202165)
\curveto(113.10934497,53.5762818)(111.22383146,52.76676058)(109.54181666,53.3926201)
\curveto(107.85980185,54.01847962)(107.01637649,55.84340946)(107.65663426,57.47161251)
\curveto(108.29689202,59.09981556)(110.17908605,59.91651915)(111.86362072,59.29706776)
\curveto(113.54815539,58.67761636)(114.3989885,56.85590081)(113.76536056,55.22525979)
}
}
{
\newrgbcolor{curcolor}{0.7019608 0.7019608 0.7019608}
\pscustom[linestyle=none,fillstyle=solid,fillcolor=curcolor]
{
\newpath
\moveto(36.04191203,17.70202545)
\curveto(35.39503436,16.0762856)(33.50952086,15.26676439)(31.82750605,15.8926239)
\curveto(30.14549125,16.51848342)(29.30206589,18.34341327)(29.94232365,19.97161632)
\curveto(30.58258142,21.59981936)(32.46477545,22.41652296)(34.14931012,21.79707156)
\curveto(35.83384479,21.17762016)(36.6846779,19.35590462)(36.05104996,17.72526359)
}
}
{
\newrgbcolor{curcolor}{0.7019608 0.7019608 0.7019608}
\pscustom[linestyle=none,fillstyle=solid,fillcolor=curcolor]
{
\newpath
\moveto(60.50620877,17.34487986)
\curveto(59.8593311,15.71914001)(57.9738176,14.90961879)(56.29180279,15.53547831)
\curveto(54.60978798,16.16133782)(53.76636263,17.98626767)(54.40662039,19.61447072)
\curveto(55.04687815,21.24267376)(56.92907218,22.05937736)(58.61360685,21.43992596)
\curveto(60.29814152,20.82047457)(61.14897463,18.99875902)(60.5153467,17.368118)
}
}
{
\newrgbcolor{curcolor}{0.7019608 0.7019608 0.7019608}
\pscustom[linestyle=none,fillstyle=solid,fillcolor=curcolor]
{
\newpath
\moveto(83.54190582,17.34488808)
\curveto(82.89502815,15.71914823)(81.00951464,14.90962701)(79.32749984,15.53548653)
\curveto(77.64548503,16.16134604)(76.80205967,17.98627589)(77.44231744,19.61447894)
\curveto(78.0825752,21.24268199)(79.96476923,22.05938558)(81.6493039,21.43993418)
\curveto(83.33383857,20.82048279)(84.18467168,18.99876724)(83.55104374,17.36812622)
}
}
{
\newrgbcolor{curcolor}{0.7019608 0.7019608 0.7019608}
\pscustom[linestyle=none,fillstyle=solid,fillcolor=curcolor]
{
\newpath
\moveto(109.61334596,17.34488209)
\curveto(108.96646829,15.71914224)(107.08095479,14.90962102)(105.39893998,15.53548054)
\curveto(103.71692517,16.16134005)(102.87349982,17.9862699)(103.51375758,19.61447295)
\curveto(104.15401534,21.242676)(106.03620937,22.05937959)(107.72074404,21.4399282)
\curveto(109.40527871,20.8204768)(110.25611182,18.99876125)(109.62248389,17.36812023)
}
}
{
\newrgbcolor{curcolor}{0.7019608 0.7019608 0.7019608}
\pscustom[linestyle=none,fillstyle=solid,fillcolor=curcolor]
{
\newpath
\moveto(142.64906955,17.52345187)
\curveto(142.00219188,15.89771202)(140.11667838,15.08819081)(138.43466357,15.71405032)
\curveto(136.75264877,16.33990984)(135.90922341,18.16483969)(136.54948117,19.79304273)
\curveto(137.18973894,21.42124578)(139.07193297,22.23794938)(140.75646764,21.61849798)
\curveto(142.44100231,20.99904658)(143.29183542,19.17733104)(142.65820748,17.54669001)
}
}
{
\newrgbcolor{curcolor}{0.7019608 0.7019608 0.7019608}
\pscustom[linestyle=none,fillstyle=solid,fillcolor=curcolor]
{
\newpath
\moveto(169.79191261,18.23774328)
\curveto(169.14503494,16.61200343)(167.25952144,15.80248221)(165.57750663,16.42834173)
\curveto(163.89549183,17.05420125)(163.05206647,18.8791311)(163.69232423,20.50733414)
\curveto(164.332582,22.13553719)(166.21477603,22.95224078)(167.8993107,22.33278939)
\curveto(169.58384537,21.71333799)(170.43467848,19.89162244)(169.80105054,18.26098142)
}
}
{
\newrgbcolor{curcolor}{0.7019608 0.7019608 0.7019608}
\pscustom[linestyle=none,fillstyle=solid,fillcolor=curcolor]
{
\newpath
\moveto(194.61336103,18.77344973)
\curveto(193.96648336,17.14770988)(192.08096986,16.33818866)(190.39895505,16.96404818)
\curveto(188.71694025,17.58990769)(187.87351489,19.41483754)(188.51377265,21.04304059)
\curveto(189.15403041,22.67124363)(191.03622445,23.48794723)(192.72075912,22.86849583)
\curveto(194.40529379,22.24904444)(195.2561269,20.42732889)(194.62249896,18.79668787)
}
}
{
\newrgbcolor{curcolor}{0.7019608 0.7019608 0.7019608}
\pscustom[linestyle=none,fillstyle=solid,fillcolor=curcolor]
{
\newpath
\moveto(223.18476607,18.23774472)
\curveto(222.5378884,16.61200487)(220.6523749,15.80248366)(218.97036009,16.42834317)
\curveto(217.28834528,17.05420269)(216.44491993,18.87913254)(217.08517769,20.50733558)
\curveto(217.72543545,22.13553863)(219.60762948,22.95224223)(221.29216415,22.33279083)
\curveto(222.97669882,21.71333943)(223.82753193,19.89162389)(223.193904,18.26098286)
}
}
{
\newrgbcolor{curcolor}{0 0 0}
\pscustom[linestyle=none,fillstyle=solid,fillcolor=curcolor]
{
\newpath
\moveto(3.76758284,18.1961502)
\curveto(3.50586409,17.9227127)(3.25000472,17.72544707)(3.00000472,17.60435332)
\curveto(2.75000472,17.48716582)(2.48047347,17.42857207)(2.19141097,17.42857207)
\curveto(1.60547347,17.42857207)(1.09375472,17.6727127)(0.65625472,18.16099395)
\curveto(0.21875472,18.65318145)(0.00000472,19.28404082)(0.00000472,20.05357207)
\curveto(0.00000472,20.82310332)(0.24219222,21.52622832)(0.72656722,22.16294707)
\curveto(1.21094222,22.80357207)(1.83398909,23.12388457)(2.59570784,23.12388457)
\curveto(3.06836409,23.12388457)(3.45898909,22.97349395)(3.76758284,22.6727127)
\lineto(3.76758284,23.66294707)
\curveto(3.76758284,24.27622832)(3.75195784,24.65318145)(3.72070784,24.79380645)
\curveto(3.69336409,24.93443145)(3.64844222,25.03013457)(3.58594222,25.08091582)
\curveto(3.52344222,25.13169707)(3.44531722,25.1570877)(3.35156722,25.1570877)
\curveto(3.25000472,25.1570877)(3.11523909,25.1258377)(2.94727034,25.0633377)
\lineto(2.87109847,25.26841582)
\lineto(4.47070784,25.92466582)
\lineto(4.73437972,25.92466582)
\lineto(4.73437972,19.7195877)
\curveto(4.73437972,19.09068145)(4.74805159,18.70591582)(4.77539534,18.56529082)
\curveto(4.80664534,18.42857207)(4.85352034,18.33286895)(4.91602034,18.27818145)
\curveto(4.98242659,18.22349395)(5.05859847,18.1961502)(5.14453597,18.1961502)
\curveto(5.25000472,18.1961502)(5.39062972,18.22935332)(5.56641097,18.29575957)
\lineto(5.63086409,18.09068145)
\lineto(4.03711409,17.42857207)
\lineto(3.76758284,17.42857207)
\lineto(3.76758284,18.1961502)
\closepath
\moveto(3.76758284,18.60630645)
\lineto(3.76758284,21.37193145)
\curveto(3.74414534,21.63755645)(3.67383284,21.87974395)(3.55664534,22.09849395)
\curveto(3.43945784,22.31724395)(3.28320784,22.48130645)(3.08789534,22.59068145)
\curveto(2.89648909,22.7039627)(2.70898909,22.76060332)(2.52539534,22.76060332)
\curveto(2.18164534,22.76060332)(1.87500472,22.60630645)(1.60547347,22.2977127)
\curveto(1.25000472,21.8914627)(1.07227034,21.2977127)(1.07227034,20.5164627)
\curveto(1.07227034,19.7274002)(1.24414534,19.12193145)(1.58789534,18.70005645)
\curveto(1.93164534,18.2820877)(2.31445784,18.07310332)(2.73633284,18.07310332)
\curveto(3.09180159,18.07310332)(3.43555159,18.2508377)(3.76758284,18.60630645)
\closepath
}
}
{
\newrgbcolor{curcolor}{0 0 0}
\pscustom[linestyle=none,fillstyle=solid,fillcolor=curcolor]
{
\newpath
\moveto(5.81836409,22.7664627)
\lineto(12.14062972,22.7664627)
\lineto(12.14062972,22.28599395)
\lineto(5.81836409,22.28599395)
\lineto(5.81836409,22.7664627)
\closepath
\moveto(5.81836409,20.85630645)
\lineto(12.14062972,20.85630645)
\lineto(12.14062972,20.3758377)
\lineto(5.81836409,20.3758377)
\lineto(5.81836409,20.85630645)
\closepath
}
}
{
\newrgbcolor{curcolor}{0 0 0}
\pscustom[linestyle=none,fillstyle=solid,fillcolor=curcolor]
{
\newpath
\moveto(13.78125472,24.7586502)
\lineto(15.71484847,25.70200957)
\lineto(15.90820784,25.70200957)
\lineto(15.90820784,18.9930252)
\curveto(15.90820784,18.5477127)(15.92578597,18.27036895)(15.96094222,18.16099395)
\curveto(16.00000472,18.05161895)(16.07812972,17.96763457)(16.19531722,17.90904082)
\curveto(16.31250472,17.85044707)(16.55078597,17.81724395)(16.91016097,17.80943145)
\lineto(16.91016097,17.59263457)
\lineto(13.92187972,17.59263457)
\lineto(13.92187972,17.80943145)
\curveto(14.29687972,17.81724395)(14.53906722,17.84849395)(14.64844222,17.90318145)
\curveto(14.75781722,17.9617752)(14.83398909,18.03794707)(14.87695784,18.13169707)
\curveto(14.91992659,18.22935332)(14.94141097,18.5164627)(14.94141097,18.9930252)
\lineto(14.94141097,23.2820877)
\curveto(14.94141097,23.8602127)(14.92187972,24.23130645)(14.88281722,24.39536895)
\curveto(14.85547347,24.52036895)(14.80469222,24.61216582)(14.73047347,24.67075957)
\curveto(14.66016097,24.72935332)(14.57422347,24.7586502)(14.47266097,24.7586502)
\curveto(14.32812972,24.7586502)(14.12695784,24.69810332)(13.86914534,24.57700957)
\lineto(13.78125472,24.7586502)
\closepath
}
}
{
\newrgbcolor{curcolor}{0 0 0}
\pscustom[linestyle=none,fillstyle=solid,fillcolor=curcolor]
{
\newpath
\moveto(4.55692,52.68359221)
\curveto(4.29520125,52.41015471)(4.03934187,52.21288909)(3.78934187,52.09179534)
\curveto(3.53934187,51.97460784)(3.26981062,51.91601409)(2.98074812,51.91601409)
\curveto(2.39481062,51.91601409)(1.88309187,52.16015471)(1.44559187,52.64843596)
\curveto(1.00809187,53.14062346)(0.78934187,53.77148284)(0.78934187,54.54101409)
\curveto(0.78934187,55.31054534)(1.03152937,56.01367034)(1.51590437,56.65038909)
\curveto(2.00027937,57.29101409)(2.62332625,57.61132659)(3.385045,57.61132659)
\curveto(3.85770125,57.61132659)(4.24832625,57.46093596)(4.55692,57.16015471)
\lineto(4.55692,58.15038909)
\curveto(4.55692,58.76367034)(4.541295,59.14062346)(4.510045,59.28124846)
\curveto(4.48270125,59.42187346)(4.43777937,59.51757659)(4.37527937,59.56835784)
\curveto(4.31277937,59.61913909)(4.23465437,59.64452971)(4.14090437,59.64452971)
\curveto(4.03934187,59.64452971)(3.90457625,59.61327971)(3.7366075,59.55077971)
\lineto(3.66043562,59.75585784)
\lineto(5.260045,60.41210784)
\lineto(5.52371687,60.41210784)
\lineto(5.52371687,54.20702971)
\curveto(5.52371687,53.57812346)(5.53738875,53.19335784)(5.5647325,53.05273284)
\curveto(5.5959825,52.91601409)(5.6428575,52.82031096)(5.7053575,52.76562346)
\curveto(5.77176375,52.71093596)(5.84793562,52.68359221)(5.93387312,52.68359221)
\curveto(6.03934187,52.68359221)(6.17996687,52.71679534)(6.35574812,52.78320159)
\lineto(6.42020125,52.57812346)
\lineto(4.82645125,51.91601409)
\lineto(4.55692,51.91601409)
\lineto(4.55692,52.68359221)
\closepath
\moveto(4.55692,53.09374846)
\lineto(4.55692,55.85937346)
\curveto(4.5334825,56.12499846)(4.46317,56.36718596)(4.3459825,56.58593596)
\curveto(4.228795,56.80468596)(4.072545,56.96874846)(3.8772325,57.07812346)
\curveto(3.68582625,57.19140471)(3.49832625,57.24804534)(3.3147325,57.24804534)
\curveto(2.9709825,57.24804534)(2.66434187,57.09374846)(2.39481062,56.78515471)
\curveto(2.03934187,56.37890471)(1.8616075,55.78515471)(1.8616075,55.00390471)
\curveto(1.8616075,54.21484221)(2.0334825,53.60937346)(2.3772325,53.18749846)
\curveto(2.7209825,52.76952971)(3.103795,52.56054534)(3.52567,52.56054534)
\curveto(3.88113875,52.56054534)(4.22488875,52.73827971)(4.55692,53.09374846)
\closepath
}
}
{
\newrgbcolor{curcolor}{0 0 0}
\pscustom[linestyle=none,fillstyle=solid,fillcolor=curcolor]
{
\newpath
\moveto(6.60770125,57.25390471)
\lineto(12.92996687,57.25390471)
\lineto(12.92996687,56.77343596)
\lineto(6.60770125,56.77343596)
\lineto(6.60770125,57.25390471)
\closepath
\moveto(6.60770125,55.34374846)
\lineto(12.92996687,55.34374846)
\lineto(12.92996687,54.86327971)
\lineto(6.60770125,54.86327971)
\lineto(6.60770125,55.34374846)
\closepath
}
}
{
\newrgbcolor{curcolor}{0 0 0}
\pscustom[linestyle=none,fillstyle=solid,fillcolor=curcolor]
{
\newpath
\moveto(18.666295,53.60937346)
\lineto(18.10965437,52.08007659)
\lineto(13.42215437,52.08007659)
\lineto(13.42215437,52.29687346)
\curveto(14.80106062,53.55468596)(15.77176375,54.58202971)(16.33426375,55.37890471)
\curveto(16.89676375,56.17577971)(17.17801375,56.90429534)(17.17801375,57.56445159)
\curveto(17.17801375,58.06835784)(17.02371687,58.48242034)(16.71512312,58.80663909)
\curveto(16.40652937,59.13085784)(16.03738875,59.29296721)(15.60770125,59.29296721)
\curveto(15.21707625,59.29296721)(14.86551375,59.17773284)(14.55301375,58.94726409)
\curveto(14.24442,58.72070159)(14.01590437,58.38671721)(13.86746687,57.94531096)
\lineto(13.65067,57.94531096)
\curveto(13.74832625,58.66796721)(13.99832625,59.22265471)(14.40067,59.60937346)
\curveto(14.80692,59.99609221)(15.31277937,60.18945159)(15.91824812,60.18945159)
\curveto(16.56277937,60.18945159)(17.09988875,59.98242034)(17.52957625,59.56835784)
\curveto(17.96317,59.15429534)(18.17996687,58.66601409)(18.17996687,58.10351409)
\curveto(18.17996687,57.70117034)(18.08621687,57.29882659)(17.89871687,56.89648284)
\curveto(17.60965437,56.26367034)(17.14090437,55.59374846)(16.49246687,54.88671721)
\curveto(15.51981062,53.82421721)(14.91238875,53.18359221)(14.67020125,52.96484221)
\lineto(16.74442,52.96484221)
\curveto(17.166295,52.96484221)(17.46121687,52.98046721)(17.62918562,53.01171721)
\curveto(17.80106062,53.04296721)(17.9553575,53.10546721)(18.09207625,53.19921721)
\curveto(18.228795,53.29687346)(18.34793562,53.43359221)(18.44949812,53.60937346)
\lineto(18.666295,53.60937346)
\closepath
}
}
{
\newrgbcolor{curcolor}{0 0 0}
\pscustom[linestyle=none,fillstyle=solid,fillcolor=curcolor]
{
\newpath
\moveto(5.07352156,89.43359221)
\curveto(4.81180281,89.16015471)(4.55594344,88.96288909)(4.30594344,88.84179534)
\curveto(4.05594344,88.72460784)(3.78641219,88.66601409)(3.49734969,88.66601409)
\curveto(2.91141219,88.66601409)(2.39969344,88.91015471)(1.96219344,89.39843596)
\curveto(1.52469344,89.89062346)(1.30594344,90.52148284)(1.30594344,91.29101409)
\curveto(1.30594344,92.06054534)(1.54813094,92.76367034)(2.03250594,93.40038909)
\curveto(2.51688094,94.04101409)(3.13992781,94.36132659)(3.90164656,94.36132659)
\curveto(4.37430281,94.36132659)(4.76492781,94.21093596)(5.07352156,93.91015471)
\lineto(5.07352156,94.90038909)
\curveto(5.07352156,95.51367034)(5.05789656,95.89062346)(5.02664656,96.03124846)
\curveto(4.99930281,96.17187346)(4.95438094,96.26757659)(4.89188094,96.31835784)
\curveto(4.82938094,96.36913909)(4.75125594,96.39452971)(4.65750594,96.39452971)
\curveto(4.55594344,96.39452971)(4.42117781,96.36327971)(4.25320906,96.30077971)
\lineto(4.17703719,96.50585784)
\lineto(5.77664656,97.16210784)
\lineto(6.04031844,97.16210784)
\lineto(6.04031844,90.95702971)
\curveto(6.04031844,90.32812346)(6.05399031,89.94335784)(6.08133406,89.80273284)
\curveto(6.11258406,89.66601409)(6.15945906,89.57031096)(6.22195906,89.51562346)
\curveto(6.28836531,89.46093596)(6.36453719,89.43359221)(6.45047469,89.43359221)
\curveto(6.55594344,89.43359221)(6.69656844,89.46679534)(6.87234969,89.53320159)
\lineto(6.93680281,89.32812346)
\lineto(5.34305281,88.66601409)
\lineto(5.07352156,88.66601409)
\lineto(5.07352156,89.43359221)
\closepath
\moveto(5.07352156,89.84374846)
\lineto(5.07352156,92.60937346)
\curveto(5.05008406,92.87499846)(4.97977156,93.11718596)(4.86258406,93.33593596)
\curveto(4.74539656,93.55468596)(4.58914656,93.71874846)(4.39383406,93.82812346)
\curveto(4.20242781,93.94140471)(4.01492781,93.99804534)(3.83133406,93.99804534)
\curveto(3.48758406,93.99804534)(3.18094344,93.84374846)(2.91141219,93.53515471)
\curveto(2.55594344,93.12890471)(2.37820906,92.53515471)(2.37820906,91.75390471)
\curveto(2.37820906,90.96484221)(2.55008406,90.35937346)(2.89383406,89.93749846)
\curveto(3.23758406,89.51952971)(3.62039656,89.31054534)(4.04227156,89.31054534)
\curveto(4.39774031,89.31054534)(4.74149031,89.48827971)(5.07352156,89.84374846)
\closepath
}
}
{
\newrgbcolor{curcolor}{0 0 0}
\pscustom[linestyle=none,fillstyle=solid,fillcolor=curcolor]
{
\newpath
\moveto(7.12430281,94.00390471)
\lineto(13.44656844,94.00390471)
\lineto(13.44656844,93.52343596)
\lineto(7.12430281,93.52343596)
\lineto(7.12430281,94.00390471)
\closepath
\moveto(7.12430281,92.09374846)
\lineto(13.44656844,92.09374846)
\lineto(13.44656844,91.61327971)
\lineto(7.12430281,91.61327971)
\lineto(7.12430281,92.09374846)
\closepath
}
}
{
\newrgbcolor{curcolor}{0 0 0}
\pscustom[linestyle=none,fillstyle=solid,fillcolor=curcolor]
{
\newpath
\moveto(14.29031844,95.26367034)
\curveto(14.51688094,95.79882659)(14.80203719,96.21093596)(15.14578719,96.49999846)
\curveto(15.49344344,96.79296721)(15.92508406,96.93945159)(16.44070906,96.93945159)
\curveto(17.07742781,96.93945159)(17.56570906,96.73242034)(17.90555281,96.31835784)
\curveto(18.16336531,96.00976409)(18.29227156,95.67968596)(18.29227156,95.32812346)
\curveto(18.29227156,94.74999846)(17.92899031,94.15234221)(17.20242781,93.53515471)
\curveto(17.69070906,93.34374846)(18.05984969,93.07031096)(18.30984969,92.71484221)
\curveto(18.55984969,92.35937346)(18.68484969,91.94140471)(18.68484969,91.46093596)
\curveto(18.68484969,90.77343596)(18.46609969,90.17773284)(18.02859969,89.67382659)
\curveto(17.45828719,89.01757659)(16.63211531,88.68945159)(15.55008406,88.68945159)
\curveto(15.01492781,88.68945159)(14.64969344,88.75585784)(14.45438094,88.88867034)
\curveto(14.26297469,89.02148284)(14.16727156,89.16406096)(14.16727156,89.31640471)
\curveto(14.16727156,89.42968596)(14.21219344,89.52929534)(14.30203719,89.61523284)
\curveto(14.39578719,89.70117034)(14.50711531,89.74413909)(14.63602156,89.74413909)
\curveto(14.73367781,89.74413909)(14.83328719,89.72851409)(14.93484969,89.69726409)
\curveto(15.00125594,89.67773284)(15.15164656,89.60546721)(15.38602156,89.48046721)
\curveto(15.62039656,89.35937346)(15.78250594,89.28710784)(15.87234969,89.26367034)
\curveto(16.01688094,89.22070159)(16.17117781,89.19921721)(16.33524031,89.19921721)
\curveto(16.73367781,89.19921721)(17.07938094,89.35351409)(17.37234969,89.66210784)
\curveto(17.66922469,89.97070159)(17.81766219,90.33593596)(17.81766219,90.75781096)
\curveto(17.81766219,91.06640471)(17.74930281,91.36718596)(17.61258406,91.66015471)
\curveto(17.51102156,91.87890471)(17.39969344,92.04492034)(17.27859969,92.15820159)
\curveto(17.11063094,92.31445159)(16.88016219,92.45507659)(16.58719344,92.58007659)
\curveto(16.29422469,92.70898284)(15.99539656,92.77343596)(15.69070906,92.77343596)
\lineto(15.50320906,92.77343596)
\lineto(15.50320906,92.94921721)
\curveto(15.81180281,92.98827971)(16.12039656,93.09960784)(16.42899031,93.28320159)
\curveto(16.74149031,93.46679534)(16.96805281,93.68749846)(17.10867781,93.94531096)
\curveto(17.24930281,94.20312346)(17.31961531,94.48632659)(17.31961531,94.79492034)
\curveto(17.31961531,95.19726409)(17.19266219,95.52148284)(16.93875594,95.76757659)
\curveto(16.68875594,96.01757659)(16.37625594,96.14257659)(16.00125594,96.14257659)
\curveto(15.39578719,96.14257659)(14.88992781,95.81835784)(14.48367781,95.16992034)
\lineto(14.29031844,95.26367034)
\closepath
}
}
{
\newrgbcolor{curcolor}{0 0 0}
\pscustom[linestyle=none,fillstyle=solid,fillcolor=curcolor]
{
\newpath
\moveto(5.57353682,125.93358458)
\curveto(5.31181807,125.66014708)(5.0559587,125.46288146)(4.8059587,125.34178771)
\curveto(4.5559587,125.22460021)(4.28642745,125.16600646)(3.99736495,125.16600646)
\curveto(3.41142745,125.16600646)(2.8997087,125.41014708)(2.4622087,125.89842833)
\curveto(2.0247087,126.39061583)(1.8059587,127.02147521)(1.8059587,127.79100646)
\curveto(1.8059587,128.56053771)(2.0481462,129.26366271)(2.5325212,129.90038146)
\curveto(3.0168962,130.54100646)(3.63994307,130.86131896)(4.40166182,130.86131896)
\curveto(4.87431807,130.86131896)(5.26494307,130.71092833)(5.57353682,130.41014708)
\lineto(5.57353682,131.40038146)
\curveto(5.57353682,132.01366271)(5.55791182,132.39061583)(5.52666182,132.53124083)
\curveto(5.49931807,132.67186583)(5.4543962,132.76756896)(5.3918962,132.81835021)
\curveto(5.3293962,132.86913146)(5.2512712,132.89452208)(5.1575212,132.89452208)
\curveto(5.0559587,132.89452208)(4.92119307,132.86327208)(4.75322432,132.80077208)
\lineto(4.67705245,133.00585021)
\lineto(6.27666182,133.66210021)
\lineto(6.5403337,133.66210021)
\lineto(6.5403337,127.45702208)
\curveto(6.5403337,126.82811583)(6.55400557,126.44335021)(6.58134932,126.30272521)
\curveto(6.61259932,126.16600646)(6.65947432,126.07030333)(6.72197432,126.01561583)
\curveto(6.78838057,125.96092833)(6.86455245,125.93358458)(6.95048995,125.93358458)
\curveto(7.0559587,125.93358458)(7.1965837,125.96678771)(7.37236495,126.03319396)
\lineto(7.43681807,125.82811583)
\lineto(5.84306807,125.16600646)
\lineto(5.57353682,125.16600646)
\lineto(5.57353682,125.93358458)
\closepath
\moveto(5.57353682,126.34374083)
\lineto(5.57353682,129.10936583)
\curveto(5.55009932,129.37499083)(5.47978682,129.61717833)(5.36259932,129.83592833)
\curveto(5.24541182,130.05467833)(5.08916182,130.21874083)(4.89384932,130.32811583)
\curveto(4.70244307,130.44139708)(4.51494307,130.49803771)(4.33134932,130.49803771)
\curveto(3.98759932,130.49803771)(3.6809587,130.34374083)(3.41142745,130.03514708)
\curveto(3.0559587,129.62889708)(2.87822432,129.03514708)(2.87822432,128.25389708)
\curveto(2.87822432,127.46483458)(3.05009932,126.85936583)(3.39384932,126.43749083)
\curveto(3.73759932,126.01952208)(4.12041182,125.81053771)(4.54228682,125.81053771)
\curveto(4.89775557,125.81053771)(5.24150557,125.98827208)(5.57353682,126.34374083)
\closepath
}
}
{
\newrgbcolor{curcolor}{0 0 0}
\pscustom[linestyle=none,fillstyle=solid,fillcolor=curcolor]
{
\newpath
\moveto(7.62431807,130.50389708)
\lineto(13.9465837,130.50389708)
\lineto(13.9465837,130.02342833)
\lineto(7.62431807,130.02342833)
\lineto(7.62431807,130.50389708)
\closepath
\moveto(7.62431807,128.59374083)
\lineto(13.9465837,128.59374083)
\lineto(13.9465837,128.11327208)
\lineto(7.62431807,128.11327208)
\lineto(7.62431807,128.59374083)
\closepath
}
}
{
\newrgbcolor{curcolor}{0 0 0}
\pscustom[linestyle=none,fillstyle=solid,fillcolor=curcolor]
{
\newpath
\moveto(19.76494307,128.25975646)
\lineto(19.76494307,127.42772521)
\lineto(18.69853682,127.42772521)
\lineto(18.69853682,125.33006896)
\lineto(17.73173995,125.33006896)
\lineto(17.73173995,127.42772521)
\lineto(14.3684587,127.42772521)
\lineto(14.3684587,128.17772521)
\lineto(18.05400557,133.43944396)
\lineto(18.69853682,133.43944396)
\lineto(18.69853682,128.25975646)
\lineto(19.76494307,128.25975646)
\closepath
\moveto(17.73173995,128.25975646)
\lineto(17.73173995,132.20311583)
\lineto(14.94267745,128.25975646)
\lineto(17.73173995,128.25975646)
\closepath
}
}
{
\newrgbcolor{curcolor}{0 0 0}
\pscustom[linestyle=none,fillstyle=solid,fillcolor=curcolor]
{
\newpath
\moveto(33.40067,7.33007659)
\lineto(35.33426375,8.27343596)
\lineto(35.52762312,8.27343596)
\lineto(35.52762312,1.56445159)
\curveto(35.52762312,1.11913909)(35.54520125,0.84179534)(35.5803575,0.73242034)
\curveto(35.61942,0.62304534)(35.697545,0.53906096)(35.8147325,0.48046721)
\curveto(35.93192,0.42187346)(36.17020125,0.38867034)(36.52957625,0.38085784)
\lineto(36.52957625,0.16406096)
\lineto(33.541295,0.16406096)
\lineto(33.541295,0.38085784)
\curveto(33.916295,0.38867034)(34.1584825,0.41992034)(34.2678575,0.47460784)
\curveto(34.3772325,0.53320159)(34.45340437,0.60937346)(34.49637312,0.70312346)
\curveto(34.53934187,0.80077971)(34.56082625,1.08788909)(34.56082625,1.56445159)
\lineto(34.56082625,5.85351409)
\curveto(34.56082625,6.43163909)(34.541295,6.80273284)(34.5022325,6.96679534)
\curveto(34.47488875,7.09179534)(34.4241075,7.18359221)(34.34988875,7.24218596)
\curveto(34.27957625,7.30077971)(34.19363875,7.33007659)(34.09207625,7.33007659)
\curveto(33.947545,7.33007659)(33.74637312,7.26952971)(33.48856062,7.14843596)
\lineto(33.40067,7.33007659)
\closepath
}
}
{
\newrgbcolor{curcolor}{0 0 0}
\pscustom[linestyle=none,fillstyle=solid,fillcolor=curcolor]
{
\newpath
\moveto(58.49637312,1.69335784)
\lineto(57.9397325,0.16406096)
\lineto(53.2522325,0.16406096)
\lineto(53.2522325,0.38085784)
\curveto(54.63113875,1.63867034)(55.60184187,2.66601409)(56.16434187,3.46288909)
\curveto(56.72684187,4.25976409)(57.00809187,4.98827971)(57.00809187,5.64843596)
\curveto(57.00809187,6.15234221)(56.853795,6.56640471)(56.54520125,6.89062346)
\curveto(56.2366075,7.21484221)(55.86746687,7.37695159)(55.43777937,7.37695159)
\curveto(55.04715437,7.37695159)(54.69559187,7.26171721)(54.38309187,7.03124846)
\curveto(54.07449812,6.80468596)(53.8459825,6.47070159)(53.697545,6.02929534)
\lineto(53.48074812,6.02929534)
\curveto(53.57840437,6.75195159)(53.82840437,7.30663909)(54.23074812,7.69335784)
\curveto(54.63699812,8.08007659)(55.1428575,8.27343596)(55.74832625,8.27343596)
\curveto(56.3928575,8.27343596)(56.92996687,8.06640471)(57.35965437,7.65234221)
\curveto(57.79324812,7.23827971)(58.010045,6.74999846)(58.010045,6.18749846)
\curveto(58.010045,5.78515471)(57.916295,5.38281096)(57.728795,4.98046721)
\curveto(57.4397325,4.34765471)(56.9709825,3.67773284)(56.322545,2.97070159)
\curveto(55.34988875,1.90820159)(54.74246687,1.26757659)(54.50027937,1.04882659)
\lineto(56.57449812,1.04882659)
\curveto(56.99637312,1.04882659)(57.291295,1.06445159)(57.45926375,1.09570159)
\curveto(57.63113875,1.12695159)(57.78543562,1.18945159)(57.92215437,1.28320159)
\curveto(58.05887312,1.38085784)(58.17801375,1.51757659)(58.27957625,1.69335784)
\lineto(58.49637312,1.69335784)
\closepath
}
}
{
\newrgbcolor{curcolor}{0 0 0}
\pscustom[linestyle=none,fillstyle=solid,fillcolor=curcolor]
{
\newpath
\moveto(80.603795,6.59765471)
\curveto(80.8303575,7.13281096)(81.11551375,7.54492034)(81.45926375,7.83398284)
\curveto(81.80692,8.12695159)(82.23856062,8.27343596)(82.75418562,8.27343596)
\curveto(83.39090437,8.27343596)(83.87918562,8.06640471)(84.21902937,7.65234221)
\curveto(84.47684187,7.34374846)(84.60574812,7.01367034)(84.60574812,6.66210784)
\curveto(84.60574812,6.08398284)(84.24246687,5.48632659)(83.51590437,4.86913909)
\curveto(84.00418562,4.67773284)(84.37332625,4.40429534)(84.62332625,4.04882659)
\curveto(84.87332625,3.69335784)(84.99832625,3.27538909)(84.99832625,2.79492034)
\curveto(84.99832625,2.10742034)(84.77957625,1.51171721)(84.34207625,1.00781096)
\curveto(83.77176375,0.35156096)(82.94559187,0.02343596)(81.86356062,0.02343596)
\curveto(81.32840437,0.02343596)(80.96317,0.08984221)(80.7678575,0.22265471)
\curveto(80.57645125,0.35546721)(80.48074812,0.49804534)(80.48074812,0.65038909)
\curveto(80.48074812,0.76367034)(80.52567,0.86327971)(80.61551375,0.94921721)
\curveto(80.70926375,1.03515471)(80.82059187,1.07812346)(80.94949812,1.07812346)
\curveto(81.04715437,1.07812346)(81.14676375,1.06249846)(81.24832625,1.03124846)
\curveto(81.3147325,1.01171721)(81.46512312,0.93945159)(81.69949812,0.81445159)
\curveto(81.93387312,0.69335784)(82.0959825,0.62109221)(82.18582625,0.59765471)
\curveto(82.3303575,0.55468596)(82.48465437,0.53320159)(82.64871687,0.53320159)
\curveto(83.04715437,0.53320159)(83.3928575,0.68749846)(83.68582625,0.99609221)
\curveto(83.98270125,1.30468596)(84.13113875,1.66992034)(84.13113875,2.09179534)
\curveto(84.13113875,2.40038909)(84.06277937,2.70117034)(83.92606062,2.99413909)
\curveto(83.82449812,3.21288909)(83.71317,3.37890471)(83.59207625,3.49218596)
\curveto(83.4241075,3.64843596)(83.19363875,3.78906096)(82.90067,3.91406096)
\curveto(82.60770125,4.04296721)(82.30887312,4.10742034)(82.00418562,4.10742034)
\lineto(81.81668562,4.10742034)
\lineto(81.81668562,4.28320159)
\curveto(82.12527937,4.32226409)(82.43387312,4.43359221)(82.74246687,4.61718596)
\curveto(83.05496687,4.80077971)(83.28152937,5.02148284)(83.42215437,5.27929534)
\curveto(83.56277937,5.53710784)(83.63309187,5.82031096)(83.63309187,6.12890471)
\curveto(83.63309187,6.53124846)(83.50613875,6.85546721)(83.2522325,7.10156096)
\curveto(83.0022325,7.35156096)(82.6897325,7.47656096)(82.3147325,7.47656096)
\curveto(81.70926375,7.47656096)(81.20340437,7.15234221)(80.79715437,6.50390471)
\lineto(80.603795,6.59765471)
\closepath
}
}
{
\newrgbcolor{curcolor}{0 0 0}
\pscustom[linestyle=none,fillstyle=solid,fillcolor=curcolor]
{
\newpath
\moveto(112.57840437,3.09374846)
\lineto(112.57840437,2.26171721)
\lineto(111.51199812,2.26171721)
\lineto(111.51199812,0.16406096)
\lineto(110.54520125,0.16406096)
\lineto(110.54520125,2.26171721)
\lineto(107.18192,2.26171721)
\lineto(107.18192,3.01171721)
\lineto(110.86746687,8.27343596)
\lineto(111.51199812,8.27343596)
\lineto(111.51199812,3.09374846)
\lineto(112.57840437,3.09374846)
\closepath
\moveto(110.54520125,3.09374846)
\lineto(110.54520125,7.03710784)
\lineto(107.75613875,3.09374846)
\lineto(110.54520125,3.09374846)
\closepath
}
}
{
\newrgbcolor{curcolor}{0 0 0}
\pscustom[linestyle=none,fillstyle=solid,fillcolor=curcolor]
{
\newpath
\moveto(142.20340437,8.10937346)
\lineto(141.74637312,7.11327971)
\lineto(139.35574812,7.11327971)
\lineto(138.83426375,6.04687346)
\curveto(139.86942,5.89452971)(140.6897325,5.50976409)(141.29520125,4.89257659)
\curveto(141.8147325,4.36132659)(142.07449812,3.73632659)(142.07449812,3.01757659)
\curveto(142.07449812,2.59960784)(141.98856062,2.21288909)(141.81668562,1.85742034)
\curveto(141.64871687,1.50195159)(141.43582625,1.19921721)(141.17801375,0.94921721)
\curveto(140.92020125,0.69921721)(140.63309187,0.49804534)(140.31668562,0.34570159)
\curveto(139.86746687,0.13085784)(139.40652937,0.02343596)(138.93387312,0.02343596)
\curveto(138.45731062,0.02343596)(138.10965437,0.10351409)(137.89090437,0.26367034)
\curveto(137.67606062,0.42773284)(137.56863875,0.60742034)(137.56863875,0.80273284)
\curveto(137.56863875,0.91210784)(137.61356062,1.00781096)(137.70340437,1.08984221)
\curveto(137.79324812,1.17577971)(137.90652937,1.21874846)(138.04324812,1.21874846)
\curveto(138.14481062,1.21874846)(138.23270125,1.20312346)(138.30692,1.17187346)
\curveto(138.385045,1.14062346)(138.51590437,1.06054534)(138.69949812,0.93163909)
\curveto(138.99246687,0.72851409)(139.28934187,0.62695159)(139.59012312,0.62695159)
\curveto(140.04715437,0.62695159)(140.447545,0.79882659)(140.791295,1.14257659)
\curveto(141.13895125,1.49023284)(141.31277937,1.91210784)(141.31277937,2.40820159)
\curveto(141.31277937,2.88867034)(141.1584825,3.33593596)(140.84988875,3.74999846)
\curveto(140.541295,4.16796721)(140.11551375,4.49023284)(139.572545,4.71679534)
\curveto(139.14676375,4.89257659)(138.56668562,4.99413909)(137.83231062,5.02148284)
\lineto(139.35574812,8.10937346)
\lineto(142.20340437,8.10937346)
\closepath
}
}
{
\newrgbcolor{curcolor}{0 0 0}
\pscustom[linestyle=none,fillstyle=solid,fillcolor=curcolor]
{
\newpath
\moveto(169.37332625,8.27343596)
\lineto(169.37332625,8.05663909)
\curveto(168.85770125,8.00585784)(168.43582625,7.90234221)(168.10770125,7.74609221)
\curveto(167.7834825,7.59374846)(167.46121687,7.35937346)(167.14090437,7.04296721)
\curveto(166.82449812,6.72656096)(166.56082625,6.37304534)(166.34988875,5.98242034)
\curveto(166.1428575,5.59570159)(165.96902937,5.13476409)(165.82840437,4.59960784)
\curveto(166.39090437,4.98632659)(166.9553575,5.17968596)(167.52176375,5.17968596)
\curveto(168.0647325,5.17968596)(168.53543562,4.96093596)(168.93387312,4.52343596)
\curveto(169.33231062,4.08593596)(169.53152937,3.52343596)(169.53152937,2.83593596)
\curveto(169.53152937,2.17187346)(169.3303575,1.56640471)(168.92801375,1.01952971)
\curveto(168.44363875,0.35546721)(167.80301375,0.02343596)(167.00613875,0.02343596)
\curveto(166.46317,0.02343596)(166.0022325,0.20312346)(165.62332625,0.56249846)
\curveto(164.88113875,1.26171721)(164.510045,2.16796721)(164.510045,3.28124846)
\curveto(164.510045,3.99218596)(164.65262312,4.66796721)(164.93777937,5.30859221)
\curveto(165.22293562,5.94921721)(165.62918562,6.51757659)(166.15652937,7.01367034)
\curveto(166.68777937,7.50976409)(167.19559187,7.84374846)(167.67996687,8.01562346)
\curveto(168.16434187,8.18749846)(168.61551375,8.27343596)(169.0334825,8.27343596)
\lineto(169.37332625,8.27343596)
\closepath
\moveto(165.728795,4.17187346)
\curveto(165.6584825,3.64452971)(165.62332625,3.21874846)(165.62332625,2.89452971)
\curveto(165.62332625,2.51952971)(165.69168562,2.11132659)(165.82840437,1.66992034)
\curveto(165.96902937,1.23242034)(166.17606062,0.88476409)(166.44949812,0.62695159)
\curveto(166.64871687,0.44335784)(166.89090437,0.35156096)(167.17606062,0.35156096)
\curveto(167.51590437,0.35156096)(167.81863875,0.51171721)(168.08426375,0.83202971)
\curveto(168.353795,1.15234221)(168.48856062,1.60937346)(168.48856062,2.20312346)
\curveto(168.48856062,2.87109221)(168.35574812,3.44921721)(168.09012312,3.93749846)
\curveto(167.82449812,4.42577971)(167.447545,4.66992034)(166.95926375,4.66992034)
\curveto(166.81082625,4.66992034)(166.65067,4.63867034)(166.478795,4.57617034)
\curveto(166.31082625,4.51367034)(166.06082625,4.37890471)(165.728795,4.17187346)
\closepath
}
}
{
\newrgbcolor{curcolor}{0 0 0}
\pscustom[linestyle=none,fillstyle=solid,fillcolor=curcolor]
{
\newpath
\moveto(192.20145125,8.10937346)
\lineto(196.46121687,8.10937346)
\lineto(196.46121687,7.88671721)
\lineto(193.81277937,-0.00000154)
\lineto(193.15652937,-0.00000154)
\lineto(195.52957625,7.15429534)
\lineto(193.34402937,7.15429534)
\curveto(192.90262312,7.15429534)(192.58817,7.10156096)(192.40067,6.99609221)
\curveto(192.072545,6.81640471)(191.80887312,6.53906096)(191.60965437,6.16406096)
\lineto(191.4397325,6.22851409)
\lineto(192.20145125,8.10937346)
\closepath
}
}
{
\newrgbcolor{curcolor}{0 0 0}
\pscustom[linestyle=none,fillstyle=solid,fillcolor=curcolor]
{
\newpath
\moveto(220.29715437,4.16601409)
\curveto(219.66824812,4.68163909)(219.26199812,5.09570159)(219.07840437,5.40820159)
\curveto(218.89871687,5.72070159)(218.80887312,6.04492034)(218.80887312,6.38085784)
\curveto(218.80887312,6.89648284)(219.00809187,7.33984221)(219.40652937,7.71093596)
\curveto(219.80496687,8.08593596)(220.33426375,8.27343596)(220.99442,8.27343596)
\curveto(221.635045,8.27343596)(222.15067,8.09960784)(222.541295,7.75195159)
\curveto(222.93192,7.40429534)(223.1272325,7.00781096)(223.1272325,6.56249846)
\curveto(223.1272325,6.26562346)(223.02176375,5.96288909)(222.81082625,5.65429534)
\curveto(222.59988875,5.34570159)(222.16043562,4.98242034)(221.49246687,4.56445159)
\curveto(222.17996687,4.03320159)(222.635045,3.61523284)(222.85770125,3.31054534)
\curveto(223.15457625,2.91210784)(223.30301375,2.49218596)(223.30301375,2.05077971)
\curveto(223.30301375,1.49218596)(223.09012312,1.01367034)(222.66434187,0.61523284)
\curveto(222.23856062,0.22070159)(221.67996687,0.02343596)(220.98856062,0.02343596)
\curveto(220.23465437,0.02343596)(219.64676375,0.25976409)(219.22488875,0.73242034)
\curveto(218.88895125,1.11132659)(218.7209825,1.52538909)(218.7209825,1.97460784)
\curveto(218.7209825,2.32617034)(218.83817,2.67382659)(219.072545,3.01757659)
\curveto(219.31082625,3.36523284)(219.71902937,3.74804534)(220.29715437,4.16601409)
\closepath
\moveto(221.21707625,4.79296721)
\curveto(221.68582625,5.21484221)(221.98270125,5.54687346)(222.10770125,5.78906096)
\curveto(222.23270125,6.03515471)(222.29520125,6.31249846)(222.29520125,6.62109221)
\curveto(222.29520125,7.03124846)(222.17996687,7.35156096)(221.94949812,7.58202971)
\curveto(221.71902937,7.81640471)(221.40457625,7.93359221)(221.00613875,7.93359221)
\curveto(220.60770125,7.93359221)(220.2834825,7.81835784)(220.0334825,7.58788909)
\curveto(219.7834825,7.35742034)(219.6584825,7.08788909)(219.6584825,6.77929534)
\curveto(219.6584825,6.57617034)(219.70926375,6.37304534)(219.81082625,6.16992034)
\curveto(219.916295,5.96679534)(220.0647325,5.77343596)(220.25613875,5.58984221)
\lineto(221.21707625,4.79296721)
\closepath
\moveto(220.572545,3.94335784)
\curveto(220.24832625,3.66992034)(220.00809187,3.37109221)(219.85184187,3.04687346)
\curveto(219.69559187,2.72656096)(219.61746687,2.37890471)(219.61746687,2.00390471)
\curveto(219.61746687,1.49999846)(219.75418562,1.09570159)(220.02762312,0.79101409)
\curveto(220.30496687,0.49023284)(220.65652937,0.33984221)(221.08231062,0.33984221)
\curveto(221.50418562,0.33984221)(221.84207625,0.45898284)(222.0959825,0.69726409)
\curveto(222.34988875,0.93554534)(222.47684187,1.22460784)(222.47684187,1.56445159)
\curveto(222.47684187,1.84570159)(222.40262312,2.09765471)(222.25418562,2.32031096)
\curveto(221.97684187,2.73437346)(221.416295,3.27538909)(220.572545,3.94335784)
\closepath
}
}
\end{pspicture}